%% file: distCCA_JMLR.tex
\documentclass[twoside,11pt]{article}

\usepackage{blindtext}

\usepackage{jmlr2e}

\input{definitionchen}

\usepackage{lastpage}
\jmlrheading{23}{2022}{1-\pageref{LastPage}}{1/21; Revised 5/22}{9/22}{21-0000}{Author One and Author Two}

\ShortHeadings{Distributed Estimation and Gap-Free  Analysis of Canonical Correlations}{Chen and Zhu}
\firstpageno{1}

\begin{document}

\title{Distributed Estimation and Gap-Free  Analysis of Canonical Correlations}

\author{\name Canyi Chen \email canyic@umich.edu \\
       \addr Department of Biostatistics\\
       University of Michigan\\
       Ann Arbor, MI 48109, USA
       \AND
       \name Liping
  Zhu \email zhu.liping@ruc.edu.cn \\
       \addr        Center for Applied Statistics\\
       Institute of Statistics and Big Data\\
       Renmin University of China\\
       Beijing 100872, China}

\editor{My editor}

\maketitle
\begin{abstract}
Massive data analysis calls for distributed algorithms and theories.  We design a multi-round distributed algorithm for canonical correlation analysis. We construct principal directions through convex formulation of canonical correlation analysis, and use the shift-and-invert preconditioning iteration to expedite the convergence rate. This distributed algorithm is communication-efficient.  The resultant estimate achieves the same convergence rate as if all observations were pooled together, but does not impose stringent restrictions on the number of machines. We take a gap-free analysis to bypass the widely used yet unrealistic assumption of an explicit gap between the successive canonical correlations in the canonical correlation analysis. 
Extensive simulations and applications to three benchmark image data are conducted to demonstrate the empirical performance of our proposed algorithms and theories. 
\end{abstract}


\section{Introduction}

Rapid advances in information technology allow people to collect data of unprecedented size in various scientific areas such as genomics, neuroscience and physics. These datasets are quite often simultaneously characterized by high dimensions and large sample sizes.  Analyzing massive datasets in a single machine through conventional in-memory statistical algorithms is perhaps prohibitive due to constraints such as privacy concerns, limited memory and storage space  \citep{fan2014ChallengesBigData}.  Easing these constraints calls for distributed algorithms and theories with low communication costs. 

We are concerned with canonical correlation analysis  \citep{hotelling1936RelationsTwoSets} that has already grown into a powerful arsenal for revealing mutual variability in multiple views of two sets of random vectors.   It has seen many successful applications in regression, clustering and word embedding \citep{dhillon2011MultiviewLearningWord}, among many others.  In the past two decades many efforts have been made to establish theoretical properties of  canonical correlation analysis in the regimes of finite- \citep{anderson1999AsymptoticTheoryCanonical} and  infinite-dimensions \citep{hardoon2004CanonicalCorrelationAnalysis,bao2014CanonicalCorrelationCoefficients, gao2015MinimaxEstimationSparse,gao2017SparseCCAAdaptive, ma2020SubspacePerspectiveCanonical}. Existing works implicitly assume that the observations are stored and processed in a single machine.

Canonical correlation analysis of massive data calls for distributed algorithms and theories. Towards this goal, many efforts have been made in a star-networked distributed setting \citep{zaharia2016ApacheSparkUnified}.  These algorithms can be classified into two classes: single- or multiple-round.  In particular, in the first  class, \cite{lv2020OneshotDistributedAlgorithm} developed a  one-shot divide-and-conquer algorithm for canonical correlation analysis, which proceeds as follows.  Each local machine conducts canonical correlation analysis and transmits the local results to the central machine over which all local estimates are aggregated to produce a final solution.  This algorithm indeed generalizes the procedure proposed by  \cite{fan2019DistributedEstimationPrincipal} to recover the eigenspace of principal component analysis. This generalization, however, suffers from at least two essential issues.  First, the local estimates of canonical correlations are no longer unbiased. Therefore, for the divide-and-conquer algorithm to achieve an optimal convergence rate, a few undesirable and even unrealistic restrictions on the number of local machines have to be imposed to eliminate the biases of the local estimates \citep{fan2019DistributedEstimationPrincipal,chen2021DistributedEstimationPrincipal}.  In a sensor network, the number of local machines usually violates the restrictions. Second, for the canonical correlation analysis to fulfill the orthogonal constraints, the divide-and-conquer algorithm requires each local machine to transmit the full covariance matrices to the central one in order to perform matrix decomposition. This additional procedure introduces an undesirable and even unbearable communication cost, particularly when the random vectors are high or even ultrahigh dimensional.  Many  multiple-round  distributed algorithms,  which fall into the second class, are also designed to relax the restrictions on the number of machines that are usually required by  distributed principal component analysis.	 To be precise,  \cite{garber2017CommunicationefficientAlgorithmsDistributeda}  presented two algorithms to estimate the leading principal component, which were later extended to a general number of principal components by \cite{chen2021DistributedEstimationPrincipal}. 

The  above methodologies designed for distributed principal component analysis, however, cannot be adapted directly  to  distributed canonical correlation analysis because the latter involves inversions of two large covariance matrices. Using the pooled inversion would result in significant communication costs, while using the local one would result in non-negligible estimation biases. In addition, calculating  inversions of   large covariance matrices usually require very  expensive computation costs.

In the present article, we design a  distributed algorithm for canonical correlation analysis that greatly mitigates  the restrictions on the number of local machines and simultaneously eases the communication cost by avoiding communicating the full covariance matrices. The resultant distributed estimate attains an optimal convergence rate after a finite number of communications and fully respects the orthogonal constraints in the canonical correlation analysis. To be precise, we construct the principal directions through convex formulation of canonical correlation analysis, and use the shift-and-invert preconditioning  \citep{golub1983MatrixComputations}, together with quadratic programming, to expedite the convergence rate of our proposed distributed algorithm.  Solving canonical correlation analysis through shift-and-invert preconditioning dates back to \cite{golub1983MatrixComputations} and \cite{golub1995CanonicalCorrelationsMatrix}. It is an iterative method that transforms the canonical correlation analysis into a series of least squares problems in order to produce a sequence of increasingly accurate estimates. 
\cite{ma2015FindingLinearStructure}, \cite{wang2016EfficientGloballyConvergent}, \cite{allen-zhu2017DoublyAcceleratedMethods} and \cite{gao2019StochasticCanonicalCorrelation}  use this concept to solve the pooled canonical component analysis, where the least squares problems are solved by stochastic techniques. However, when this concept is implemented in the distributed canonical component analysis,  numerous iterations may be involved and consequently,  unbearable communication costs are to be introduced.  We shall show that,  the distributed estimate of our proposed algorithm achieves the same convergence rate as if all observations are pooled together  after a finite number of communications.  This merely requires to transmit vectors and thus reduces the communication cost substantially. By contrast, the divide-and-conquer method proposed by  \cite{lv2020OneshotDistributedAlgorithm} requires to transmit large dimensional matrices, which results in non-negligible  communication cost.

We further make an important theoretical contribution to the literature by deriving a gap-free error bound for the distributed canonical correlation analysis.  Most existing convergence analyses depend upon the assumption of an explicit canonical correlation gap, $\rho_L - \rho_{L + 1}> B_L$, where $\rho_\ell$ is the $\ell$-th largest canonical correlation, $L$ is a user-specified integer and $B_L$ is a theoretical lower bound. Examples of requiring this assumption include, but not limited to,   \cite{wang2016EfficientGloballyConvergent} and \cite{gao2019StochasticCanonicalCorrelation}. This assumption is widely used though, it is likely unrealistic. We demonstrate this phenomenon through three benchmark image data sets,  MEDIAMILL \citep{snoek2006ChallengeProblemAutomated}, MNIST \citep{lecun1998GradientbasedLearningApplied}, and MFEAT \citep{Dua:2019}.  The Wilks' lambda test  \citep{anderson2003IntroductionMultivariateStatistical} indicates that the top $50$ canonical correlations are significant with p-values all less than $0.01$. In \Cref{fig:gap_plot},  we plot the histogram of the first-order differences, $\rho_{\ell + 1} - \rho_{\ell}$, of the top $50$ canonical correlations. We also mark the theoretical lower bound of $B_L$   derived by \cite{cheng2021TacklingSmallEigengaps} with a small diamond on the horizontal axis.  It can be clearly seen that, all the first-order differences of the top $50$ canonical correlations are substantially smaller than the lower bounds, violating the theoretical requirement apparently.  We take a gap-free analysis to bypass this widely used and yet unrealistic requirement in the distributed canonical correlation analysis.

\graphicspath{{figs/}}
\begin{figure}[!htbp]
	\centerline{
		\begin{tabular}{ccc}
		\psfig{figure=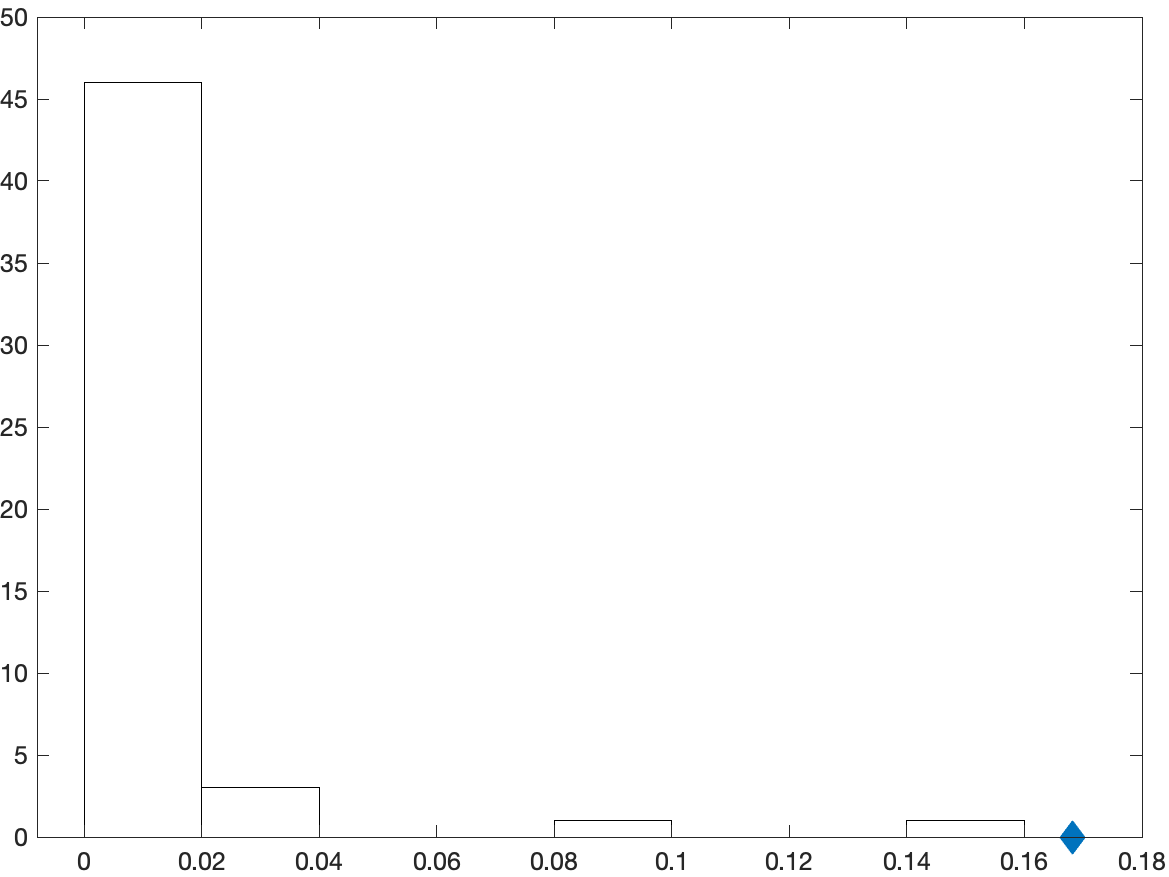,width=1.8in,angle=0} & \psfig{figure=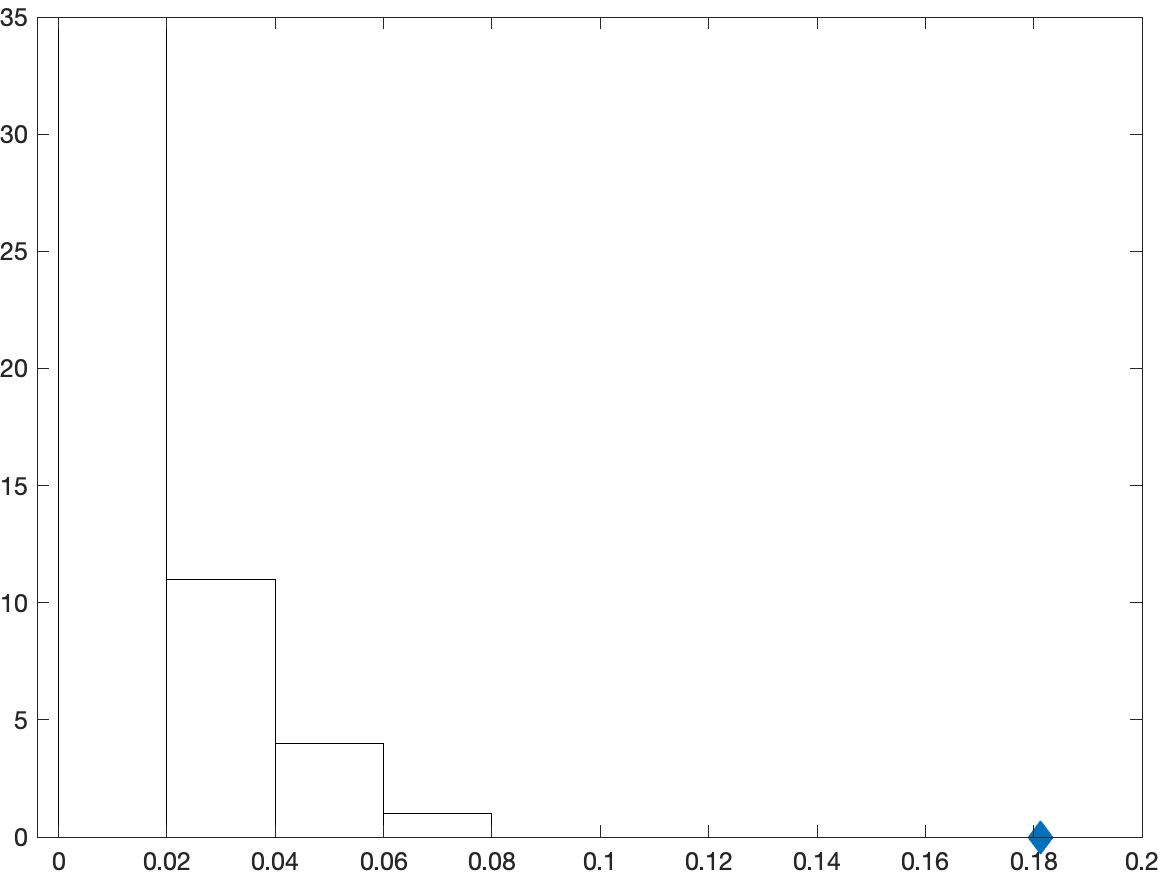,width=1.8in,angle=0} & \psfig{figure=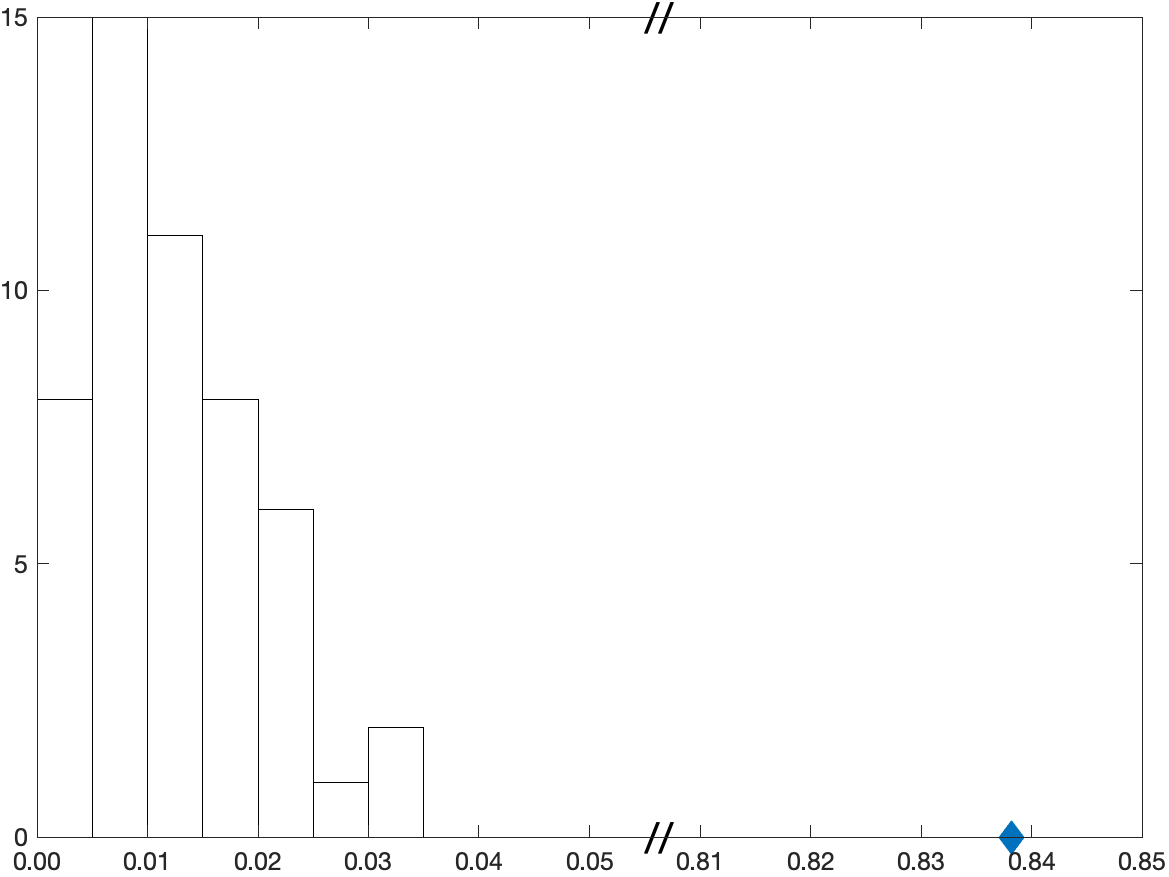,width=1.8in,angle=0}\\
			(A): MMILL   & (B): MNIST  & MFEAT
		\end{tabular}
	}
	  \captionsetup{font=footnotesize}
	\caption{The histograms of the first-order differences of the top $50$  canonical correlations for the MMILL, MNIST, and MFEAT datasets. The diamonds marked in the horizontal axes stand for the theoretical lower bounds required by \cite{cheng2021TacklingSmallEigengaps}. }
	\label{fig:gap_plot}
\end{figure}

The rest of this paper is organized as follows.  We develop a   distributed algorithm for the canonical correlation analysis in  \Cref{section:methodology} and derive an asymptotic error bound for the resultant distributed estimate in \Cref{section:theory}. Extensive simulations and applications to three benchmark image data sets are conducted in \Cref{section:numerical_study} and \Cref{section:applicatioin}   to augment our theoretical findings. We conclude this paper with brief discussions in \Cref{section:conclusion}. 
All technical proofs are relegated to the appendices.

The following notations will be used repetitively in subsequent exposition. 
We use $C, C_0, C_1, \ldots, c, c_0, c_1,\ldots$ to denote generic constants which may vary at each appearance.    A random vector $\x$ is said to be sub-Gaussian \citep{wainwright2019HighDimensionalStatisticsNonAsymptotic} if it satisfies  $\sup_{\norm{\scriptsize\ba} = 1} E \Big[ \exp\{t(\ba\trans\x)^2\} \Big] \leq C$ for some $t>0$ and $C>0$.
For  a matrix  $\A = (a_{kl})\in\mR^{p\times q}$,  we define $\norm{\A} $ and $\norm{\A}_F$ to be the respective spectral and Frobenius norms of $\A$.  In particular, if $\A$ is a vector, $\norm{\A} = \norm{\A}_F$.    We use $\sigma_k(\A)$ to denote the $k$-th largest singular value of a matrix $\A$, and use the more obvious notations $\sigma_{\max}(\A)$ and $\sigma_{\min}(\A)$ to denote the largest and smallest singular values of $\A$ respectively.  By definition, $\norm{\A}  = \sigma_{\max}(\A) = \sigma_1(\A)$.  Let $\W$ be a generic Hermitian semi-positive definite matrix.  For matrices $\A_1$ and $\A_2$, we  define  $\inner{\A_1}{\A_2}_\W \defn \A_1\trans \W\A_2$, which differs from  the usual inner product  unless both $\A_1$ and $\A_2$ are vectors. In particular,  $\inner{\A_1}{\A_2} \defn \A_1\trans \A_2$.  Let ${\bf0}_{p\times q}$ be a matrix with all entries being zero and $\I_{p\times p}$ be an identity matrix throughout.

\section{{Methodology Development}\label{section:methodology}}

\subsection{A brief review}
Canonical correlation analysis aims to identify the co-variations between two sets of random vectors.  To be precise, let $\x\in\mR^{d_\x}$ and $\y\in\mR^{d_\y}$ have zero  mean.  We denote the joint distribution function of $(\x, \y)$ by $\calD(\x, \y)$, which is  sub-Gaussian.  Define $\bSigx\defn \cov(\x,\x\trans)$, $\bSigy\defn \cov( \y, \y\trans)$ and $\bSigxy\defn \cov(\x,\y\trans)$ to be the  auto- and cross-covariance matrices. Canonical correlation analysis is invariant to linear transformations of the two random vectors   \citep{gao2019StochasticCanonicalCorrelation}, which allows us to   assume without loss of generality that $\max\{\sigma_{\max}(\bSigx), \sigma_{\max}(\bSigy) \} \leq 1$ and $\gamma \defn\min\{\sigma_{\min}(\bSigx), \sigma_{\min}(\bSigy)\}>0$.   Let $L$ be a user-specified integer such that $1\le L\le \min(d_\x, d_\y)$. Define ${\cal O}_\U  \defn  \{\U\in\mR^{d_\x\times L}:  \inner{\U}{\U}_{\bSigx} =  \I_{L\times L} \}$ and ${\cal O}_\V  \defn  \{\V\in\mR^{d_\y\times L}:  \inner{\V}{\V}_{\bSigy} =  \I_{L\times L} \}$.  Canonical correlation analysis solves the non-convex optimization of
\beqr\label{CCA_population}
 (\U_{1:L}^{\ast}, \V_{1:L}^{\ast}) \ \defn\  \argmax_{
 \U\in{\cal O}_\U, \V\in{\cal O}_\V} 
\tr\{ \inner{\U }{\V}_{\bSigxy}\},
\eeqr
where $\tr(\cdot)$ stands for  the trace operator, $\U_{1:L}^{\ast} = (\bu_1^{\ast}, \ldots, \bu_L^{\ast})$ and $\V_{1:L}^{\ast} = (\bv_1^{\ast}, \ldots, \bv_L^{\ast})$,  $(\bu_\ell^{\ast}, \bv_\ell^{\ast})$   stand for the $\ell$-th pair of canonical directions and $\rho_\ell\defn  \inner{\bu_\ell}{\bv_\ell}_{\bSigxy}$ is $\ell$-th canonical correlation satisfying $1\geq \rho_1\geq\ldots\geq \rho_L\geq 0$.  Let  $\bT \defn \bSig_{\x}^{-1/2}\bSig_{\x, \y}\bSig_{\y}^{-1/2}$.  The $L$-truncated singular value decomposition of  $\bT$ is $\bPhi_{1:L}^\ast\D_{1:L}^\ast\bPsi\strans_{1:L}$. \cite{hotelling1936RelationsTwoSets} showed that  $(\U_{1:L}^{\ast}, \V_{1:L}^{\ast}) = (\bSig_{\x}^{-1/2}\bPhi_{1:L}^\ast, \bSig_{\y}^{-1/2}\bPsi_{1:L}^\ast)$, and $\rho_1, \ldots, \rho_L$ are the singular values of $\bT$, which correspond to the respective diagonal elements of $\D_{1:L}^\ast$.

Suppose a random sample of size $N$, denoted by $\{(\x_i, \y_i), i = 1,\ldots, N\}$,  is available. Let  $\wh\bSig_{\x}$, $\wh\bSig_{\y}$ and $ \wh\bSig_{\x, \y}$ be the respective estimates of $\bSig_{\x}$, $\bSig_{\y}$ and $\bSig_{\x, \y}$. Define $\wh{\cal O}_\U  \defn  \{\U\in\mR^{d_\x\times L}:  \inner{\U}{\U}_{{\wh\bSig}_\x} =  \I_{L\times L} \}$ and $\wh {\cal O}_\V  \defn  \{\V\in\mR^{d_\y\times L}:  \inner{\V}{\V}_{{\wh\bSig}_\y} =  \I_{L\times L} \}$.  At the sample level,  canonical correlation analysis solves the non-convex optimization of
\beqr
\label{CCA_object}
(\wh\U_{1:L}^{\pool}, \wh\V_{1:L}^{\pool}) \ \defn\  \argmax_{
	\U\in\wh{\cal O}_\U, \V\in\wh{\cal O}_\V} 
\tr\{ \inner{\U }{\V}_{\wh{\bSig}_{\x,\y}}\},
\eeqr
where $\wh\U_{1:L}^{\pool} = (\wh\bu_1^{\pool}, \ldots, \wh\bu_L^{\pool})$ and $\wh\V_{1:L}^{\pool} = (\wh\bv_1^{\pool}, \ldots, \wh\bv_L^{\pool})$. The pooled estimate of the $\ell$-th canonical correlation $\rho_\ell$ is given by $\wh\rho_\ell\defn  \inner{\wh\bu_\ell^{\pool}}{\wh\bv_\ell^{\pool}}_{\wh\bSig_{\x, \y}}$.  We define the empirical version of $\bT$ by $\wh\bT \defn \wh\bSig_{\x}^{-1/2}\wh\bSig_{\x, \y}\wh\bSig_{\y}^{-1/2}$. The $L$-truncated singular value decomposition of $\wh\bT$ is  $\wh\bPhi_{1:L}\wh\D_{1:L}\wh\bPsi\trans_{1:L}$, where $\wh\bPhi_{1:L} = (\wh\bphi_1, \ldots, \wh\bphi_L)$ and $\wh\bPsi_{1:L} = (\wh\bpsi_1, \ldots, \wh\bpsi_L)$. It follows that $(\wh\U_{1:L}^{\pool}, \wh\V_{1:L}^{\pool}) = (\wh\bSig_{\x}^{-1/2}\wh\bPhi_{1:L}, \wh\bSig_{\y}^{-1/2}\wh\bPsi_{1:L})$ and $\wh\rho_1, \ldots, \wh\rho_L$ are the corresponding singular values of $\wh\bT$, which correspond to the respective diagonal elements of $\wh\D_{1:L}$. It is important to remark here that, directly evaluating $(\wh\U_{1:L}^{\pool}, \wh\V_{1:L}^{\pool})$ through its explicit form  involves singular value decomposition, which is computationally expensive and even prohibitive for massive data analysis.

To ease the computational complexity of singular value decomposition, \cite{garber2015FastSimplePCA} and  \cite{wang2016EfficientGloballyConvergent} developed a variant  shift-and-invert preconditioning iteration method to recast the non-convex optimization problem  \eqref{CCA_object} as a sequence of least squares problems, each of which corresponds to a  linear system and can be solved efficiently. To illustrate the rationale of this iteration method, we first consider approximating the top pair of canonical directions $(\bu_1^\ast, \bv_1^\ast)$. 
Let $\overline{\rho}_1$ be a crude estimate of an upper bound of the top canonical correlation $\rho_1$. Let $d \defn d_\x + d_\y$.
Define 
\[
\wh\bC \ \defn\ 
\begin{pmatrix}
	{\bf 0}_{d_\x\times d_\x} & \wh\bT_{d_\x\times d_\y} \\
	\wh\bT\trans_{d_\x\times d_\y} & {\bf 0}_{d_\y\times d_\y} 
\end{pmatrix}\in \mR^{d\times d}.
\]
It is straightforward to verify that the $L$ largest eigenvalues of $(\overline{\rho}_1 \I - \wh\bC)^{-1}$ are $1/(\overline{\rho}_1 - \wh\rho_1), \cdots  1/(\overline{\rho} _1- \wh\rho_L)$, and the  $L$ smallest  eigenvalues are $ 1/(\overline{\rho}_1 + \wh\rho_L), \cdots, 1/(\overline{\rho} _1+ \wh\rho_1)$. These eigenvalues   correspond to the eigenvectors 
\[
\begin{pmatrix}
	\wh\bphi_1\\
	\wh\bpsi_1
\end{pmatrix}\Bigg/2^{1/2}, \cdots, 
\begin{pmatrix}
	\wh\bphi_L\\
	\wh\bpsi_L
\end{pmatrix}\Bigg/2^{1/2},  \cdots, 
\begin{pmatrix}
	\wh\bphi_L\\
	-\wh\bpsi_L
\end{pmatrix}\Bigg/2^{1/2}, \cdots, 
\begin{pmatrix}
	\wh\bphi_1\\
	-\wh\bpsi_1
\end{pmatrix}\Bigg/2^{1/2}.
\]
To approximate the top pair of canonical directions $(\bu_1^\ast, \bv_1^\ast)$, we run power iterations on $(\overline{\rho}_1 \I - \wh\bC)^{-1}$ through the shift-and-invert preconditioning iteration method  \citep{golub1983MatrixComputations}.  To be precise, we let $(\wh\bu^{(0)}, \wh\bv^{(0)})$ be  an initial estimate of the top pair of canonical directions $(\bu_1^\ast, \bv_1^\ast)$. During the $t$-th iteration,   $t = 0, 1, \ldots,$ the variant shift-and-invert preconditioning iteration method takes the form of
\beqrs
\begin{pmatrix}
	\wt\bu^{(t + 1)}\\
	 \wt\bv^{(t + 1)}
\end{pmatrix}
 = 
 \begin{pmatrix}
 	\bhSigx^{1/2} & {\bf 0}\\
 	{\bf 0} & \bhSigy^{1/2}
 \end{pmatrix}^{-1}
 (\overline{\rho}_1 \I - \wh\bC)^{-1}
  \begin{pmatrix}
 	\bhSigx^{1/2} & {\bf 0}\\
 	{\bf 0} & \bhSigy^{1/2}
 \end{pmatrix}
\begin{pmatrix}
\wh\bu^{(t)}\\
\wh\bv^{(t)}
\end{pmatrix},
\eeqrs
which, after some straightforward algebraic calculations, reduces to
\beqr\label{eq:iteration}
\begin{pmatrix}
	\wt\bu^{(t + 1)}\\
	 \wt\bv^{(t + 1)}
\end{pmatrix} 
&=& \bhH^{-1}
  \begin{pmatrix}
	\bhSigx & {\bf 0}\\
	{\bf 0} & \bhSigy
\end{pmatrix}
\begin{pmatrix}
	\wh\bu^{(t)}\\
	\wh\bv^{(t)}
\end{pmatrix}, \textrm{ where }
 \bhH   \defn   \begin{pmatrix}
	\overline{\rho}_1\bhSig_{\x} &  -\bhSig_{\x, \y}\\
	-\bhSig\trans_{\x, \y} & \overline{\rho}_1\bhSig_{\y}\\
\end{pmatrix}.
\eeqr
Iterating  \eqref{eq:iteration} can be recast as a one-step Newton's iteration. In particular,
\beqr\label{eq:newton}
\begin{pmatrix}
	\wt\bu^{(t + 1)}\\
	\wt\bv^{(t + 1)}
\end{pmatrix} 
&=& \begin{pmatrix}
	\wh\bu^{(t)}\\
	\wh\bv^{(t)}
\end{pmatrix}
-
\bhH^{-1} \wh\g, \textrm{ where }
\wh\g \defn
\left\{\bhH
-
\begin{pmatrix}
	\bhSigx & {\bf 0}\\
	{\bf 0} & \bhSigy
\end{pmatrix}\right\}
\begin{pmatrix}
	\wh\bu^{(t)}\\
	\wh\bv^{(t)}
\end{pmatrix}.
\eeqr
Iterating \eqref{eq:iteration}  is also equivalent to solving the following  quadratic optimization,
\beqr\label{eq:quadratic_programming}
\begin{pmatrix}
	\wt\bu^{(t + 1)}\\ \wt\bv^{(t + 1)}
\end{pmatrix}   =   \argmin_{\bu,\bv} \left\{Q(\bu, \bv)\defn \frac12  \begin{pmatrix}
	\bu\\
	\bv
\end{pmatrix}\trans 
\bhH
\begin{pmatrix}
	\bu\\
	\bv
\end{pmatrix} - 
\begin{pmatrix}
	\bu\\
	\bv
\end{pmatrix}\trans\begin{pmatrix}
	\bhSigx & {\bf 0}\\
	{\bf 0} & \bhSigy
\end{pmatrix}
\begin{pmatrix}
	\wh\bu^{(t)}\\
	\wh\bv^{(t)}
\end{pmatrix}\right\}.
\eeqr
 Finally,  we   perform an additional normalization, 
\beqrs
\begin{pmatrix}
	\wh\bu^{(t + 1)}\\
	 \wh\bv^{(t + 1)}
\end{pmatrix}  = 2^{1/2}\begin{pmatrix}
	\wt\bu^{(t + 1)}\\
	 \wt\bv^{(t + 1)}
\end{pmatrix} \Bigg / \Bigg(\inner{\wt\bu^{(t + 1)}}{\wt\bu^{(t + 1)}}_{\bhSig_{\x}}+ \inner{\wt\bv^{(t + 1)}}{\wt\bv^{(t + 1)}}_{\bhSig_{\y}}\Bigg)^{1/2},
\eeqrs
to ensure that  $\inner{\wh\bu^{(t + 1)}}{\wh\bu^{(t + 1)}}_{\bhSig_{\x}} +   \inner{\wh\bv^{(t + 1)}}{\wh\bv^{(t + 1)}}_{\bhSig_{\y}} = 2$, which updates  $(\wh\bu^{(t)}, \wh\bv^{(t)})$ to  $(\wh\bu^{(t+1)}, \wh\bv^{(t+1)})$, for $t = 0, 1,\ldots$. We iterate this procedure until convergence.

Performing direct iterations using \eqref{eq:iteration}  requires to invert a high-dimensional Hessian matrix $\bhH$. This would result in a prohibitive  complexity in both communication and computation.  The non-convex canonical correlation analysis in \eqref{CCA_object}  has been formulated as a sequence of least squares problems. Taking this advantage,  we design a computation- and communication-efficient distributed algorithm. 

\subsection{A distributed estimation}
Next we develop a distributed algorithm for canonical correlation analysis when the  observations, $\{(\x_i, \y_i), i = 1,\ldots, N\}$,  are scattered uniformly at  $K$ machines, each of size $n$.  In other words, $N = nK$.  Let $\X_k\in\mR^{n\times d_\x}$ and $\Y_k\in\mR^{n\times d_\y}$ represent the data matrices in the $k$-th machine,   $\wh\bSig_{\x,k}$, $\wh\bSig_{\y,k}$, $ \wh\bSig_{\x, \y,k}$ and $\wh\bT_k \defn \bhSig_{\x, k}^{-1/2}\bhSig_{\x, \y, k}\bhSig_{\y, k}^{-1/2}$ be the respective   estimates of $\bSig_{\x}$, $\bSig_{\y}$, $\bSig_{\x, \y}$ and $\bT$, which are obtained using the data matrices $(\X_k,\Y_k)$ in the $k$-th machine, for $k =1, \ldots, K.$  

Our proposed distributed algorithm is iterative. To start, we  construct an initial estimate of the top pair of canonical directions, $(\wh\bu^{(0)}, \wh\bv^{(0)})$, and a crude estimate of an upper bound, $\overline{\rho}_1$,  of the top canonical correlation. We merely use the observations in the first machine to avoid communication overhead.  Let  $\wh\bphi^{(0)}$ and $ \wh\bpsi^{(0)}$ be the respective left and right singular vectors of  $\wh\bT_1$  that correspond to the largest singular value. We take $(\wh\bu^{(0)}, \wh\bv^{(0)}) = (\bhSig_{\x, 1}^{-1/2}\wh\bphi^{(0)}, \bhSig_{\y, 1}^{-1/2}\wh\bpsi^{(0)})$, and  $\rhob_1 = \sigma_{\max}(\wh\bT_1) + 1.5\ \omega$, where    $\omega\defn c_0\ (d\log^2d/n)^{1/2}$, for  some  $c_0 > 0$.   Alternatively, we can use some stochastic algorithms that completely avoid matrix decomposition to construct a crude upper bound $\rhob_1$. Examples include  \cite{wang2016EfficientGloballyConvergent},  \cite{allen-zhu2017DoublyAcceleratedMethods} and \cite{gao2019StochasticCanonicalCorrelation}.

We  propose to estimate the spaces spanned by 
$\U_{1:L}^\ast$ and $\V_{1:L}^\ast$ in a sequential manner.  To illustrate the concept of our proposed algorithm, we start from   $L = 1$. Define 
\beqr
\bhH_k & = & \begin{pmatrix}
	\overline{\rho}_1\bhSig_{\x, k} & -\bhSig_{\x, \y, k}\\
	-\bhSig\trans_{\x, \y, k} & \overline{\rho}_1\bhSig_{\y, k}\\
\end{pmatrix}. \label{eq:local_hessian}
\eeqr In the one-step Newton's iteration \eqref{eq:newton}, calculating the full Hessian matrix $\bhH$ requires each local machine to transmit the $d\times d$ local Hessian matrices, $\bhH_k$s, to the central machine, which incurs an expensive and even prohibitive communication cost in high dimensions. To address this issue, a  natural idea is to approximate the full Hessian matrix $\bhH$ in  \eqref{eq:newton} with a local one, say, $\bhH_1$  \citep{shamir2014communication,michaeli.jordan2019communication,fan2021CommunicationEfficientAccurateStatistical}.  We further approximate   $\wh\g$ with an average of all  $\wh\g_k$s, which leads  to an approximation of the  following one-step Newton's iteration, 
\beqrs
\begin{pmatrix}
	\wh\bu_1^{(t + 1)}\\
	\wh\bv_1^{(t + 1)}
\end{pmatrix} 
= \begin{pmatrix}
	\wh\bu^{(t)}\\
	\wh\bv^{(t)}
\end{pmatrix}
-
(\bhH_1K)^{-1} \sum_{k = 1}^K \wh\g_k, 
\textrm{ where }
\wh\g_k \defn
\left\{\bhH_k
-
\begin{pmatrix}
	\bhSig_{\x, k} & {\bf 0}\\
	{\bf 0} & \bhSig_{\y, k}
\end{pmatrix}\right\}
\begin{pmatrix}
	\wh\bu^{(t)}\\
	\wh\bv^{(t)}
\end{pmatrix}.
\eeqrs
The updated estimate $(\wh\bu_1^{(t + 1)}, \wh\bv_1^{(t + 1)})$ is  however sub-optimal in contrast to $(\wt\bu^{(t + 1)}, \wt\bv^{(t + 1)})$ in the exact one-step Newton's iteration \eqref{eq:newton}.  In particular, 
$(\wt\bu^{(t + 1)}, \wt\bv^{(t + 1)})$ is exactly the minimizer of \eqref{eq:quadratic_programming}, by contrast, $(\wh\bu_1^{(t + 1)}, \wh\bv_1^{(t + 1)})$ is not. It is important to remark here that, by  \Cref{theorem:top_canonical_vector} in \Cref{section:theory}, such one-step approximation in the inner loop  may lead to a sub-optimal convergence rate. For each given $(t+1)$, to further refine the estimate $(\wh\bu_j^{(t + 1)}, \wh\bv_j^{(t + 1)})$ for $j=1,2,\ldots$, we suggest to update with
\beqr\label{eq:inner}
\begin{pmatrix}
	\wh\bu_{j+1}^{(t + 1)}\\
	\wh\bv_{j+1}^{(t + 1)}
\end{pmatrix} 
&=& \begin{pmatrix}
	\wh\bu_{j}^{(t+1)}\\
	\wh\bv_{j}^{(t+1)}
\end{pmatrix}
-
(\bhH_1K)^{-1} \sum_{k = 1}^K \wh\g_{k,j}^{(t+1)}, 
\textrm{ where }\\
\label{local_gradient}\wh\g_{k,j}^{(t+1)} &\defn&
\left\{\bhH_k \begin{pmatrix}
	\wh\bu^{(t + 1)}_{j}\\
	\wh\bv^{(t + 1)}_{j}
\end{pmatrix} -  
\begin{pmatrix}
	\bhSig_{\x, k} & {\bf 0}\\
	{\bf 0} & \bhSig_{\y, k}
\end{pmatrix}
\begin{pmatrix}
	\wh\bu^{(t)}\\
	\wh\bv^{(t)}
\end{pmatrix}\right\}
\eeqr
is the local gradient of \eqref{eq:quadratic_programming} evaluated at $(\wh\bu_j^{(t + 1)}, \wh\bv_j^{(t + 1)})$.  For each given $(t+1)$, the inner iteration procedure proceeds in a distributed fashion.  In particular, for each machine, we calculate the local gradients, $\wh\g_{k,j}^{(t+1)}$s,  which are transmitted to the central machine to form a full gradient. Therefore, for each given $(t+1)$, the total communication cost  is  $O(dK)$. This distributed algorithm is summarized in Algorithm \ref{alg:DistributedCCA_top}.

\begin{algorithm}[!htp]
	\caption{\small Distributed Canonical Correlation Analysis.}
	\label{alg:DistributedCCA_top}
	\begin{algorithmic}[1]\footnotesize
		\REQUIRE Auto- and cross-covariances $(\bhSig_{\x,k}, \bhSig_{\y,k}, \bhSig_{\x, \y, k})$ for  machines labeled with $1\leq k\leq K$. The initial estimates {$\wb\rho_1$},   $\wh\bu^{(0)}$ and $\wh\bv^{(0)}$. The number of outer iterations $T$ and the number of inner iterations $T^\prime$. 
		\STATE \textbf{// Phase I: Shift-and-invert preconditioning}
		\STATE Transmit $\wb\rho_1$  to the machines labeled with $2\le k\le K$. For each machine, we calculate  $\bhH_k$s with  \eqref{eq:local_hessian}.
		\FOR{$t = 0,\ldots,(T - 1)$}
		\STATE  Transmit $\wh\bu^{(t)}$ and $\wh\bv^{(t)}$ to each local machine. We set  $\wh\bu^{(t + 1)}_0 = \wh\bu^{(t)}$ and $\wh\bv^{(t + 1)}_0 = \wh\bv^{(t)}$. 
		\FOR{$j = 0, \ldots, (T^\prime - 1)$}
		\FOR{each local machine labeled with $k = 1, \ldots, K$}
		\STATE Compute the local gradient information $\wh\g_{k,j}^{(t+1)}$ by \eqref{local_gradient}
		and transmit the local gradient information $\wh\g_{k,j}^{(t+1)} $ to the central machine.
		\ENDFOR   
		\STATE Perform the approximate Newton's iteration: 
		\beqrs
		\begin{pmatrix}
			\wh\bu^{(t + 1)}_{j + 1}\\
			\wh\bv^{(t + 1)}_{j + 1}
		\end{pmatrix} = 
		\begin{pmatrix}
			\wh\bu^{(t + 1)}_{j}\\
			\wh\bv^{(t + 1)}_{j}
		\end{pmatrix} - (\bhH_1K)^{-1} \sum_{k = 1}^K \wh\g_{k,j}^{(t+1)}.
		\eeqrs
		\ENDFOR
		\STATE The first machine updates 
		\beqrs
		\begin{pmatrix}
			\wh\bu^{(t + 1)}\\
			\wh\bv^{(t + 1)}
		\end{pmatrix} = 2^{1/2}
		\begin{pmatrix}
			\wh\bu^{(t + 1)}_{T^\prime}\\
			\wh\bv^{(t + 1)}_{T^\prime}
		\end{pmatrix}\bigg / \left(\inner{\wh\bu^{(t + 1)}_{T^\prime}}{\wh\bu^{(t + 1)}_{T^\prime}}_{\bhSig_{\x}} + \inner{\wh\bv^{(t + 1)}_{T^\prime}}{\wh\bv^{(t + 1)}_{T^\prime}	}_{\bhSig_{\y}}\right)^{1/2}.
		\eeqrs
		\ENDFOR
		\STATE \textbf{%
			// Phase II: Final normalization}
		\STATE $\wh\bu_1^{\dist} = \wh\bu^{(T)}/(\inner{\wh\bu^{(T)}}{\wh\bu^{(T)}}_{\bhSig_{\x}})^{1/2}$ and $\wh\bv_1^{\dist} = \wh\bv^{(T)}/(\inner{\wh\bv^{(T)}}{\wh\bv^{(T)}}_{\bhSig_{\y}})^{1/2}$
		\ENSURE The final estimate $(\wh\bu_1^{\dist}, \wh\bv_1^{\dist})$ obtained from the first machine.
	\end{algorithmic}
\end{algorithm}


Let us compare the computation cost of direct and indirect methods in calculating the inversions of large dimensional matrices.  The update  \eqref{eq:iteration} requires  direct calculation of large matrix inversions for the full sample.  However, the computational complexity of direct matrix inversion of a Hessian matrix with dimensions $d\times d$ is $O(d^3)$, which can be prohibitive for extremely large $d$. In the present context, $d = d_{\x} + d_\y$. To address this issue, we take the full advantage of the finite-sum  structure of $\bhH_1$, with the diagonal blocks being $\wh\bSig_{\x, 1} = \X_1\trans\X_1/n$ and $\wh\bSig_{\y,1} = \Y_1\trans\Y_1/n$, and the off-diagonal blocks being $\wh\bSig_{\x, \y, 1} = \X_1\trans\Y_1/n$.  We propose to solve a linear system in the inner loop \eqref{eq:inner}, which enabled us to achieve computational efficiency in  our proposed distributed algorithm. Numerous algorithms exist in the optimization literature that can efficiently solve linear systems with a finite-sum structure. See, for example, \cite{hestenes1952MethodsConjugateGradients} and \cite{johnson2013AcceleratingStochasticGradient}. In  \eqref{eq:inner}, all computation involves only matrix-vector products, without the need of explicitly constructing $\wh\bH_1$ or calculating $\wh\bH_1^{-1}$. As a result, our proposal can significantly reduce the computation cost.

%

Next, we extend the framework of \cite{allen-zhu2017DoublyAcceleratedMethods}  to the top-L-dim canonical correlation analysis. Given the top $\ell$ pairs of canonical directions $(\wh\U_{1:\ell}^\dist, \wh\V_{1:\ell}^\dist)$   in \Cref{alg:DistributedCCA_topL}, we define $\bQ_{\wh\U_{1:\ell}^\dist}(\bhSig_{\x})\defn \I_{d_\x\times d_\x} - \wh\U_{1:\ell}^\dist
 \big\{(\wh\U_{1:\ell}^\dist)\trans\bhSigx \wh\U_{1:\ell}^\dist\big\}^{-1}$ $
 (\wh\U_{1:\ell}^\dist)\trans\bhSigx =  \I_{d_\x\times d_\x} - \wh\U_{1:\ell}^\dist
 (\wh\U_{1:\ell}^\dist)\trans\bhSig$ be the projection matrix onto the  orthogonal complement of the column subspace spanned by $\wh\U_{1:\ell}^\dist$ under the $\bhSigx$-inner product. Similarly, we define $\bQ_{\wh\V_{1:\ell}^\dist}(\bhSig_{\y})\defn \I_{d_\y\times d_\y} - \wh\V_{1:\ell}^\dist
  \big\{(\wh\V_{1:\ell}^\dist)\trans\bhSigy \wh\V_{1:\ell}^\dist\big\}^{-1}
 (\wh\V_{1:\ell}^\dist)\trans\bhSigy = \I_{d_\y\times d_\y} - \wh\V_{1:\ell}^\dist
 (\wh\V_{1:\ell}^\dist)\trans\bhSigy$. We remark here that these two projection matrices, $\bQ_{\wh\U_{1:\ell}^\dist}(\bhSig_{\x})$ and $\bQ_{\wh\V_{1:\ell}^\dist}(\bhSig_{\y})$, can be calculated  in a distributed manner to  enhance computational efficiency without increasing  communication cost. We further define the cross-covariance $\bhSig_{\x, \y, k}$ under both the $\bhSigx$- and the $\bhSigy$-inner product by $\bhSig_{\x, \y, k}^{(\ell)} = \bQ\trans_{\wh\U_{1:\ell}^\dist}(\bhSig_{\x})\bhSig_{\x, \y, k}\bQ_{\wh\V_{1:\ell}^\dist}(\bhSig_{\y})$, for $k = 1, \ldots, K$.  Indeed the $(\ell+1)$-th  canonical direction can be obtained in the same manner as \Cref{alg:DistributedCCA_top} through 
 $\{(\bhSig_{\x,k}, \bhSig_{\y,k}, \bhSig_{\x, \y, k}^{(\ell)}), k = 1,\ldots, K \}$, which in spirit projects the intermediate estimates onto the subspaces  spanned by $\bQ_{\wh\U_{1:\ell}^\dist}(\bhSig_{\x})$ and $\bQ_{\wh\V_{1:\ell}^\dist}(\bhSig_{\y})$. We repeat this procedure  successively to obtain the top $L$ pairs of  canonical directions $\wh\U_{1:L}^\dist = (\wh\bu_1^\dist, \ldots, \wh\bu_L^\dist)$ and $\wh\V_{1:L}^\dist = (\wh\bv_1^\dist, \ldots, \wh\bv_{L}^\dist)$. This scheme is clear-cut and performs quite well in our numerical studies and convergence analysis.

\begin{algorithm}[!htbp]
	\caption{\small Distributed Top-$L$-dim Canonical Correlation Analysis.}
	\label{alg:DistributedCCA_topL}
	\begin{algorithmic}[1]\footnotesize
		\REQUIRE Data $\{(\x_{i},\y_{i})_{i \in\mathcal{H}_k}\}$ in each machine $1\leq k\leq K$. The required number of canonical components $L$. 
		\STATE Initial $\wh\U_{1:0}^\dist = ()$, $\wh\V_{1:0}^\dist = ()$, $\bhSig_{\x, \y, k}^{(0)} = \bhSig_{\x, \y, k}$.
		\FOR{$\ell = 1, \ldots, L$}
		\STATE  Compute the $\ell$-th initial estimate $\overline{\rho}_\ell$, $\wh\bu_\ell^{(0)}$ and $\wh\bv_\ell^{(0)}$.
		\STATE Call \Cref{alg:DistributedCCA_top} with $\{(\bhSig_{\x,k}, \bhSig_{\y,k}, \bhSig_{\x, \y, k}^{(\ell - 1)})\}_{k = 1}^K$ to obtain $(\wh\bu_\ell^{\dist\prime}, \wh\bv_\ell^{\dist\prime})$ on the central machine.
		\STATE Project $\wh\bu_\ell^{\dist\prime}$ and $\wh\bv_\ell^{\dist\prime}$ to $(\wh\U_{1:(\ell - 1)}^\dist)^{\perp}$ and $(\wh\V_{1:(\ell - 1)}^\dist)^{\perp}$ respectively by  computing 
		\[
		\wh\bu_\ell^{\dist} = \Big\{\bQ_{\wh\U_{1:(\ell - 1)}^\dist}(\bhSig_{\x})\wh\bu_\ell^{\dist\prime}\Big\}\Big/\Big\{\inner{\bQ_{\wh\U_{1:(\ell - 1)}^\dist}(\bhSig_{\x})\wh\bu_\ell^{\dist\prime}}{\bQ_{\wh\U_{1:(\ell - 1)}^\dist}(\bhSig_{\x})\wh\bu_\ell^{\dist\prime}}_{\bhSigx}\Big\}^{1/2}
		\] and  
		\[
		 \wh\bv_\ell^{\dist} = \Big\{\bQ_{\wh\V_{1:(\ell - 1)}^\dist}(\bhSig_{\y})\wh\bv_\ell^{\dist\prime}\Big\}\Big/\Big\{\inner{\bQ_{\wh\V_{1:(\ell - 1)}^\dist}(\bhSig_{\y})\wh\bv_\ell^{\dist\prime}}{\bQ_{\wh\V_{1:(\ell - 1)}^\dist}(\bhSig_{\y})\wh\bv_\ell^{\dist\prime}}_{\bhSigy}\Big\}^{1/2}.
		\]
		\STATE Update $\wh\U_{1:\ell}^\dist = (\wh\U_{1:(\ell -1)}^\dist, \wh\bu_\ell^\dist)$ and $\wh\V_{1:\ell}^\dist = (\wh\V_{1:(\ell -1)}^\dist, \wh\bv_\ell^\dist)$.
		\STATE Transmit $\wh\bu_\ell^\dist$ and $\wh\bv_\ell^\dist$ to each local machine.
		\FOR{each local machine $k = 1, \ldots, K$}
		\STATE Calculate $\bhSig_{\x,k}\wh\bu_\ell^\dist$ and $\bhSig_{\y,k}\wh\bv_\ell^\dist$, and transmit back to the central machine.
		\ENDFOR
		\STATE Transmit $\bhSig_{\x}\wh\bu_\ell^\dist = \sum_{k = 1}^K\bhSig_{\x,k}\wh\bu_\ell^\dist/K$ and $\bhSig_{\y}\wh\bv_\ell^\dist = \sum_{k = 1}^K\bhSig_{\y,k}\wh\bv_\ell^\dist/K$ to each local machine.
		\FOR{each local machine $k = 1, \ldots, K$}
		\STATE Update $\bhSig_{\x, \y, k}^{(\ell)} = \bQ\trans_{\wh\U_{1:\ell}^\dist}(\bhSig_{\x})\bhSig_{\x, \y, k}\bQ_{\wh\V_{1:\ell}^\dist}(\bhSig_{\y}) = \bQ\trans_{\wh\bu_\ell^\dist}(\bhSig_{\x})\bhSig_{\x, \y, k}^{(\ell - 1)}\bQ_{\wh\bv_\ell^\dist}(\bhSig_{\y})$. 
		\ENDFOR
		\ENDFOR
		\ENSURE The final estimate $(\wh\U_{1:L}^\dist, \wh\V_{1:L}^\dist)$ obtained from the first machine.
	\end{algorithmic}
\end{algorithm}

\section{Theoretical Analysis\label{section:theory}}
We are in the position to analyze the theoretical properties for the resultant distributed estimates, and relegate all technical proofs to the appendices.  We study the distributed estimates of the top and the top-L-dim canonical directions separately, which are derived through   \Cref{alg:DistributedCCA_top,alg:DistributedCCA_topL}.  We apply the full singular value decomposition  to obtain that  $\wh\bT = \wh\bPhi\wh\D\wh\bPsi\trans$, where $\wh\bPhi = (\wh\bphi_1, \ldots, \wh\bphi_{d_\x})\in\mR^{d_\x\times d_\x}$ and $\wh\bPsi = (\wh\bpsi_1, \ldots, \wh\bpsi_{d_\y})\in\mR^{d_\y\times d_\y}$ are orthonormal matrices, $\wh\D\in\mR^{d_\x\times d_\y}$ is a diagonal matrix  with its $(\ell,\ell)$-th element $\wh\D_{\ell\ell}$ being  $\wh\rho_\ell$, for $\ell = 1,\ldots, \min(d_\x, d_\y)$,  and $0$  for $\ell = \min(d_\x, d_\y) + 1, \ldots, \max(d_\x, d_\y)$.

\subsection{Distributed  estimation of the top canonical directions}
We denote the distributed estimates of the top pair of canonical directions by $(\wh\bu_1^\dist, \wh\bv_1^\dist)$. It is important to measure its distance from the pooled estimate $(\wh\bu_1^\pool, \wh\bv_1^\pool)$, which serves as a benchmark in the present context. When there is no confusion, we  write
\[
\sum_{1\leq \ell\leq d_\x \colon \wh\rho_\ell\leq (1- \delta)\wh\rho_1}\abs{\inner{\wh\bu_\ell^\pool}{\wh\bu_1^\dist}_{\bhSig_{\x}}}^2 \textrm{ and }
\sum_{1\leq \ell\leq d_\y\colon \wh\rho_\ell\leq (1- \delta)\wh\rho_1}\abs{\inner{\wh\bv_\ell^\pool}{\wh\bv_1^\dist}_{\bhSig_{\y}}}^2
\]
as
 \[
 \sum_{\ell \colon \wh\rho_\ell\leq (1- \delta)\wh\rho_1}\abs{\inner{\wh\bu_\ell^\pool}{\wh\bu_1^\dist}_{\bhSig_{\x}}}^2\textrm{ and }
 \sum_{\ell \colon \wh\rho_\ell\leq (1- \delta)\wh\rho_1}\abs{\inner{\wh\bv_\ell^\pool}{\wh\bv_1^\dist}_{\bhSig_{\y}}}^2,
 \]
respectively, where $\delta$ is reserved to denote a relative gap threshold. We shall show that the distance  between $(\wh\bu_1^\dist, \wh\bv_1^\dist)$  and $(\wh\bu_1^\pool, \wh\bv_1^\pool)$ satisfies
\beqr\label{eq:measure1}
\max\left(\sum_{\ell\colon \wh\rho_\ell\leq (1- \delta)\wh\rho_1}\abs{\inner{\wh\bu_\ell^\pool}{\wh\bu_1^\dist}_{\bhSig_{\x}}}^2, \sum_{\ell\colon \wh\rho_\ell\leq (1- \delta)\wh\rho_1}\abs{\inner{\wh\bv_\ell^\pool}{\wh\bv_1^\dist}_{\bhSig_{\y}}}^2\right) \leq \varepsilon^2/\delta^2,
\eeqr
for some $\varepsilon>0$ and  an arbitrary constant $0<\delta<1$.  The left-hand side of \eqref{eq:measure1} stands for the distance we shall use in the present context which, at the conceptual level, it is quite different from and more general than the  sine distance $\max\big(\sin^2\wh\theta, \sin^2\wt\theta\big)$, where $\wh\theta \ \defn\ \arccos\abs{\inner{\wh\bu_1^\pool}{\wh\bu_1^\dist}_{\bhSig_{\x}}}$ and $\wt\theta \ \defn\ \arccos\abs{\inner{\wh\bv_1^\pool}{\wh\bv_1^\dist}_{\bhSig_{\y}}}$. To see this, we simply consider evaluating the accuracy of $\wh\bu_1^\dist$. Assuming $(\wh\rho_1 - \wh\rho_2) > 0$,  \cite{wang2016EfficientGloballyConvergent} and \cite{gao2019StochasticCanonicalCorrelation} showed that $\sin\wh\theta \leq C\varepsilon \wh\rho_1/(\wh\rho_1 - \wh\rho_2)$,
where $C>0$ is a generic constant.
This upper bound involves a positive gap between the first two singular values of $\wh\bT$, which is undesirable in the present context.  We connect this sine distance with the left-hand side of \eqref{eq:measure1}. To be precise, we set $\delta = (\wh\rho_1 - \wh\rho_2)/\wh\rho_1$ to obtain
\[
\sin^2\wh\theta = 1 - \abs{\inner{\wh\bu_1^\pool}{\wh\bu_1^\dist}_{\bhSig_{\x}}}^2 =  \sum_{\ell\colon \wh\rho_\ell\leq (1- \delta)\wh\rho_1}\abs{\inner{\wh\bu_\ell^\pool}{\wh\bu_1^\dist}_{\bhSig_{\x}}}^2,
\]
which, by \eqref{eq:measure1}, is smaller than or equal to $(\varepsilon^2/\delta^2),$  or equivalently, 
$\sin \wh\theta \le( \varepsilon/\delta) =\varepsilon \wh\rho_1/(\wh\rho_1 - \wh\rho_2)$. By the definition of the sine distance, the first equality holds. The second is a direct consequence of the fact that
\[
\sum_{\ell = 1}^{d_\x}\abs{\inner{\wh\bu_\ell^\pool}{\wh\bu_1^\dist}_{\bhSig_{\x}}}^2  = 
(\wh\bu_1^\dist)\trans\big\{(\wh\U^{\pool})\trans \bhSig_{\x} \wh\U^{\pool}\big\}\wh\bu_1^\dist = (\wh\bu_1^\dist)\trans \wh\bu_1^\dist = 1.
\]
In other words, by choosing an appropriate $\delta$ in \eqref{eq:measure1}, we can obtain the same result as that in \cite{wang2016EfficientGloballyConvergent} and \cite{gao2019StochasticCanonicalCorrelation}, indicating that the distance we shall use in this context is more general than the sine distance. However, we do not require an explicit gap between the successive singular values of $\wh\bT$. We consider instead the enlarged spaces $\{\wh\bu_\ell^\pool\colon \wh\rho_\ell>(1 - \delta)\wh\rho_1\}$ and $\{\wh\bv_\ell^\pool\colon \wh\rho_\ell>(1 - \delta)\wh\rho_1\}$ to bypass this widely used and yet unrealistic requirement.

Let  $(\wh\bu^{(0)}, \bhv^{(0)})$ be the initial estimate and $\wb\rho_1$ be an upper bound of $\wh\rho_1$. Let $T$ and $T^\prime$ be the respective number of outer and inner iterations. Define $\kappa \defn \norm{\wh\bT_1 - \wh\bT}$ and  $\gamma \defn\min\{\sigma_{\min}(\bSigx), \sigma_{\min}(\bSigy)\}$.  By using  matrix concentration inequalities, we can show that $\kappa = O_p\{(d\log^2d/n)^{1/2}\}$. One can refer to Lemma A.1 for details in Appendix A. We first provide convergence guarantee for  \Cref{alg:DistributedCCA_top}, which does not require an explicit gap between the successive singular values of $\wh\bT$.

\begin{theo}
\label{theorem:top_canonical_vector}
Assume there exists a positive $\omega$ such that
$2\kappa\leq \omega\leq \overline{\rho}_1 - \wh\rho_1\leq 2\omega$,
and
\beqr\label{definition:constant}
\max\left(\sum_{\ell\colon \wh\rho_\ell\leq (1- \delta)\wh\rho_1}\abs{\inner{\wh\bu_\ell^\pool}{\wh\bu^{(0)}}_{\bhSig_{\x}}}^2, \sum_{\ell\colon \wh\rho_\ell\leq (1- \delta)\wh\rho_1}\abs{\inner{\wh\bv_\ell^\pool}{\wh\bv^{(0)}}_{\bhSig_{\y}}}^2\right) \leq 3/8.
\eeqr
Then for an arbitrary constant $\delta\in (0,1)$ that is not necessarily  $(\wh\rho_1 - \wh\rho_2)/\wh\rho_1$,  we have, 
\beqrs
   \max\left(\sum_{\ell\colon \wh\rho_\ell\leq (1- \delta)\wh\rho_1} \abs{\inner{\wh\bu_\ell^\pool}{\wh\bu_1^\dist}_{\bhSig_{\x}}}^2,  \sum_{\ell\colon \wh\rho_\ell\leq (1- \delta)\wh\rho_1}\abs{\inner{\wh\bv_\ell^\pool}{\wh\bv_1^\dist}_{\bhSig_{\y}}}^2\right) \\ 
 =   O_p\left\{ \left(\frac{128\omega^2}{\delta^2\wh\rho_1^2}\right)^T   \right.  \left.+ \frac{1}{1 - 128\omega^2/(\delta \wh\rho_1)^2}\left(\frac{64\kappa^2}{\gamma^2\omega^2}\right)^ {T^\prime}\right\}.
\eeqrs
\end{theo}

We set $\wb\rho_1 = \sigma_1(\wh\bT_1) + 1.5\ \omega$ in \Cref{alg:DistributedCCA_top}.  In Lemma A.1 of Appendix A, we show that  the requirement \eqref{definition:constant}  and $\norm{\wh \bT_1 - \wh \bT}\leq 0.5 \ \omega$ hold with an overwhelming probability. Invoking  Weyl's inequality that $\max\limits_k|\sigma_k(\wh \bT_1) - \sigma_k(\wh \bT) |\le \|\wh \bT_1 - \wh \bT\|$, we further obtain that  $\omega\leq \overline{\rho}_1 - \wh\rho_1\leq 2\omega$.   In other words, the requirements of Theorem  \ref{theorem:top_canonical_vector} are fulfilled with an overwhelming probability.

The error bound in  Theorem  \ref{theorem:top_canonical_vector}  can be further simplified. In particular, if  $\omega = (\kappa\delta\wh \rho_1/\gamma)^{1/2}/3$, then $\omega  = O_p(\kappa^{1/2})=O_p( \norm{\wh\bT_1 - \wh\bT}^{1/2})$ because both $\delta$ and $\wh \rho_1$ fall within $(0,1)$ and $\gamma$ is a positive constant. By Lemma A.1 of Appendix A,  $\omega = O_p\big\{\left(d\log^2 d /n\right)^{1/4}\big\}$.  We further assume  $\left(d\log^2 d /n\right)^{1/4}\ll (\delta \wh\rho_1)/16$ to ensure that $1- 128\omega^2/(\delta \wh\rho_1)^2 \ge 1/2$, which simplifies the error bound as follows. 

{\coro If  $T^\prime = T$, $\omega = (\kappa\delta\wh \rho_1/\gamma)^{1/2}/3$,   and $\left(d\log^2 d /n\right)^{1/4}\ll (\delta \wh\rho_1)/16$ , then 
\beqr\label{coro:convergence_rate_top:eq1}
   \max\left(\sum_{\ell\colon \wh\rho_\ell\leq (1- \delta)\wh\rho_1} \abs{\inner{\wh\bu_\ell^\pool}{\wh\bu_1^\dist}_{\bhSig_{\x}}}^2,  \sum_{\ell\colon \wh\rho_\ell\leq (1- \delta)\wh\rho_1}\abs{\inner{\wh\bv_\ell^\pool}{\wh\bv_1^\dist}_{\bhSig_{\y}}}^2\right) \\
   \nonumber = O_p\left\{\left(\frac{576\kappa}{\gamma\delta\wh\rho_1}\right)^{T}\right\}.\
\eeqr 
}
This corollary indicates that, if  $(576\kappa)/(\gamma\delta\wh\rho_1)\ll 1$, \Cref{alg:DistributedCCA_top} converges at  a linear convergence rate, which implicitly requires $\kappa = O_p\{(d\log^2d/n)^{1/2}\}$ needs to be smaller, i.e.,    $\kappa = o_p(\delta\wh\rho_1\gamma)$.  By definition, $\kappa = \norm{\wh \bT_1 - \wh\bT}$. Therefore, if the number of observations in each local machine  becomes relatively larger, $\kappa $ would be smaller, consequently, the condition  $\kappa = o_p(\delta\wh\rho_1\gamma)$ can be met more easily.

Next, we compare the communication cost of \Cref{alg:DistributedCCA_top}  with that of the one-shot divide-and-conquer method. The latter is generally regarded as communication-efficient because its communication cost is as small as $O(dK)$. By contrast, the  communication cost of \Cref{alg:DistributedCCA_top} is  $O(TT^\prime dK)$. Theorem \ref{theorem:top_canonical_vector} indicates that, for the resultant estimates to attain the error rate $\varepsilon$, the number of iterations is required to satisfy $T  = O\{\log(1/\varepsilon)\}$ if $T^\prime= T$, which is a logarithm order of $\varepsilon$. The total number of iterations, $(TT^\prime)$, can thus be regarded as a   small number. \cite{michaeli.jordan2019communication} stated that the algorithms that merely transmit an $O(d)$ vector (instead of $O(d^2)$ matrices) during each iteration are generally regarded as communication-efficient. In this regard, our proposed distributed algorithm is communication-efficient, although its communication cost is slightly higher than that of the divide-and-conquer method.

{\coro\label{coro:variability_top}
If the distributed estimate of \Cref{alg:DistributedCCA_top} satisfies that  
\beqrs
\max\left(\sum_{\ell\colon \wh\rho_\ell\leq (1- \delta)\wh\rho_1} \abs{\inner{\wh\bu_\ell^\pool}{\wh\bu_1^\dist}_{\bhSig_{\x}}}^2,  \sum_{\ell\colon \wh\rho_\ell\leq (1- \delta)\wh\rho_1}\abs{\inner{\wh\bv_\ell^\pool}{\wh\bv_1^\dist}_{\bhSig_{\y}}}^2\right)\leq \varepsilon/2,
\eeqrs
it follows that
$(\wh\bu_1^\dist)\trans\bhSigxy(\wh\bv_1^\dist)  \geq  (1 - \delta)(1-\varepsilon)\wh\rho_1$, where $\wh\rho_1 = (\wh\bu_1^\pool)\trans\bhSigxy(\wh\bv_1^\pool)$.
}

Identifying the top pair of canonical directions is theoretically difficult if the gap of successive canonical correlations is extremely small. To the most extreme, if the gap is zero, the first two pairs of canonical directions are not statistically identifiable. To bypass this gap requirement, we turn to capture the co-variability between two sets of random vectors, which is indeed the ultimate goal of canonical correlation analysis. This corollary  states that  the distributed estimate  $(\wh\bu_1^\dist, \wh\bv_1^\dist)$ captures almost the same amount of co-variability as the pooled estimate $(\wh\bu_1^\pool, \wh\bv_1^\pool)$ up to a $(1 - \delta)(1-\varepsilon)$ multiplicative factor. This phenomenon is referred to as the gap-free bound in the optimization literature  \citep{allen-zhu2016LazySVDEvenFaster,allen-zhu2017DoublyAcceleratedMethods}.  The parameter $\delta$ is  pre-specified to measure the proportion of the co-variability captured by the distributed estimate $(\wh\bu_1^\dist, \wh\bv_1^\dist)$. In particular, if we choose $\delta = \varepsilon$,  $(\wh\bu_1^\dist, \wh\bv_1^\dist)$ captures at least $(1-\varepsilon)^2\ge (1-2\varepsilon)$ of the co-variability captured by the top pair of canonical directions.  {If we choose $\delta = c_0/(\gamma\wh\rho_1)$ such that $(\delta\wh\rho_1\gamma )= c_0$ and $\delta \in (0,1)$,  for \Cref{alg:DistributedCCA_top} to converge  at  a linear convergence rate,  it implicitly requires   $\kappa = o_p(1)$.} This is a mild condition.

Next, we show that the distributed estimates approximate the underlying true canonical directions pretty well. 
 \begin{coro}
	\label{corollary:top_canonical_vector:population}
	Assume that the distributed estimate obtained from \Cref{alg:DistributedCCA_topL} satisfies
\[
\max\left(\sum_{\ell\colon \wh\rho_\ell\leq (1- \delta)\wh\rho_1} \abs{\inner{\wh\bu_\ell^\pool}{\wh\bu_1^\dist}_{\bhSig_{\x}}}^2,  \sum_{\ell\colon \wh\rho_\ell\leq (1- \delta)\wh\rho_1}\abs{\inner{\wh\bv_\ell^\pool}{\wh\bv_1^\dist}_{\bhSig_{\y}}}^2\right)\leq \varepsilon/2,
\]
for some error term $\varepsilon>0$.	 It follows that
\beqrs
   \max\left(\sum_{\ell\colon \rho_\ell\leq (1- 2\delta)\rho_1} \abs{\inner{\bu_\ell^\ast}{\wh\bu_1^\dist}_{\bhSig_{\x}}}^2,  \sum_{\ell\colon \rho_\ell\leq (1- 2\delta)\rho_1}\abs{\inner{\bv_\ell^\ast}{\wh\bv_1^\dist}_{\bhSig_{\y}}}^2\right)\\
    \leq \frac{2\norm{\wh\bT - \bT}^2}{\{(1 - \delta) (\wh\rho_1 - \rho_1) + \delta\rho_1\}^2}  + \varepsilon.
\eeqrs
\end{coro}


\subsection{Distributed estimation of the  top-L-dim canonical directions}
Next, we study the theoretical properties of \Cref{alg:DistributedCCA_topL}, which provides distributed estimates for the top $L$ pairs of canonical directions. Define $L_\delta\defn\argmax\{1\leq \ell\leq \min(d_\x, d_\y) \colon \wh\rho_\ell > (1- \delta)\wh\rho_L\}$,   which can be larger than or equal to $L$.  We further define $\wh\U_{1:L_\delta}^\pool \defn (\wh\bu_{1}^\pool, \ldots \wh\bu_{L_\delta}^\pool)$ and $\wh\V_{1:L_\delta}^\pool \defn (\wh\bv_{1}^\pool, \ldots \wh\bv_{L_\delta}^\pool)$ to be the pooled estimates of canonical directions associated with the largest $L_\delta$   canonical correlations $\wh\rho_1\ge \cdots\ge  \wh\rho_{L_\delta}$. Let $\wh\U_{(L_\delta+1):d_\x}^\pool \defn (\wh\bu_{L_\delta + 1}^\pool, \ldots \wh\bu_{d_\x}^\pool)$ and $\wh\V_{(L_\delta+1):d_\y}^\pool \defn (\wh\bv_{L_\delta + 1}^\pool, \ldots \wh\bv_{d_\y}^\pool)$, which form the orthogonal complements of  $\wh\U_{1:L_\delta}^\pool$ and $\wh\V_{1:L_\delta}^\pool$, respectively.  In addition, we introduce some notations for distributed estimates. To be precise,   $\wh\bT^{(\ell)}\defn \bhSig_{\x}^{-1/2}\bhSig_{\x, \y}^{(\ell)}\bhSig_{\y}^{-1/2}$, where 
\[
\bhSig_{\x, \y}^{(\ell)} \defn \frac1K  \sum_{k = 1}^K\bhSig_{\x, \y, k}^{(\ell)}
\textrm{ and }
\bhSig_{\x, \y, k}^{(\ell)} = \bQ\trans_{\wh\U_{1:\ell}^\dist}(\bhSig_{\x})\bhSig_{\x, \y, k}\bQ_{\wh\V_{1:\ell}^\dist}(\bhSig_{\y}),
\]
for $ \ell = 1, \ldots, L$ and $k = 1, \ldots, K$. 

We first discuss the selection of initial estimates. For  $\omega$ defined in \Cref{theorem:top_canonical_vector} which corresponds to the special case of $\ell = 1$,  we take $\wb\rho_\ell = \sigma_{\max}(\wh\bT_1^{(\ell)}) + 1.5 \ \omega$ where $\wh\bT_1^{(\ell)} \defn \bhSig_{\x,1}^{-1/2}\bhSig_{\x, \y, 1}^{(\ell)}\bhSig_{\y,1}^{-1/2}$ is the whitened cross-covariance. It can be verified  that $\norm{\bQ_{\wh\U_{1:\ell}^\dist}(\bhSig_{\x})}\leq 1$ and $\norm{\bQ_{\wh\V_{1:\ell}^\dist}(\bhSig_{\y})}\leq 1$. By definition,
\[
 \Norm{\wh\bT_1^{(\ell)} - \wh\bT^{(\ell)}}   \leq  \norm{\bhSig_{\x,1}^{-1/2}\bhSig_{\x, \y,1}\bhSig_{\y,1}^{-1/2}   - \bhSig_{\x}^{-1/2}\bhSig_{\x, \y}\bhSig_{\y}^{-1/2}} \leq  0.5\ \omega,
\]
which holds uniformly for $\ell = 1, \ldots, L$.  
This, together with Weyl's inequality, ensures that $\omega\leq \wb\rho_\ell - \sigma_1(\wh\bT^{(\ell)})\leq 2\omega$.  We summarize our main result as follows.

\begin{theo}
	\label{theorem:topL_canonical_vector}
Assume that $2\kappa\leq \omega\leq \overline{\rho}_\ell - \sigma_{\max}(\wh\bT^{(\ell)})\leq 2\omega$ for $\ell = 1, \ldots, L$.
Let  $T$ and $T^\prime$  be the respective number of  outer and inner iterations in \Cref{alg:DistributedCCA_topL}. Then for an arbitrary constant $\delta\in (0,1)$, we have
\beqrs
  \max\left(\Norm{\Inner{\wh\U_{(L_\delta+1):d_\x}^\pool}{\wh\U_{1:L}^\dist}_{\bhSig_{\x}}}^2, \Norm{\Inner{\wh\V_{(L_\delta+1):d_\y}^\pool}{\wh\V_{1:L}^\dist}_{\bhSig_{\y}}}^2\right) \\
  =   O_p\Bigg[\left(\frac{\wh\rho_1L^2}{\wh\rho_L\delta}\right)^2\Bigg\{\left(\frac{128\omega^2}{\delta^2\wh\rho_L^2}\right)^T + \Bigg.\Bigg. \Bigg.\Bigg. \frac{512}{1 - 128\omega^2/(\delta\wh\rho_L)^2}\left(\frac{64\kappa^2}{\gamma^2\omega^2}\right)^{T^\prime}\Bigg\}\Bigg].
\eeqrs

%
\end{theo}

\Cref{theorem:topL_canonical_vector} can be further simplified as follows.
\begin{coro}
	\label{corollary:topL_canonical_vector} If  $T^\prime = T$,   $\omega = (\kappa\delta\wh\rho_L/\gamma)^{1/2}/3$, and $(d\log^2 d /n)^{1/4} \ll \delta\wh\rho_L/ 16 $, then \beqrs\label{coro:convergence_rate_topL:eq1}
\max\left(\Norm{\Inner{\wh\U_{(L_\delta+1):d_\x}^\pool}{\wh\U_{1:L}^\dist}_{\bhSig_{\x}} }^2, \Norm{\Inner{\wh\V_{(L_\delta+1):d_\y}^\pool}{\wh\V_{1:L}^\dist}_{\bhSig_{\y}} }^2\right)  = O_p\Bigg\{\left(\frac{\wh\rho_1L^2}{\wh\rho_L\delta}\right)^2\left(\frac{576\kappa}{\gamma\delta\wh\rho_L}\right)^{T}\Bigg\}.
\eeqrs
\end{coro}
This is an extension of the classical  $\sin\Theta$ theorem  \citep[Theorem 4]{yu2015UsefulVariantDavis}   in the sense that we do not require there be an explicit gap between the successive singular values, which differ from most existing studies such as  \cite{wang2016EfficientGloballyConvergent}, \cite{gao2019StochasticCanonicalCorrelation} and references therein. We  bypass this gap assumption by considering the enlarged spaces $(\U^\ast_{1:L_\delta}, \V^\ast_{1:L_\delta})$.

The following corollary indicates that  the distributed estimate $(\wh\U_{1:L}^\dist, \wh\V_{1:L}^\dist)$ can capture almost the same amount of co-variability as the pooled estimate $(\wh\U_{1:L}^\pool, \wh\V_{1:L}^\pool)$.

\begin{coro}
	\label{corollary:variability_topL}
	Assume that our estimate from \Cref{alg:DistributedCCA_topL} satisfies 
\[\max\left(\Norm{\Inner{\wh\U_{(L_\delta+1):d_\x}^\pool}{\wh\U_{1:L}^\dist}_{\bhSig_{\x}}  }, \Norm{\Inner{\wh\V_{(L_\delta+1):d_\y}^\pool}{\wh\V_{1:L}^\dist}_{\bhSig_{\y}}  }\right)\leq C\delta/(\wh\rho_1/\wh\rho_{L + 1}),
\]
for some $C>0$.	Then we have
\beqrs
    & & \wh\rho_{L + 1}\leq \Norm{ \bQ\trans_{\wh\U_{1:L}^\dist}(\bhSig_{\x})\bhSig_{\x, \y}\bQ_{\wh\V_{1:L}^\dist}(\bhSig_{\y})}\leq \wh\rho_{L + 1}/(1-\delta),\ \text{and}\\
    & & (1-\delta)\wh\rho_\ell\leq (\wh\bu_\ell^\dist)\trans\bhSigxy\wh\bv_\ell^\dist\leq \wh\rho_\ell/(1 - \delta),\ \text{for}\ \ell = 1,\ldots, L.
\eeqrs
\end{coro}
 Next we quantify the distance between our proposed distributed estimate $\big(\wh\U_{1:L}^\dist, \wh\V_{1:L}^\dist\big)$  and the orthogonal complements of the first $L_{2\delta}^\ast$ underlying true canonical directions $\big(\U^\ast_{(L_{2\delta}^\ast + 1):d_\x}, \V^\ast_{(L_{2\delta}^\ast + 1):d_\y}\big)$, where $L_{2\delta}^\ast\defn\argmax\{1\leq \ell\leq \min(d_\x, d_\y) \colon \rho_\ell > (1- 2\delta)\rho_L\}$. 

\begin{coro}
	\label{corollary:topL_canonical_vector:population}
	Assume that the distributed estimate obtained from \Cref{alg:DistributedCCA_topL} satisfies 
\[
\max\left(\Norm{\Inner{ \wh\U_{(L_\delta+1):d_\x}^\pool}{\wh\U_{1:L}^\dist}_{\bhSig_{\x}} }, \Norm{\Inner{\wh\V_{(L_\delta+1):d_\y}^\pool}{\wh\V_{1:L}^\dist}_{\bhSig_{\y}}}\right)\leq \varepsilon,
\textrm{ for some } \varepsilon>0.
\]
It follows that
\beqrs
\max\left(\Norm{\Inner{\U_{(L_{2\delta}^\ast + 1):d_\x}^\ast}{\wh\U_{1:L}^\dist}_{\bhSig_{\x}}} , \Norm{\Inner{\V_{(L_{2\delta}^\ast + 1):d_\y}^\ast}{\wh\V_{1:L}^\dist}_{\bhSig_{\y}} }\right)\leq \frac{\norm{\wh\bT - \bT}}{(1 - \delta) (\wh\rho_L - \rho_L) + \delta\rho_L}  + \varepsilon.
\eeqrs
\end{coro}

Though we assume the observations are evenly scattered at different machines, we remark here that this assumption is merely assumed to ease the illustration of our distributed algorithms, and all theoretical properties derived in \Cref{section:theory} merely require the number of observations in the first or central machine is sufficiently large. 

\section{Simulation Studies\label{section:numerical_study}}
We use synthetic examples to demonstrate the performance of our distributed algorithms. We generate the observations in the same way as   \cite{cai2018RateoptimalPerturbationBounds}.  To be precise, we generate $\{(\x_i , \y_i ), i = 1,\ldots, N\}$ independently from multivariate normal distribution with mean zero and covariance matrix $\bSig$. The  diagonal blocks of $\bSig$ are $\bSigx = \I_{d_\x\times d_\x} + (\Z_{d_\x\times d_\x} + \Z_{d_\x\times d_\x}\trans)/\norm{2(\Z_{d_\x\times d_\x} + \Z_{d_\x\times d_\x}\trans)}$ and $\bSigy = \I_{d_\y\times d_\y} + (\Z_{d_\y\times d_\y} + \Z_{d_\y\times d_\y}\trans)/\norm{2(\Z_{d_\y\times d_\y} + \Z_{d_\y\times d_\y}\trans)}$, where $\Z_{d_\x\times d_\x}\in\mR^{d_\x\times d_\x}$ and $\Z_{d_\y\times d_\y}\in\mR^{d_\y\times d_\y}$ are independent and standard normal matrices. The off-diagonal blocks of $\bSig$ are $\bSigxy = \bSigx^{1/2}\bPhi\D\bPsi\trans\bSigy^{1/2}$ and $\bSigxy\trans$, where  $\bPhi\in\mR^{d_\x\times r}$ and $\bPsi\in\mR^{d_\y\times r}$ are column-wise orthonormal matrices, and $\D\in\mR^{r\times r}$ is a diagonal matrix of the form  
 $\D = 0.1\times \I_{r\times r} + \diag(3\delta, 2\delta, \delta, 0, \ldots , 0)$,
 $\delta$ is used to control for the gap of canonical correlations.  The column-wise orthonormal matrices  $\bPhi$ and $\bPsi$ are obtained as follows. 
Let $\wt\bPhi  = (\wt{\phi}_{\ell k})\in\mR^{d_\x\times r}$.  We first sample  $\wt{\phi}_{\ell k}$ independently from standard normal, for $\ell = 1,\ldots, d_\x$ and $k = 1, \ldots, r$, then use the Gram-Schmidt process to orthonormalize the matrix $\wt\bPhi$ to obtain $\bPhi$. A similar procedure is applied to generate $\bPsi$.

We compare our proposed distributed estimate with the following competitors.
\begin{enumerate}[{(1)}]
	\item The pooled estimate. It is the classic implementation of the canonical correlation analysis, which pools all observations together to produce a single estimate.
	
	\item The naive divide-and-conquer estimate.  We merely use the observations in the $k$-th machine to calculate $\wh\bT_k$, and apply the singular value decomposition to obtain $\wh\bPhi_{1:L, k}$ and $\wh\bPsi_{1:L, k}$, which are transmitted to the central machine to form $\wt\bT_{\x} = K^{-1}\sum_{k = 1}^K\wh\bPhi_{1:L, k}\wh\bPhi_{1:L, k}\trans$ and $\wt\bT_{\y} = K^{-1}\sum_{k = 1}^K\wh\bPsi_{1:L, k}\wh\bPsi_{1:L, k}\trans$.  
The final estimates,  denoted by  $(\wh\bPhi_{1:L}^{\mathrm{DC}}, \wh\bPsi_{1:L}^{\mathrm{DC}})$,  are formed by the top-L-dim eigenspaces of $\wt\bT_{\x}$ and $\wt\bT_{\y}$.
	\item The whitened divide-and-conquer estimate. It has  the form  of $(\bhSig_{\x, 1}^{-1/2}\wh\bPhi_{1:L}^{\mathrm{DC}}, \bhSig_{\y, 1}^{-1/2}\wh\bPsi_{1:L}^{\mathrm{DC}})$. 
	\end{enumerate}

We use the following errors, which are motivated from the error bounds \eqref{coro:convergence_rate_top:eq1} and \eqref{coro:convergence_rate_topL:eq1} with the pooled estimate replaced by the population canonical directions. Define $\rho_\ell = 0$ for $\ell = \min(d_\x, d_\y) + 1, \ldots, \max(d_\x, d_\y)$. 
Specifically, for the case of top pair of canonical directions,  we evaluate the accuracy of   $(\wh\bu, \wh\bv)$ through
\beqr\label{error:top1}
\mathrm{error}(\wh\bu,\wh\bv) = \max\left(\sum_{\ell\colon \rho_\ell\leq (1- \delta)\rho_1}\abs{\inner{\bu_\ell^\ast}{\wh\bu}_{\bhSig_{\x}}}^2, \sum_{\ell\colon \rho_\ell\leq (1- \delta)\rho_1}\abs{\inner{\bv_\ell^\ast}{\wh\bv}_{\bhSig_{\y}}}^2\right).
\eeqr
 For the case of top-L-dim  pair of canonical directions,  we evaluate the accuracy of  $\wh\U_{1:L}$ and $\wh\V_{1:L}$ through
\beqr\label{error:topL}
\mathrm{error}(\wh\U_{1:L}^\dist, \wh\V_{1:L}^\dist)  = \max\left(\Norm{\inner{\U^\ast_{(L_{\delta}^\ast + 1):d_\x}}{\wh\U_{1:L}}_{\bhSig_{\x}}}, \Norm{\inner{\V^\ast_{(L_{\delta}^\ast + 1):d_\y}}{\wh\V_{1:L}}_{\bhSig_{\y}}}\right).
\eeqr
All the above errors are based on the averages of $500$ independent replications.

\subsection{The number of outer iterations\label{section:iterations}}
We first evaluate how the number of outer iterations,  $T$, affects the performance of our proposed distributed estimate in \Cref{alg:DistributedCCA_top}. We fix  $d_\x = 15$, $d_\y = 20$, $r = \min(d_\x, d_\y)$, $n = 2000$ and  $K = 30$.  Because the errors in \eqref{error:top1} and \eqref{error:topL} are very small, we report  the logarithmic errors throughout.  In this particular example, we plot the logarithmic errors against the number of outer iterations in \Cref{fig:iterations_delta15,fig:iterations_delta25} for  $\delta = 0.15$ and $\delta = 0.25$, respectively. The logarithmic errors of the pooled estimate, the naive and whitened divide-and-conquer estimates are illustrated as three horizontal straight lines because they are not obtained from iterative algorithms. The logarithmic errors of our proposed distributed estimate decay approximately linearly with respect to the number of outer iterations, which echoes our theoretical investigation in \Cref{theorem:topL_canonical_vector}.  \Cref{fig:iterations_delta15,fig:iterations_delta25} also indicate that our proposed distributed algorithm converges to the pooled estimate after around 40 iterations and outperforms both the naive and whitened divide-and-conquer estimates significantly. Our limited experience shows that it suffices to set the number of inner iterations $T^\prime = 10$  to ensure a good performance for  \Cref{alg:DistributedCCA_topL}. Therefore,  we fix $T^\prime = 10$ in what follows unless stated otherwise. 

\graphicspath{{figs/}}
\begin{figure}[!htbp]
	\centerline{
		\begin{tabular}{cc}
		\psfig{figure=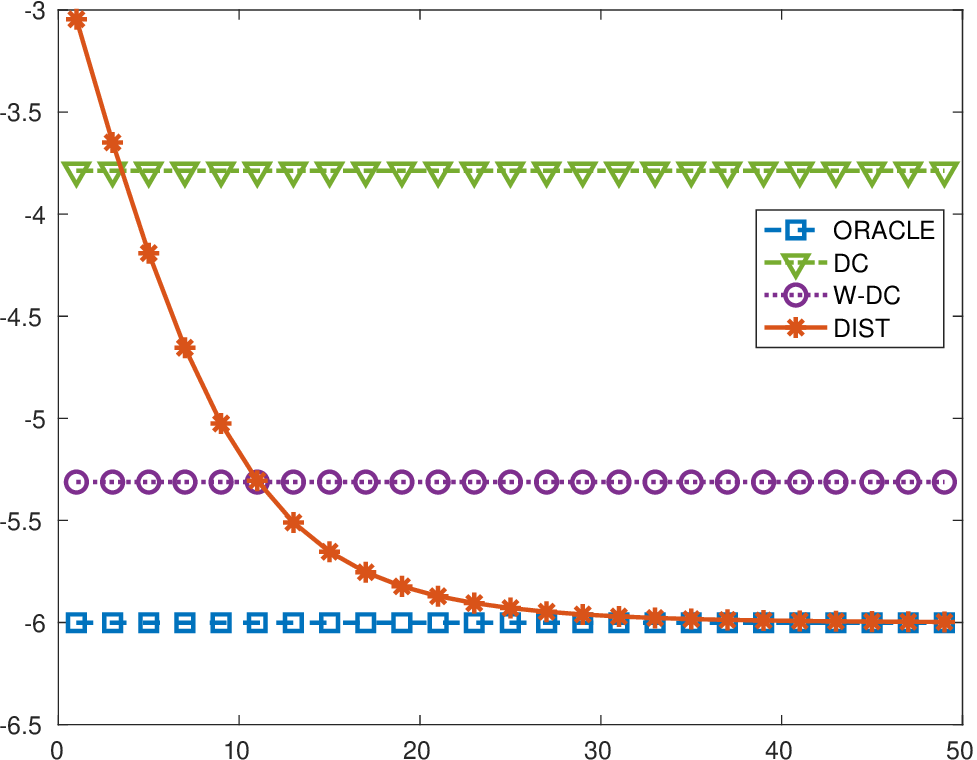,width=2.5in,angle=0} & \psfig{figure=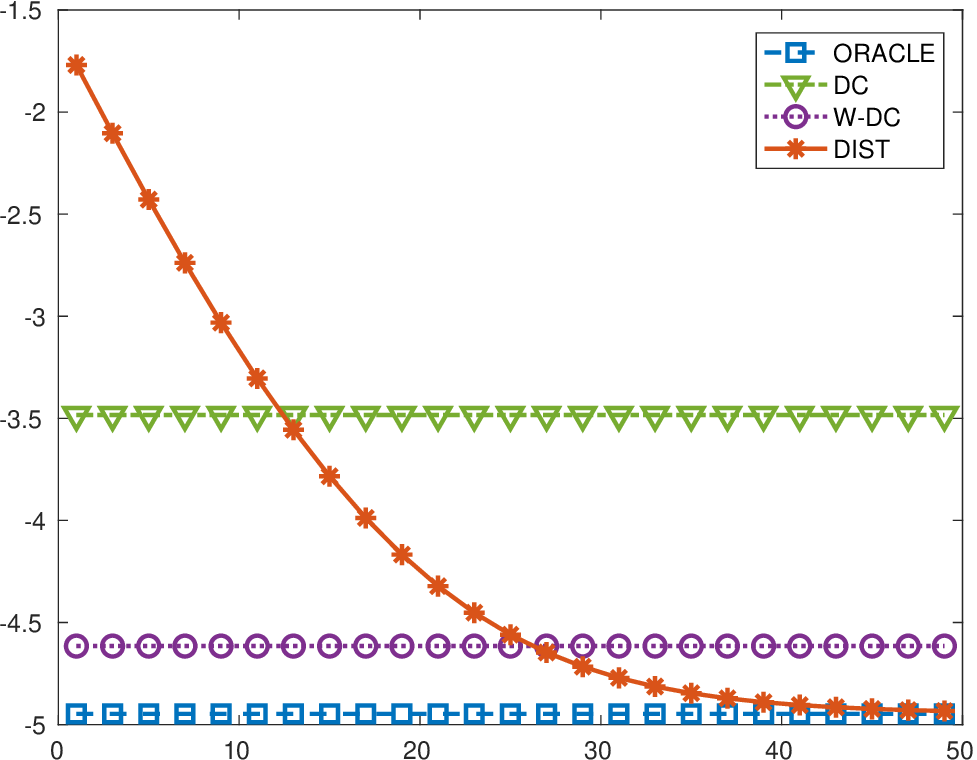,width=2.5in,angle=0} \\
			(A): Top-1-dim canonical direction   & (B): Top-2-dim canonical directions  \\ 
			\psfig{figure=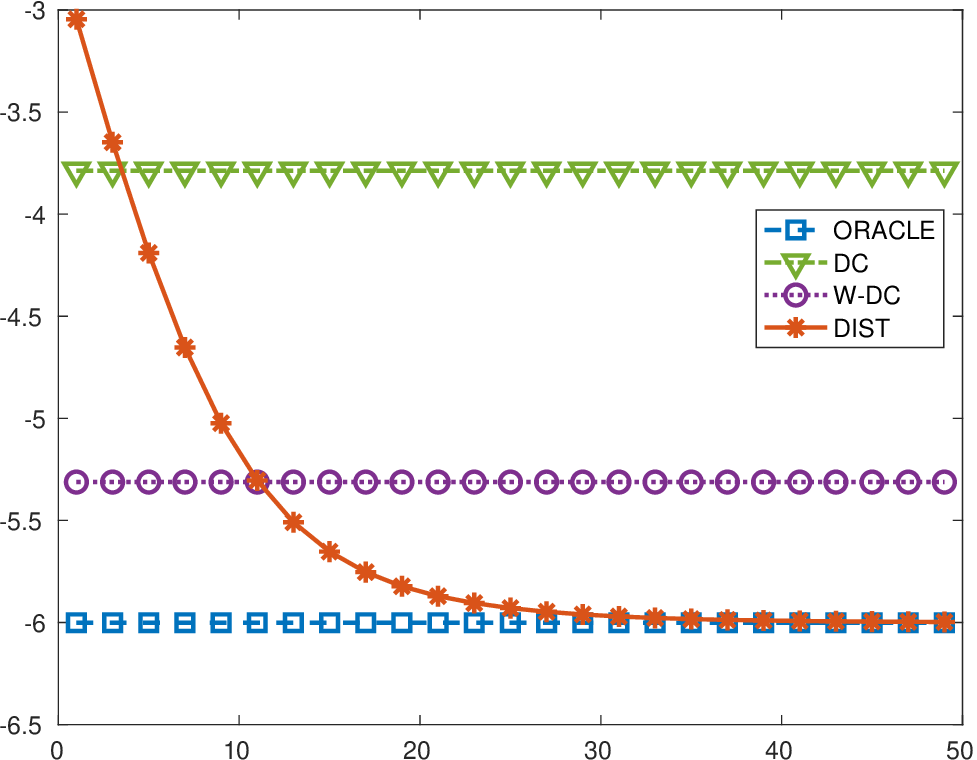,width=2.5in,angle=0} & \psfig{figure=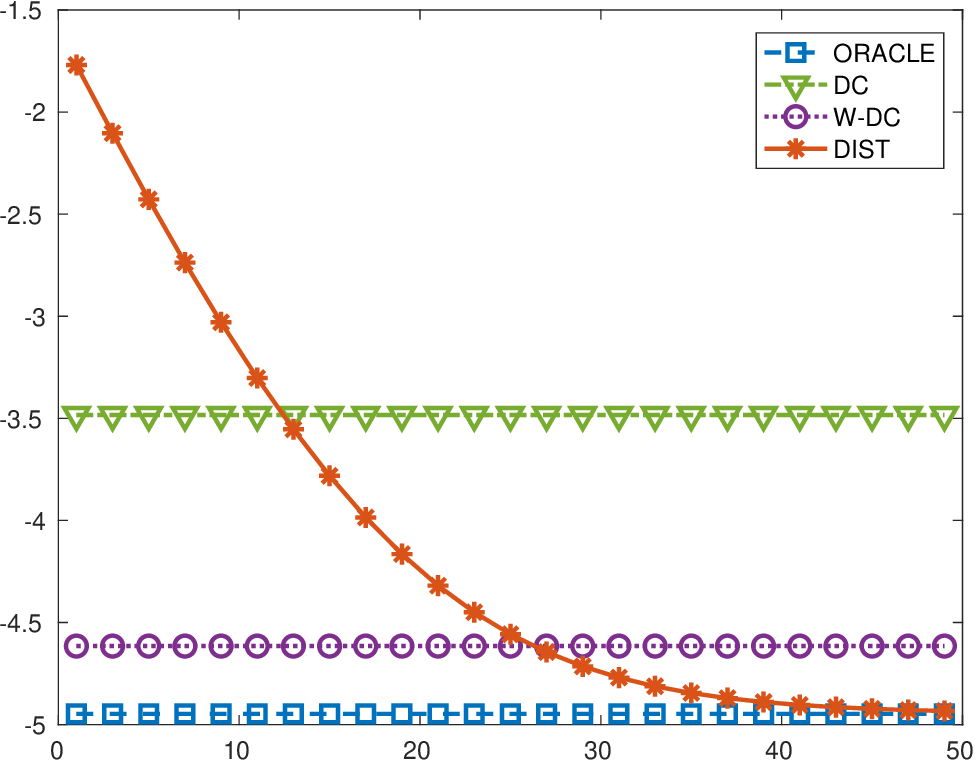,width=2.5in,angle=0} \\
			(C): Top-1-dim canonical direction   & (D): Top-2-dim canonical directions  \\ 
		\end{tabular}
	}
	  \captionsetup{font=footnotesize}
	\caption{The horizontal axis stands for the number of outer iterations, $T$, and the vertical axis stands for the logarithmic error of the naive divide-and-conquer (\chain), the whitened divide-and-conquer (\dottedcircle), the pooled (\dashed) and the distributed (\full) estimates. The number of inner iterations, $T^\prime$, equals $5$ in (A) and (B) and $10$ in (C) and (D), respectively.  We fix $\delta = 0.15$.}
	\label{fig:iterations_delta15}
\end{figure}

\graphicspath{{figs/}}
\begin{figure}[!htbp]
	\centerline{
		\begin{tabular}{cc}
		\psfig{figure=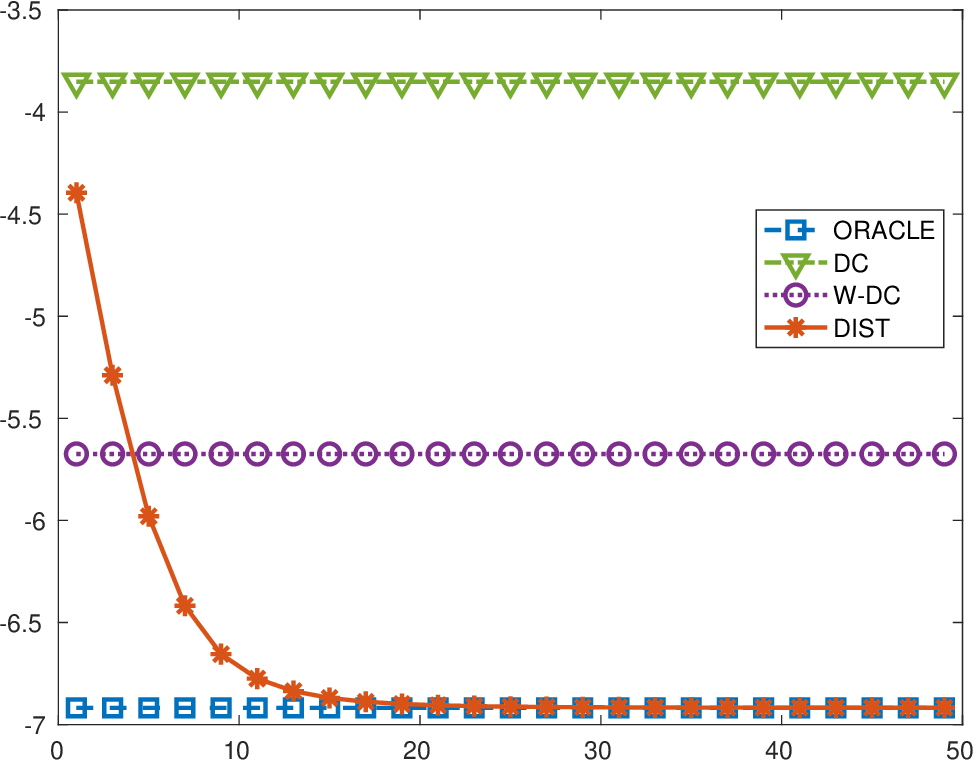,width=2.5in,angle=0} & \psfig{figure=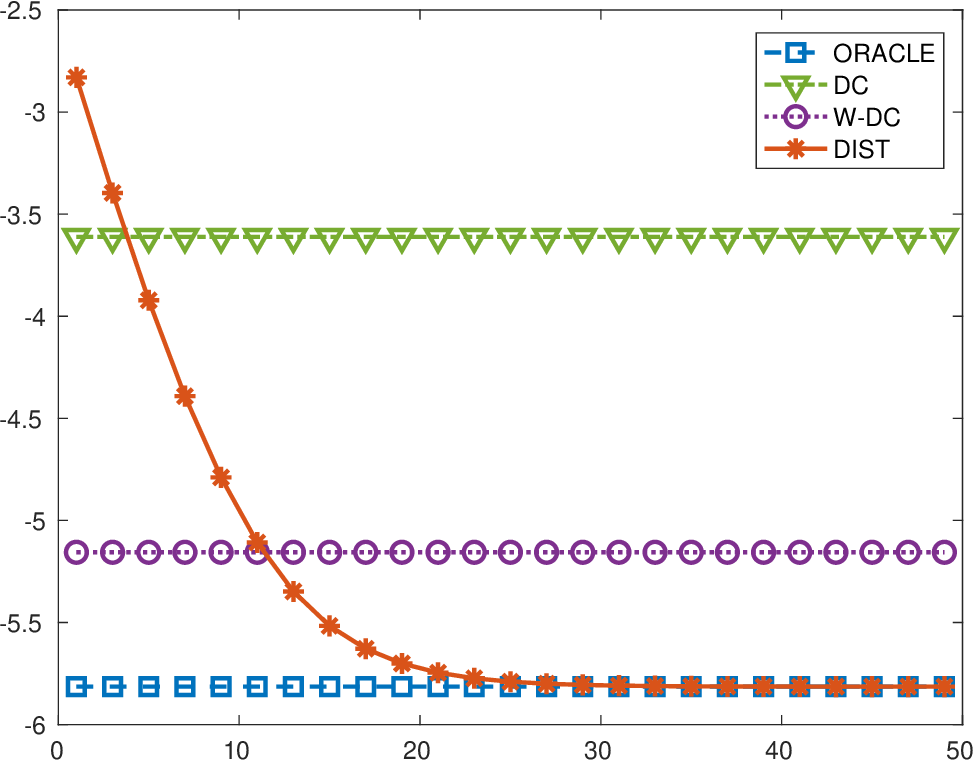,width=2.5in,angle=0} \\
			(A): Top-1-dim canonical direction   & (B): Top-2-dim canonical directions  \\ 
			\psfig{figure=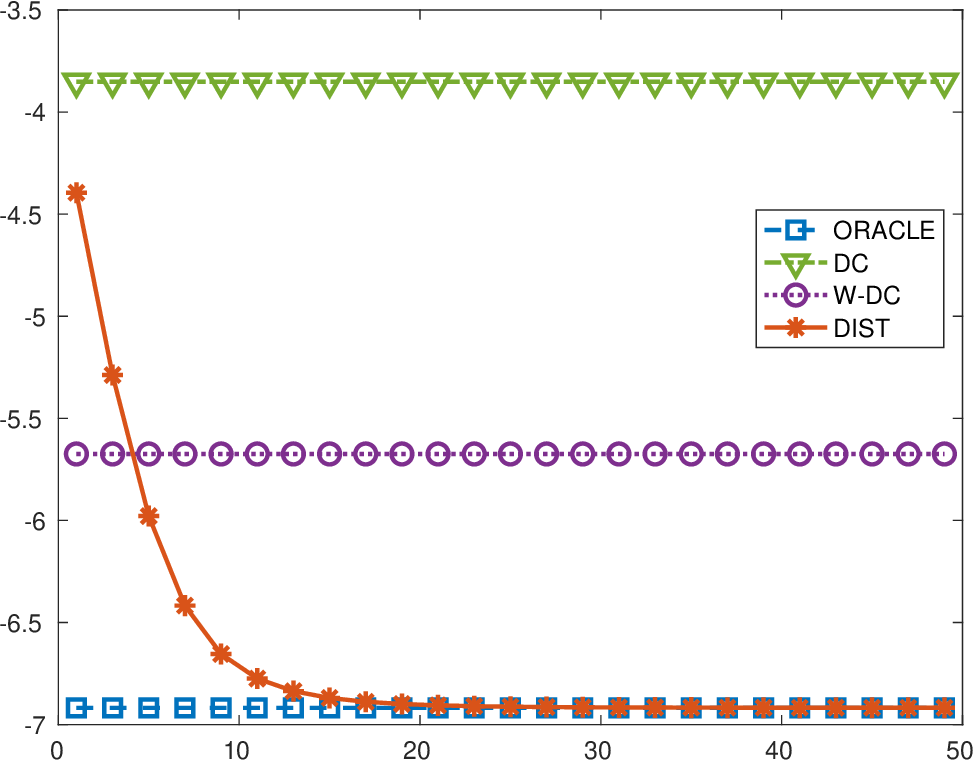,width=2.5in,angle=0} & \psfig{figure=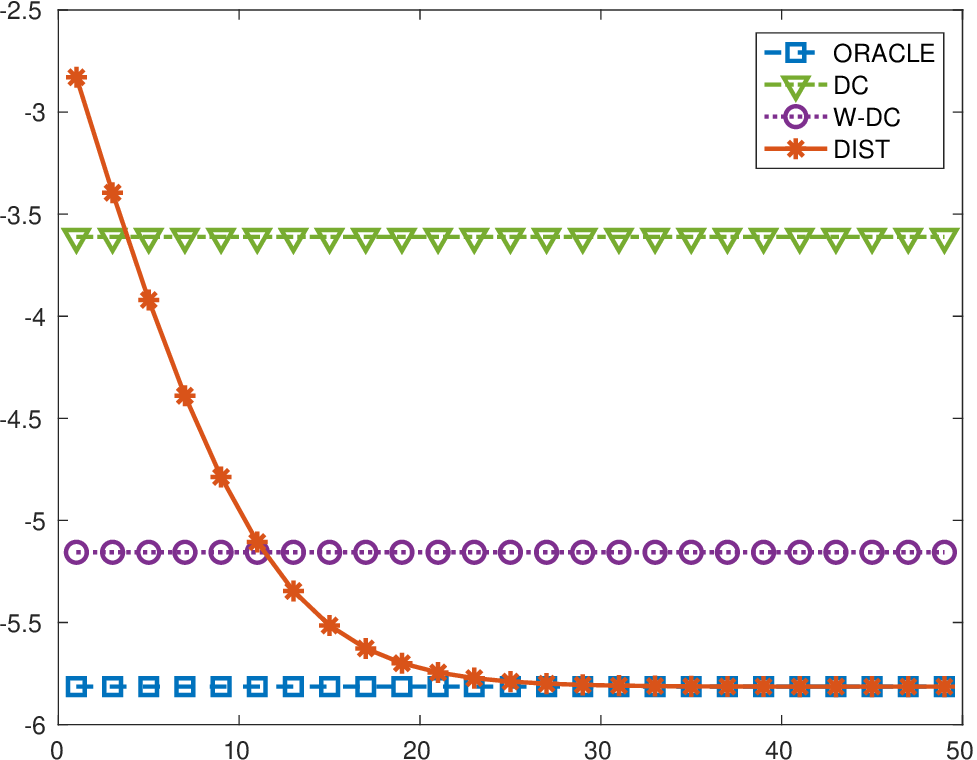,width=2.5in,angle=0} \\
			(C): Top-1-dim canonical direction   & (D): Top-2-dim canonical directions  \\ 
		\end{tabular}
	}
	  \captionsetup{font=footnotesize}
	\caption{The horizontal axis stands for the number of outer iterations, $T$, and the vertical axis stands for the logarithmic error of the naive divide-and-conquer (\chain), the whitened divide-and-conquer (\dottedcircle), the pooled (\dashed) and the distributed (\full) estimates. The number of inner iterations, $T^\prime$, equals $5$ in (A) and (B) and $10$ in (C) and (D), respectively.  We fix $\delta = 0.25$.}
	\label{fig:iterations_delta25}
\end{figure}

\subsection{The  gap of canonical correlations}
The gap of canonical correlations, $\delta$, plays a central role in the error bounds of our theoretical analysis. In general, estimating the canonical directions becomes more challenging if $\delta$ is smaller. In this example, we examine the relationship empirically between the estimation errors and the gaps of canonical correlations. 

We fix $d_\x = 15$, $d_\y = 20$, $r = \min(d_\x, d_\y)$,  $n = 2000$,  $K = 30$, and $T = 50$. We vary the  $\delta$ value and plot the results in \Cref{fig:eigen_gap}, from which it can be clearly seen that the logarithmic errors of our proposed distributed estimates increase with $1/\delta$, which aligns with the theoretical results in \Cref{theorem:topL_canonical_vector}. It can be clearly seen once again that our proposed distributed estimate performs as well as the pooled one and substantially outperforms both the naive and the whitened divide-and-conquer estimates.

\graphicspath{{figs/}}
\begin{figure}[!htbp]
	\centerline{
		\begin{tabular}{cc}
		\psfig{figure=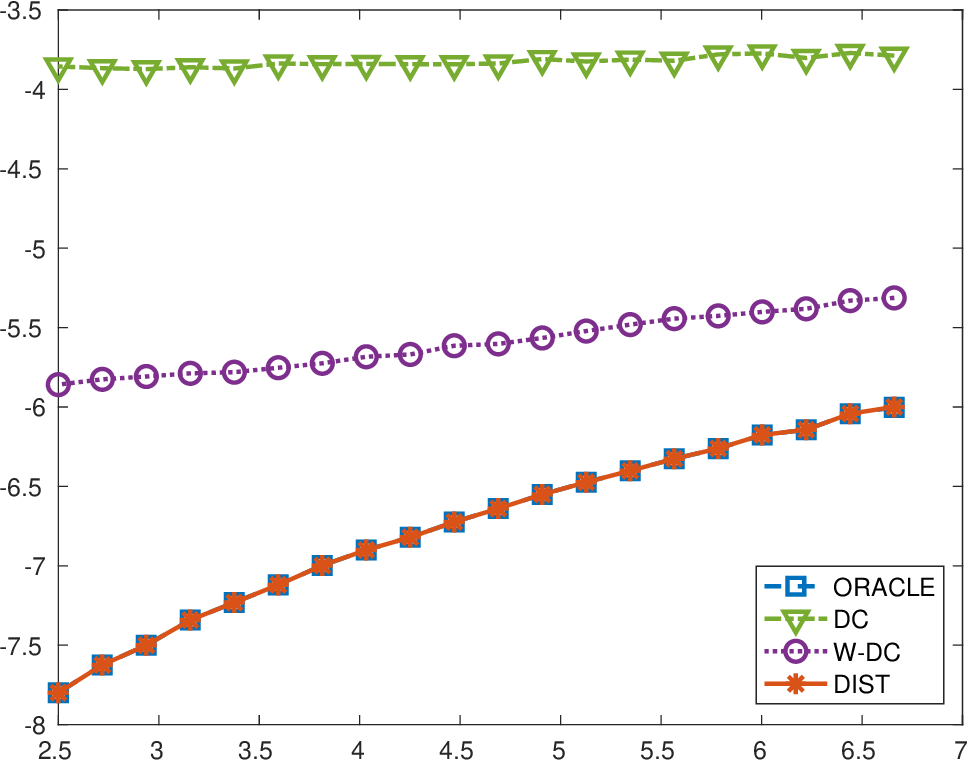,width=2.5in,angle=0} & \psfig{figure=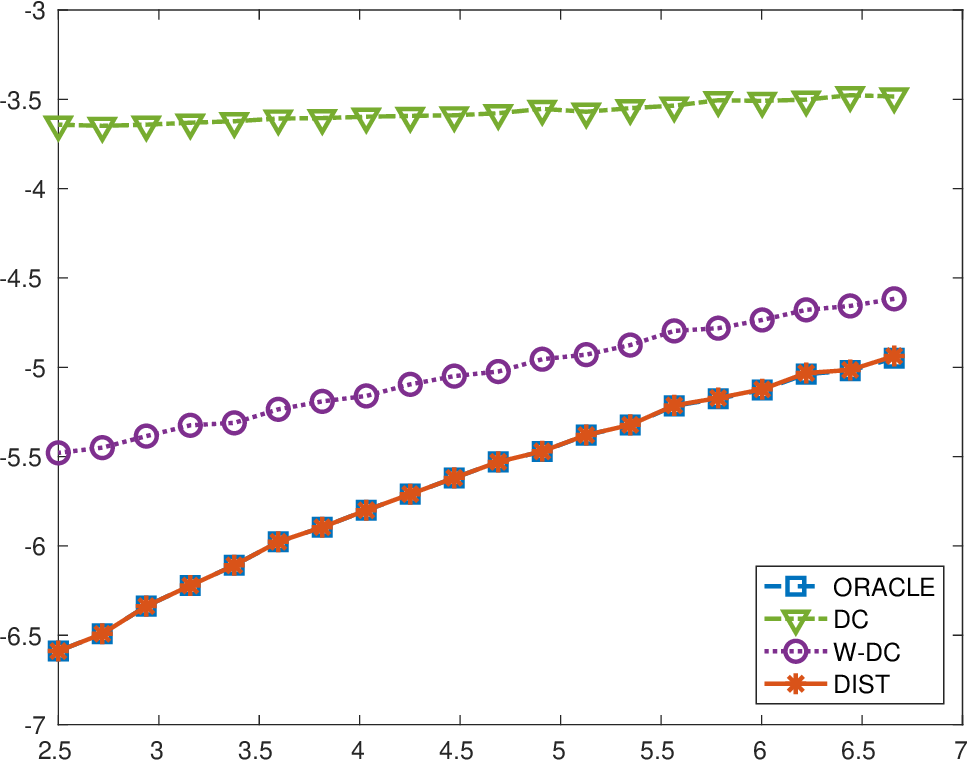,width=2.5in,angle=0} \\
			(A): Top-1-dim canonical direction   & (B): Top-2-dim canonical directions  \\ 
		\end{tabular}
	}
		  \captionsetup{font=footnotesize}
	\caption{The horizontal axis stands for  $1/\delta$, and the vertical axis stands for the logarithmic errors of the naive divide-and-conquer (\chain), the whitened divide-and-conquer (\dottedcircle), the pooled (\dashed) and the distributed (\full) estimates.}
	\label{fig:eigen_gap}
\end{figure}

\subsection{The number of local machines}
We evaluate how the number of local machines affects various distributed estimates. In this example, we  fix  $T = 50$, $T^\prime = 10$,  $d_\x = 15$, $d_\y = 20$, $r = \min(d_\x, d_\y)$ and $n = 2000$, and vary the number of local machines   $K\in\{8,  16,  32,  64, 128, 256, 512\}$. We plot the logarithmic errors against the logarithmic number of local machines in \Cref{fig:number_of_machines}, from which it can be clearly seen once again that our proposed distributed estimate achieves almost the same convergence rate as the pooled one.  In addition, the estimation errors of the whitened divide-and-conquer estimate decrease at a parallel rate as the number of machines increases. This is an instance of the bias-variance trade-off where a certain constraint is imposed to achieve the optimal convergence rate. The whitened divide-and-conquer method involves intrinsic biases for all local estimates. Similar phenomenon exists  in the principal component analysis \citep{fan2019DistributedEstimationPrincipal} and quantile regression \citep{zhao2015GeneralFrameworkRobust}, etc.

\graphicspath{{figs/}}
\begin{figure}[!htbp]
	\centerline{
		\begin{tabular}{cc}
		\psfig{figure=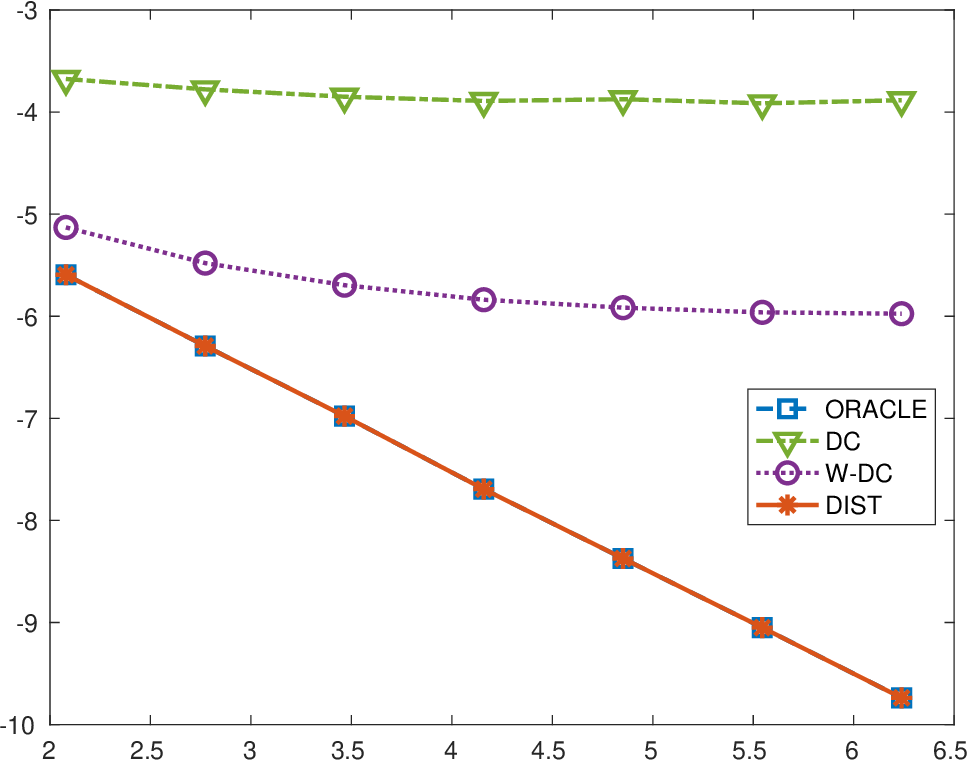,width=2.5in,angle=0} & \psfig{figure=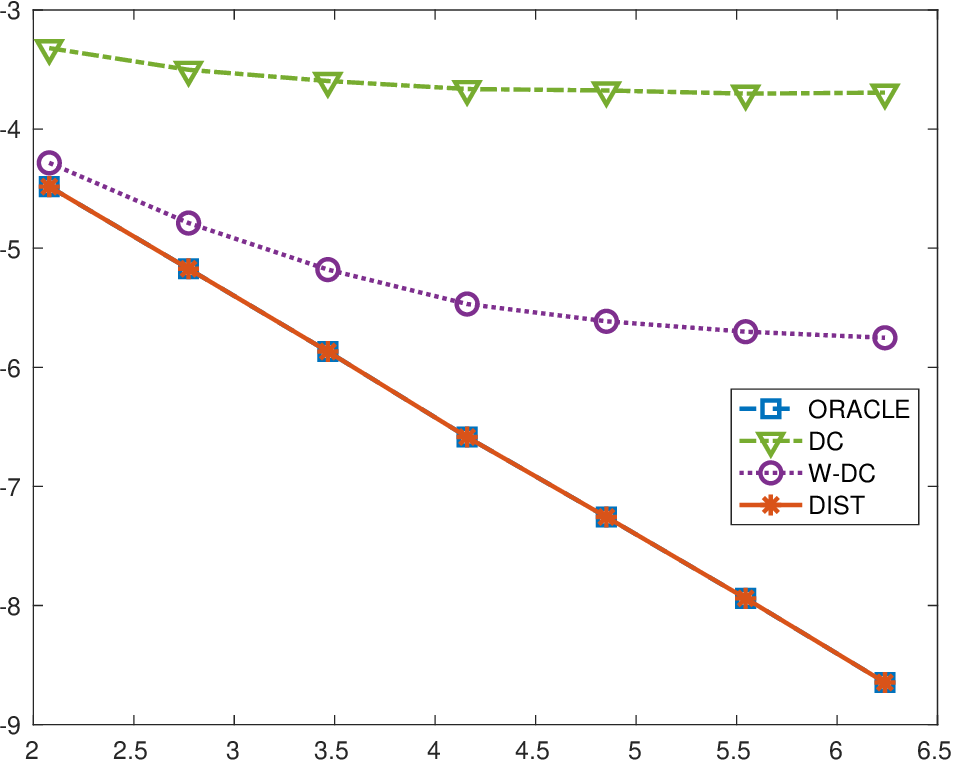,width=2.5in,angle=0} \\
			(A): Top-1-dim canonical direction   & (B): Top-2-dim canonical directions  \\ 
		\end{tabular}
	}
	  \captionsetup{font=footnotesize}
	\caption{The horizontal axis stands for the logarithmic number of local machines, and the vertical axis stands for the logarithmic errors of the naive divide-and-conquer (\chain), the whitened divide-and-conquer (\dottedcircle), the pooled (\dashed) and the distributed (\full) estimates.}
	\label{fig:number_of_machines}
\end{figure}

\section{Applications\label{section:applicatioin}}
We apply our distributed algorithms to three benchmark datasets,  MEDIAMILL \citep{snoek2006ChallengeProblemAutomated}, MNIST \citep{lecun1998GradientbasedLearningApplied}, and MFEAT \citep{Dua:2019}. These datasets  are often used to evaluate  the empirical performance of canonical correlation analysis \citep{ma2015FindingLinearStructure,wang2016EfficientGloballyConvergent}.   The first  is an  annotated video data set with $N = 30, 000$ observations, each of which is a representative key frame of a video shot annotated with $d_\x = 101$ labels and  $d_y = 120$ features.  The second is  well-known handwritten digit image data set, which   consists of   $N = 60, 000$  $28\times 28$ gray-scale images and is available at \url{http://yann.lecun.com/exdb/mnist}. We  split each image into two halves, the  left and the right of which  are  labeled respectively with $\x$ and $\y$,  and $ d_\x = d_\y =  (28\times 28)/2  = 392$. The last data set consists of $N = 2,000$ observations, each of which corresponds to six views of handwritten numerals extracted from a collection of Dutch utility maps. We merely consider two views, mfeat-fac, and mfeat-pix, with $d_\x = 216$ and $d_\y = 240$ features in our analysis.  All covariates are scaled to fall within the range of $[0, 1]$.  

We compare our proposed distributed estimate with the naive and the whitened divide-and-conquer estimates. For the purposes of visualization, this comparison is merely built upon the top three canonical directions.  We use the pooled estimate to serve as the oracle one because the underlying true canonical directions are no longer available in these applications.  To implement the distributed algorithms, we split each of the three datasets into $K\in\{2, 4, 8, 16, 25, 50, 80\}$ sub-samples evenly. We use the logarithmic estimation errors to evaluate the performance of different estimates. The results are summarized in \Cref{table:realdata}, from which it can be clearly seen that, our proposed distributed estimate deteriorates as the number of sub-samples, $K$, increases. We remark here that we fix the respective number of outer and inner iterations in our proposed iterative algorithm. Our theoretical analysis indicates that, more iterations are required to achieve the oracle rate if the number of sub-samples, $K$, is relatively larger,   or equivalently, the number of observations in each sub-sample, $n$, is smaller.   In addition, our distributed estimate outperforms all other estimates across all possible values of $K$. The naive divide-and-conquer estimate performs the worst, partly because it does not account for the orthogonality constraints.

\begin{table}[!htbp]   \footnotesize	                                                           
\centering         
  \captionsetup{font=footnotesize}
\caption{The logarithmic estimation errors of the naive divide-and-conquer estimate (N-DC), the whitening divide-and-conquer estimate (W-DC), and our proposed distributed estimate (DIST) for different numbers of sub-samples $K$.  }                                                   
\label{table:realdata}                                                             
\begin{tabular}{l|cccccccc}     
\hline                            
  & $K=2$ & $K=4$ & $K=8$ &  $K=16$ & $K=25$ & $K=50$ & $K=80$ \\   
\hline
& \multicolumn{7}{c}{MMILL}\\   
\cline{2-8}
N-DC & -4.2752 & -4.2712 & -4.2673 & -4.2634 & -4.2477 & -4.2323 & -4.1880 \\    
W-DC & -9.7303 & -8.4329 & -7.4726 & -6.5812 & -6.0755 & -5.6983 & -5.2458 \\    
DIST & -24.6840 & -17.6880 & -19.2427 & -12.6477 & -9.4220 & -8.0301 & -9.2124 \\ 
\cline{2-8}    
& \multicolumn{7}{c}{MNIST}\\   
\cline{2-8}        
N-DC & -1.8975 & -1.8966 & -1.8951 & -1.8926 & -1.8901 & -1.8831 & -1.8718 \\      
W-DC & -9.1328 & -8.1677 & -7.1788 & -6.3508 & -5.9515 & -5.0023 & -4.5328 \\      
DIST & -41.5599 & -35.1307 & -26.3739 & -21.4233 & -21.7905 & -12.4132 & -9.0029 \\
\cline{2-8}  
& \multicolumn{7}{c}{MFEAT}\\  
\cline{2-8}
N-DC & -2.1189 & -2.1130 & -2.1535 & -2.0217 & -1.9283 & -1.8018 & -1.7706 \\    
W-DC & -7.7868 & -6.8604 & -6.1246 & -5.0453 & -4.7029 & -4.1052 & -3.2374 \\    
DIST & -24.2076 & -17.1580 & -14.1297 & -10.8736 & -9.4998 & -7.5504 & -6.1650 \\  
\hline
\end{tabular}                                                    
\end{table}

\section{Concluding Remarks\label{section:conclusion}}
We propose distributed algorithms for canonical correlation analysis, which adapts the convex formulation to a distributed setting, and allows us to use stochastic algorithms to solve  linear systems. This algorithm uses a multi-round scheme to achieve a fast convergence rate and to relieve the constraint on the number of machines. We provide theoretical guarantees for the convergence rate of the resultant estimates.   These distributed algorithms can be generalized from  different perspectives. For example, one can introduce sparsity to consider  ultrahigh dimensional canonical correlation  analysis \citep{gao2017SparseCCAAdaptive}, or adapt   our theoretical analyses to many other relevant problems, such as partial least squares regression \citep{chen2019DroppingConvexityMore} and linear discriminant analysis \citep{bach2005ProbabilisticInterpretationCanonical}.  Investigations along these lines are under way.

\acks{This research is supported by grants from Renmin University of China (22XNA026), Beijing Natural Science Foundation (Z190002) and National Natural Science Foundation of China (12225113 and 12171477). Liping Zhu is the corresponding author. }


\newpage

\renewcommand{\thetheo}{\Alph{section}.\arabic{theo}}
\renewcommand{\theprop}{\Alph{section}.\arabic{prop}}
\renewcommand{\thelemm}{\Alph{section}.\arabic{lemm}}
\renewcommand{\thefigure}{\Alph{section}.\arabic{figure}}
\renewcommand{\thetable}{\Alph{section}.\arabic{table}}
\appendix

\section{Some Technical Lemmas}
This section collects some technical lemmas used in the proof of the main results. To begin with, we first introduce several notations for ease of presentation. Denote $r = \min(d_\x, d_\y)$. The sorted eigenvalues of $\wh \bC$ are
\beqrs
\wh\lambda_1\defn \wh\rho_1\geq\ldots\geq \wh\lambda_r\defn \rho_r\geq 0 = \ldots = 0 \geq  \wh\lambda_{d - r + 1}\defn -\wh\rho_r\geq \ldots\geq \wh\lambda_d\defn -\wh\rho_1,
\eeqrs
with corresponding unit eigenvectors
\beqr\label{structure_of_C_eigenvectors}
\nonumber \wh\bmr_1 \defn \begin{pmatrix}
	\bhSigx^{1/2}\wh\bu_1^\pool\\
	\bhSigy^{1/2}\wh\bv_1^\pool\\
\end{pmatrix}\Bigg / 2^{1/2},\ \ldots,\  \wh\bmr_r \defn \begin{pmatrix}
	\bhSigx^{1/2}\wh\bu_r^\pool\\
	\bhSigy^{1/2}\wh\bv_r^\pool\\
\end{pmatrix}\Bigg / 2^{1/2},
 \ldots \\
  \wh\bmr_{d - r + 1} \defn \begin{pmatrix}
	- \bhSigx^{1/2}\wh\bu_r^\pool\\
	\bhSigy^{1/2}\wh\bv_r^\pool\\
\end{pmatrix}\Bigg / 2^{1/2},\ \ldots,\ \wh\bmr_{d} \defn \begin{pmatrix}
	- \bhSigx^{1/2}\wh\bu_1^\pool\\
	\bhSigy^{1/2}\wh\bv_1^\pool\\
\end{pmatrix}\Bigg / 2^{1/2}.
\eeqr

\begin{lemm}
	\label{lemma:initial_estimate}
	If our data $\{(\x_i, \y_i)$, $i = 1,\ldots, N\}$ are sub-Gaussian vectors,   then  we have
\beqr\label{lemma:initial_estimate:T_contraction}
  \norm{\wh\bT_1 - \wh\bT} = O_p\{(d\log^2 d/n)^{1/2}\}. 
\eeqr
The top singular value on the first machine $\sigma_1(\wh\bT_1)$ satisfies,
\beqr\label{lemma:initial_estimate:canonical_correlation}
  \abs{\sigma_1(\wh\bT_1) - \wh\rho_1}\leq \norm{\wh\bT_1 - \wh\bT}.
\eeqr
Furthermore, let $(\wh\bu^{(0)}, \wh\bv^{(0)})$ be the top pair of canonical directions estimate on the first machine, we have the following gap-free contraction bound for 
\beqrs
  \wh\bmr^{(0)} \defn \begin{pmatrix}
	\bhSigx^{1/2}\wh\bu^{(0)}\\
	\bhSigy^{1/2}\wh\bv^{(0)}\\
\end{pmatrix}\Bigg / 2^{1/2},
\eeqrs
and $\wh\bmr_1,\ldots, \wh\bmr_d$,
\beqr\label{lemma:initial_estimate:gap_free_bounds}
  \sum_{ \ell \colon \wh\lambda_\ell\leq (1- \delta)\wh\lambda_1}\abs{\inner{\wh\bmr_\ell}{\wh\bmr^{(0)}}}^2 \leq 2\norm{\wh\bT_1 - \wh\bT}/(\gamma\delta\wh\lambda_1),
\eeqr
with an overwhelming probability, which further implies that
\beqr\label{lemma:initial_estimate:gap_free_bounds_for_singular}
\nonumber  \sum_{ \ell\colon \wh\rho_\ell\leq (1- \delta)\wh\rho_1} \abs{\inner{\wh\bu_\ell^\pool}{\wh\bu^{(0)}}_{\bhSigx}}^2 + \sum_{ \ell\colon \wh\rho_\ell\leq (1- \delta)\wh\rho_1}\abs{\inner{\wh\bv_\ell^\pool}{\wh\bv^{(0)}}_{\bhSigy}}^2 \\
\leq  4\norm{\wh\bT_1 - \wh\bT}/(\gamma\delta\wh\rho_1).
\eeqr

\end{lemm}
\begin{proof}[\bf Proof of \Cref{lemma:initial_estimate}]
	By Corollary 7 in \cite{gao2019StochasticCanonicalCorrelation}, we have with high probability that 
\beqrs
  \norm{\wh\bT_1 - \bT}\leq \nu, \ \text{and}\  \norm{\wh\bT - \bT}\leq \nu
\eeqrs
if $n\geq C d\log^2 d/\nu^2$ and $N\geq C d\log^2 d/\nu^2$ where $C$ is a positive constant independent of $\nu$, $n$, $N$ and $d$. Therefore, our first claim follows the matrix concentrations above and the triangle inequality for the spectral norm.

Denote by $\wh\bphi_1$ and $\wh\bpsi_1$ the top left and right singular vectors for $\wh\bT$, and $\wh\bphi_1^{(0)}$ and $\wh\bpsi_1^{(0)}$ are  those for $\wh\bT_1$. Without loss of generality, assume that $\wh\rho_1\geq \rho_1(\wh\bT_1)$. Then we have
\beqrs
  \abs{\sigma_1(\wh\bT_1) - \wh\rho_1} = \wh\bphi_1\trans\wh\bT\wh\bpsi_1 - (\wh\bphi_1^{(0)})\trans\wh\bT_1\wh\bpsi_1^{(0)} \leq \wh\bphi_1\trans\wh\bT\wh\bpsi_1 - \wh\bphi_1\trans\wh\bT_1\wh\bpsi_1\leq \norm{\wh\bT_1 - \wh\bT}, 
\eeqrs
where the first inequality is by the definition of singular value, and the second one owes to the definition of induced norm.

It remains to deal with \eqref{lemma:initial_estimate:gap_free_bounds} and \eqref{lemma:initial_estimate:gap_free_bounds_for_singular}. 
Let 
\beqrs
  \wt\bmr^{(0)} \defn \begin{pmatrix}
	\bhSig_{\x,1}^{1/2}\wh\bu^{(0)}\\
	\bhSig_{\y,1}^{1/2}\wh\bv^{(0)}\\
\end{pmatrix}\Bigg / 2^{1/2}.
\eeqrs
By (variant) Davis-Kahan Theorem \citep{yu2015UsefulVariantDavis}, it is standard to show that 
\beqrs
  \sum_{\ell\colon \wh\lambda_\ell\leq (1- \delta)\wh\lambda_1}\abs{\inner{\wh\bmr_\ell}{\wt\bmr^{(0)}}}^2 \leq \norm{\wh\bT_1 - \wh\bT}/(\delta\wh\rho_1). 
\eeqrs
For each $\ell$, we have
\beqrs
\abs{\inner{\wh\bmr_\ell}{\wh\bmr^{(0)}}} = \Abs{\wh\bmr_\ell\trans
\begin{pmatrix}
	\bhSigx^{1/2} & \bm{0}\\
	\bm{0} & \bhSigy^{1/2} 
\end{pmatrix}
\begin{pmatrix}
	\bhSig_{\x,1}^{-1/2} & \bm{0}\\
	\bm{0} & \bhSig_{\y,1}^{-1/2} 
\end{pmatrix}
\wt\bmr^{(0)}
}\leq (2/\gamma)^{1/2}\abs{\inner{\wh\bmr_\ell}{\wt\bmr^{(0)}}},
\eeqrs
where the last inequality is because $\max(\norm{\bhSig_{\x,1}^{-1}}, \norm{\bhSig_{\y,1}^{-1}})\leq 2/\gamma$ holds with an overwhelming probability. 
This together with the inequality above proves \eqref{lemma:initial_estimate:gap_free_bounds}.

As for the last claim, combining \eqref{structure_of_C_eigenvectors} and \eqref{lemma:initial_estimate:gap_free_bounds}, we have
\beqrs
   & & \sum_{\ell\colon \wh\rho_\ell\leq (1- \delta)\wh\rho_1} \abs{\inner{\wh\bu_\ell}{\wh\bu^{(0)}}_{\bhSigx}}^2/2 + \abs{\inner{\wh\bv_\ell}{\wh\bv^{(0)}}_{\bhSigy}}^2/2 \\
   & = &  \sum_{ \ell \colon \wh\rho_\ell\leq (1- \delta)\wh\rho_1} \abs{\inner{\wh\bmr_\ell}{\wh\bmr^{(0)}}}^2 + \abs{\inner{\wh\bmr_{d - \ell + 1}}{\wh\bmr^{(0)}}}^2\\
   & \leq & \sum_{\ell\colon \wh\lambda_\ell\leq (1- \delta)\wh\lambda_1}\abs{\inner{\wh\bmr_\ell}{\wh\bmr^{(0)}}}^2 \\
   & \leq & 2\norm{\wh\bT_1 - \wh\bT}/(\gamma\delta\wh\lambda_1).
\eeqrs
This completes the proof of \Cref{lemma:initial_estimate}.

\end{proof}

\section{Proof of \Cref{theorem:top_canonical_vector}}

Define 
\beqr\label{definition:Mrho1}
  \M_{\rhob_1, 1} = 
  \begin{pmatrix}
  	\bhSigx^{1/2} & \\
  	& \bhSigy^{1/2}
  \end{pmatrix}
    \begin{pmatrix}
  	\rhob_1\bhSig_{\x,1} & -\bhSig_{\x, \y,1} \\
  	-\bhSig_{\x, \y,1}\trans & \rhob_1\bhSig_{\y,1}
  \end{pmatrix}^{-1}
    \begin{pmatrix}
  	\bhSigx^{1/2} & \\
  	& \bhSigy^{1/2}
  \end{pmatrix}
\eeqr
and
\beqr\label{definition:Mrho}
\nonumber  \M_{\rhob_1} & = &
  \begin{pmatrix}
  	\bhSigx^{1/2} & \\
  	& \bhSigy^{1/2}
  \end{pmatrix}
    \begin{pmatrix}
  	\rhob_1\bhSig_{\x} & -\bhSig_{\x, \y} \\
  	-\bhSig_{\x, \y}\trans & \rhob_1\bhSig_{\y}
  \end{pmatrix}^{-1}
    \begin{pmatrix}
  	\bhSigx^{1/2} & \\
\nonumber  	& \bhSigy^{1/2}
  \end{pmatrix}\\
  & = &
  \begin{pmatrix}
  	\bhSigx^{-1/2} & \\
\nonumber  	& \bhSigy^{-1/2}
  \end{pmatrix}^{-1}
    \begin{pmatrix}
  	\rhob_1\bhSig_{\x} & -\bhSig_{\x, \y} \\
  	-\bhSig_{\x, \y}\trans & \rhob_1\bhSig_{\y}
  \end{pmatrix}^{-1}
    \begin{pmatrix}
  	\bhSigx^{-1/2} & \\
  	& \bhSigy^{-1/2}
  \end{pmatrix}^{-1}\\
\nonumber  & = &
  \left\{
  \begin{pmatrix}
  	\bhSigx^{-1/2} & \\
\nonumber  	& \bhSigy^{-1/2}
  \end{pmatrix}
    \begin{pmatrix}
  	\rhob_1\bhSig_{\x} & -\bhSig_{\x, \y} \\
  	-\bhSig_{\x, \y}\trans & \rhob_1\bhSig_{\y}
  \end{pmatrix}
    \begin{pmatrix}
  	\bhSigx^{-1/2} & \\
  	& \bhSigy^{-1/2}
  \end{pmatrix}
  \right\}^{-1}\\
  & = &     \begin{pmatrix}
  	\rhob_1\I & -\wh\bT \\
  	-\wh\bT\trans & \rhob_1\I
  \end{pmatrix}^{-1} = (\rhob_1\I - \wh\bC)^{-1}.
\eeqr

The following lemma  characterizes the convergence rate of the \textit{outer loop} in \Cref{alg:DistributedCCA_top}.  
{\lemm\label{lemma:outer}
Suppose the initial estimate $\overline{\rho}_1$ satisfies
\beqr\label{lemma:outer:eigen}
\omega\leq \overline{\rho}_1 - \wh\rho_1\leq 2\omega, \ \text{for some} \ \omega>0.
\eeqr
For any $(\bu\trans, \bv\trans)\trans\in\mR^d$ and $(\wt\bu\trans, \wt\bv\trans)\trans\in \mR^d$, define 
\beqrs
\bmr = \begin{pmatrix}
	\bhSigx^{1/2}\bu\\
	\bhSigy^{1/2}\bv\\
\end{pmatrix}\Bigg / 2^{1/2},\ \text{and}\ 
\wt\bmr = \begin{pmatrix}
	\bhSigx^{1/2}\wt\bu\\
	\bhSigy^{1/2}\wt\bv\\
\end{pmatrix}\Bigg / 2^{1/2}.
\eeqrs
Assume that
\beqr\label{lemma:outer:initial}
\norm{\bmr} = 1, \sum_{ \ell \colon \wh\lambda_\ell\leq (1- \delta)\wh\lambda_1}\abs{\inner{\wh\bmr_\ell}{\bmr}}^2\leq 3/4,
\eeqr
and 
\beqr\label{lemma:outer:residual}
\Norm{
\wt\bmr - \M_{\overline{\rho}_1}\bmr}\leq \varepsilon\leq (8\omega)^{-1}.
\eeqr
For each index $\ell = 1\ldots$ such that $\wh\rho_\ell\leq (1- \delta)\wh\rho_1$, we have
\beqrs
\frac{\abs{\inner{\wh\bmr_\ell}{\wt\bmr}}}{\norm{\wt\bmr}}\leq \frac{8\omega}{\delta\wh\rho_1}  \frac{\abs{\inner{\wh\bmr_\ell}{\bmr}}}{\norm{\bmr}}  + 8 \omega\varepsilon.
\eeqrs
Moreover, we have
\beqr\label{lemma:outer:total}
\sum_{\ell\colon \wh\lambda_\ell\leq (1- \delta)\wh\lambda_1} \frac{\abs{\inner{\wh\bmr_\ell}{\wt\bmr}}^2}{\norm{\wt\bmr}^2}\leq \frac{128\omega^2}{\delta^2\wh\rho_1^2}
 \sum_{\ell\colon \wh\lambda_\ell\leq (1- \delta)\wh\lambda_1} \frac{\abs{\inner{\wh\bmr_\ell}{\bmr}}^2}{\norm{\bmr}^2} + 128 \omega^2\varepsilon^2.
\eeqr
}
For the outer loop in our \Cref{alg:DistributedCCA_top}, $(\bu, \bv)$ and $(\wt\bu, \wt\bv)$ in \Cref{lemma:outer} can be explained as the $t$-th round and $(t+1)$-th round estimate $(\wh\bu^{(t)}, \wh\bv^{(t)})$ and $(\wh\bu^{(t + 1)}, \wh\bv^{(t + 1)})$. Define $\B = \diag(\wh\bSig_\x, \wh\bSig_\y)$.   This lemma shows that up to a numerical tolerance $\varepsilon$ for inverting $\wh\bH$ (Condition \eqref{lemma:outer:residual}), each iteration of the outer loop reduces the magnitude of the projection of $(\wh\bu^{(t)}, \wh\bv^{(t)})$ onto $(\wh\bu_\ell, \wh\bv_\ell)$ by a factor of $O\{\omega/(\delta\wh\rho_1)\}\ll 1$ under the inner product induced by $\B$ given $\omega\ll 1$ (if we have a good initial estimate of $\wh\rho_1$ and $\delta\wh\rho_1 = \Omega(1)$).  Notice that if $(\wh\bu^{(t)}, \wh\bv^{(t)})$ satisfies Condition \eqref{lemma:outer:initial}, the result of \Cref{lemma:outer} \eqref{lemma:outer:total} claims that $(\wh\bu^{(t + 1)}, \wh\bv^{(t + 1)})$ fulfills Condition \eqref{lemma:outer:initial} as well. This condition is justified if $(\wh\bu^{(0)}, \wh\bv^{(0)})$ satisfies Condition \eqref{lemma:outer:initial}, which is a conclusion for the local empirical risk minimizer by \Cref{lemma:initial_estimate} or for the stochastic optimization solution of \cite{wang2016EfficientGloballyConvergent} by Theorem 4 of it.

The second lemma concerns the convergence rate of distributively solving the linear system 
\beqrs
\begin{pmatrix}
	\bu\\
	 \bv
\end{pmatrix} 
= \begin{pmatrix}
	\overline{\rho}_1\bhSig_{\x} & -\bhSig_{\x, \y}\\
	-\bhSig\trans_{\x, \y} & \overline{\rho}_1\bhSig_{\y}\\
\end{pmatrix}^{-1}
\begin{pmatrix}
	\bhSig_{\x} & \bm{0}\\
	\bm{0} & \bhSig_{\y} \\
\end{pmatrix}
\begin{pmatrix}
	\wh\bu^{(t)}\\
	 \wh\bv^{(t)}
\end{pmatrix},
\eeqrs
in the \textit{inner loop} of \Cref{alg:DistributedCCA_top}. Recall that in \eqref{eq:iteration}, 
\beqrs
\begin{pmatrix}
	\wt\bu^{(t + 1)}\\
	 \wt\bv^{(t + 1)}
\end{pmatrix} 
= \begin{pmatrix}
	\overline{\rho}_1\bhSig_{\x} & -\bhSig_{\x, \y}\\
	-\bhSig\trans_{\x, \y} & \overline{\rho}_1\bhSig_{\y}\\
\end{pmatrix}^{-1}
\begin{pmatrix}
	\bhSig_{\x} & \bm{0}\\
	\bm{0} & \bhSig_{\y} \\
\end{pmatrix}
\begin{pmatrix}
	\wh\bu^{(t)}\\
	 \wh\bv^{(t)}
\end{pmatrix}
\eeqrs
is the exact solution of this linear system. Define 
\beqrs
\wh\bmr^{(t + 1)}_j = \begin{pmatrix}
	\bhSigx^{1/2}\wh\bu^{(t + 1)}_j\\
	\bhSigy^{1/2}\wh\bv^{(t + 1)}_j\\
\end{pmatrix}\Bigg / 2^{1/2},\ \text{and}\ 
\wt\bmr^{(t + 1)} = \begin{pmatrix}
	\bhSigx^{1/2}\wt\bu^{(t + 1)}\\
	\bhSigy^{1/2}\wt\bv^{(t + 1)}\\
\end{pmatrix}\Bigg / 2^{1/2}.
\eeqrs

{\lemm\label{lemma:inner}
Recall $\kappa = \norm{\wh\bT_1 - \wh\bT}$ in \Cref{theorem:top_canonical_vector}. 
Suppose the initial estimate $\overline{\rho}_1$ satisfies
\beqrs
\rhob_1 - \wh\rho_1\geq \omega\geq 2\norm{\wh\bT_1 - \wh\bT}.
\eeqrs
Then for each $j = 0, 1, \ldots, (T^\prime - 1)$, we have
\beqr\label{lemma:inner:result}
\norm{\wh\bmr^{(t + 1)}_{j + 1} - \wt\bmr^{(t + 1)}} 
\leq 8\kappa/(\gamma\omega)\norm{\wh\bmr^{(t + 1)}_{j} - \wt\bmr^{(t + 1)}}, 
\eeqr
with an overwhelming probability.
}

Here $\kappa$ on the right hand side of \eqref{lemma:inner:result} is due to the approximation using the Hessian matrix $\wh\bH_1$ on the first machine in place of original Hessian matrix $\wh\bH$. By standard matrix concentration inequalities, we have $\kappa = O_p\{(d\log^2d/n)^{1/2}\}$. Consequently, the inner loop of \Cref{alg:DistributedCCA_top} has a contraction rate of order $O\{1/(\gamma\omega)(d\log^2d/n)^{1/2}\}$, which is inversely proportional to the gap $(\wb\rho_1 - \wh\rho_1)$ and the condition number $\gamma$.

\begin{proof}[\bf Proof of \Cref{theorem:top_canonical_vector}]
	We first give the analysis of Phase I of \Cref{alg:DistributedCCA_top}. By \Cref{lemma:inner}, we have (recall that $\wh\bu_0^{(t + 1)} = \wh\bu^{(t)}_{T^\prime}$ and $\wh\bv_0^{(t + 1)} = \wh\bv^{(t)}_{T^\prime}$ after normalizing),
\beqrs
  \norm{\wh\bmr^{(t + 1)}_{T^\prime} - \M_{\rhob_1}\wh\bmr^{(t)}_{T^\prime}}\leq \{8\kappa/(\gamma\omega)\}^{T^\prime}\norm{\wh\bmr^{(t)}_{T^\prime} - \M_{\rhob_1}\wh\bmr^{(t)}_{T^\prime}}\leq \{8\kappa/(\gamma\omega)\}^{T^\prime}2/\omega \defn \varepsilon_{T^\prime},
\eeqrs
where we use the fact that $\norm{\I - \M_{\rhob_1}}\leq \norm{\I}  +  \norm{\M_{\rhob_1}} = 1 + (\wb\rho_1 - \wh\rho_1)^{-1}\leq 1 + 1/\omega\leq 2/\omega$, and $\norm{\wh\bmr^{(t)}_{T^\prime}} = 1$. Now we can recursively apply \eqref{lemma:outer:total} with $\varepsilon$ replaced by $\varepsilon_{T^\prime}$ to obtain
\beqr\label{theorem:top_canonical_vector:medium_results}
\nonumber & &  \sum_{\ell\colon \wh\lambda_\ell\leq (1- \delta)\wh\lambda_1}\abs{\inner{\wh\bmr_\ell}{\wh\bmr^{(T)}}}^2 \\
\nonumber & \leq &  \frac{128\omega^2}{\delta^2\wh\rho_1^2}\sum_{\ell\colon \wh\lambda_\ell\leq (1- \delta)\wh\lambda_1}\abs{\inner{\wh\bmr_\ell}{\wh\bmr^{(T - 1)}}}^2 + 128\omega^2\varepsilon_{T^\prime}^2 \\
\nonumber & \leq & \left(\frac{128\omega^2}{\delta^2\wh\rho_1^2}\right)^2 \sum_{\ell\colon \wh\lambda_\ell\leq (1- \delta)\wh\lambda_1}\abs{\inner{\wh\bmr_\ell}{\wh\bmr^{(T - 2)}}}^2 + \{1 + \frac{128\omega^2}{\delta^2\wh\rho_1^2}\}128\omega^2\varepsilon_{T^\prime}^2 \\
\nonumber & \leq & \ldots \\
\nonumber & \leq &  \left(\frac{128\omega^2}{\delta^2\wh\rho_1^2}\right)^T\sum_{\ell\colon \wh\lambda_\ell\leq (1- \delta)\wh\lambda_1}\abs{\inner{\wh\bmr_\ell}{\wh\bmr^{(0)}}}^2 +  128\omega^2\varepsilon_{T^\prime}^2\sum_{t = 0}^T \left(\frac{128\omega^2}{\delta^2\wh\rho_1^2}\right)^t\\
  & \leq &   \left(\frac{128\omega^2}{\delta^2\wh\rho_1^2}\right)^T + \frac{128\omega^2\varepsilon_{T^\prime}^2}{\{1 - 128\omega^2/(\delta^2\wh\rho_1)^2\}}\\
\nonumber  & = &\left(\frac{128\omega^2}{\delta^2\wh\rho_1^2}\right)^T + \frac{512}{1 - 128\omega^2/(\delta^2\wh\rho_1)^2} \left(\frac{64\kappa^2}{\gamma^2\omega^2}\right)^{T^\prime}\defn \varepsilon_{T, T^\prime}.
\eeqr
which implies that
\beqrs
  \sum_{\ell\colon \wh\rho_\ell\leq (1- \delta)\wh\rho_1} \abs{\inner{\wh\bu_\ell^\pool}{\wh\bu^{(T)}}_{\bhSigx}}^2 + \sum_{\ell\colon \wh\rho_\ell\leq (1- \delta)\wh\rho_1} \abs{\inner{\wh\bv_\ell^\pool}{\wh\bv^{(T)}}_{\bhSigy}}^2 \\ \leq   2\left(\frac{128\omega^2}{\delta^2\wh\rho_1^2}\right)^T + 
    \frac{1024}{1 - 128\omega^2/(\delta^2\wh\rho_1)^2} \left(\frac{64\kappa^2}{\gamma^2\omega^2}\right)^{T^\prime}.
\eeqrs


We now bound the normalization step in Phase II of \Cref{alg:DistributedCCA_top}. 
Note that  
$$\sum_{\ell\colon \wh\lambda_\ell\leq (1- \delta)\wh\lambda_1}\abs{\inner{\wh\bmr_\ell}{\wh\bmr^{(T)}}}^2\leq \varepsilon_{T, T^\prime}.$$
It follows that,
\beqrs
(\wh\bu^{(T)})\trans\bhSigxy\wh\bv^{(T)} & = & (\wh\bmr^{(T)})\trans\wh\bC\wh\bmr^{(T)} = \sum_{i = 1}^d\wh\lambda_i \abs{\inner{\wh\bmr_i}{\wh\bmr^{(T)}}}^2\\
& = & \sum_{i\colon \wh\lambda_i> (1-\delta)\wh\lambda_1}\wh\lambda_i \abs{\inner{\wh\bmr_i}{\wh\bmr^{(T)}}}^2 + \sum_{i\colon \wh\lambda_i\leq (1-\delta)\wh\lambda_1}\wh\lambda_i \abs{\inner{\wh\bmr_i}{\wh\bmr^{(T)}}}^2\\
& \geq & (1-\delta)\wh\lambda_1\sum_{i\colon \wh\lambda_i> (1-\delta)\wh\lambda_1}\abs{\inner{\wh\bmr_i}{\wh\bmr^{(T)}}}^2 - (1-\delta)\wh\lambda_1\varepsilon_{T, T^\prime}\\
& \geq & (1-\delta)\wh\lambda_1(1 - \varepsilon_{T, T^\prime}) - (1-\delta)\wh\lambda_1\varepsilon_{T, T^\prime} = (1 - \delta)(1 - 2\varepsilon_{T, T^\prime})\wh\rho_1. 
\eeqrs
We have
\beqr\label{theorem:top_canonical_vector:eq2}
\nonumber  (\wh\bu_1^\dist)\trans\bhSigxy\wh\bv_1^\dist & = &  (\wh\bu^{(T)})\trans\bhSigxy\wh\bv^{(T)}/(\inner{\wh\bu^{(T)}}{\wh\bu^{(T)}}_{\bhSigx}\inner{\wh\bv^{(T)}}{\wh\bv^{(T)}}_{\bhSigy})^{1/2}\\
\nonumber  & \geq & (1 - \delta)(1 - 2\varepsilon_{T, T^\prime})\wh\rho_1 /(\inner{\wh\bu^{(T)}}{\wh\bu^{(T)}}_{\bhSigx}\inner{\wh\bv^{(T)}}{\wh\bv^{(T)}}_{\bhSigy})^{1/2}. \\
\eeqr

We assert that both $\inner{\wh\bu^{(T)}}{\wh\bu^{(T)}}_{\bhSigx}/2$ and $\inner{\wh\bv^{(T)}}{\wh\bv^{(T)}}_{\bhSigy}/2$ must be in the range of $[0.01, 0.99]$. Because $\inner{\wh\bu^{(T)}}{\wh\bu^{(T)}}_{\bhSigx}/2 + \inner{\wh\bv^{(T)}}{\wh\bv^{(T)}}_{\bhSigy}/2 = 1$, assume for instance that $\inner{\wh\bu^{(T)}}{\wh\bu^{(T)}}_{\bhSigx}/2 = c$ for some $c\in[0,1]$, which means $\inner{\wh\bv^{(T)}}{\wh\bv^{(T)}}_{\bhSigy}/2 = 1 - c$. Then, it follows that $(\inner{\wh\bu^{(T)}}{\wh\bu^{(T)}}_{\bhSigx}\inner{\wh\bv^{(T)}}{\wh\bv^{(T)}}_{\bhSigy})^{1/2} = \{4c(1 - c)\}^{1/2}$. This, in together with \eqref{theorem:top_canonical_vector:eq2}, implies that $(1 - \delta)(1 - 2\varepsilon_{T, T^\prime}) /\{4c(1 - c)\}^{1/2}\leq 1$. If $\delta<1/2$ and $\varepsilon_{T, T^\prime}\leq 1/4$, this implies that $c - 1/2\in [-(15)^{1/2}/8, (15)^{1/2}/8]$ and hence 
\beqr\label{c_range}
c \in [0.01, 0.99].
\eeqr 
We have
\beqrs
  \sum_{\ell\colon \wh\rho_\ell\leq (1- \delta)\wh\rho_1} \abs{\inner{\wh\bu_\ell^\pool}{\wh\bu_1^\dist}_{\bhSigx}}^2 + \sum_{\ell\colon \wh\rho_\ell\leq (1- \delta)\wh\rho_1}\abs{\inner{\wh\bv_\ell^\pool}{\wh\bv_1^\dist}_{\bhSigy}}^2  =  \\
      O\bigg\{\left(\frac{128\omega^2}{\delta^2\wh\rho_1^2}\right)^T + 
\frac{1}{1 - 128\omega^2/(\delta\wh\rho_1)^2} \left(\frac{64\kappa^2}{\gamma^2\omega^2}\right)^{T^\prime}\bigg\}.
  \eeqrs
This completes the proof of \Cref{theorem:top_canonical_vector}.
\end{proof}

\section{Proof of \Cref{coro:variability_top}}
\begin{proof}
Define 
$$
 \wh\bmr \defn \begin{pmatrix}
	\bhSigx^{1/2}\wh\bu^\dist_1\\
	\bhSigy^{1/2}\wh\bv^\dist_1\\
\end{pmatrix}\Bigg / 2^{1/2}.
$$
By the condition of \Cref{coro:variability_top}, we have 
\beqrs
\sum_{\ell\colon \wh\lambda_\ell\leq (1- \delta)\wh\lambda_1}\abs{\inner{\wh\bmr_\ell}{\wh\bmr}}^2 =  \sum_{\ell\colon \wh\rho_\ell\leq (1- \delta)\wh\rho_1} \abs{\inner{\wh\bu_\ell^\pool}{\wh\bu_1^\dist}_{\bhSigx}}^2/2 \\
+ \sum_{\ell\colon \wh\rho_\ell\leq (1- \delta)\wh\rho_1}\abs{\inner{\wh\bv_\ell^\pool}{\wh\bv_1^\dist}_{\bhSigy}}^2/2\leq \varepsilon/2.
\eeqrs	
It follows that
\beqrs
(\wh\bu_1^\dist)\trans\bhSigxy\wh\bv_1^\dist & = & \wh\bmr\trans\wh\bC\wh\bmr = \sum_{i = 1}^d\wh\lambda_i \abs{\inner{\wh\bmr_i}{\wh\bmr}}^2\\
& = & \sum_{i\colon \wh\lambda_i> (1-\delta)\wh\lambda_1}\wh\lambda_i \abs{\inner{\wh\bmr_i}{\wh\bmr}}^2 + \sum_{i\colon \wh\lambda_i\leq (1-\delta)\wh\lambda_1}\wh\lambda_i \abs{\inner{\wh\bmr_i}{\wh\bmr}}^2\\
& \geq & (1-\delta)\wh\lambda_1\sum_{i\colon \wh\lambda_i> (1-\delta)\wh\lambda_1}\abs{\inner{\wh\bmr_i}{\wh\bmr}}^2 - (1-\delta)\wh\lambda_1\varepsilon/2\\
& \geq & (1-\delta)\wh\lambda_1(1 - \varepsilon/2) - (1-\delta)\lambda_1\varepsilon/2 = (1 - \delta)(1 - \varepsilon)\wh\rho_1. 
\eeqrs
This completes the proof of \Cref{coro:variability_top}.
\end{proof}

\section{Proof of \Cref{corollary:top_canonical_vector:population}}

\begin{proof}[\bf Proof of \Cref{corollary:top_canonical_vector:population}]
	This is a special case of \Cref{corollary:topL_canonical_vector:population} when $L = 1$.
\end{proof}

\section{Proof of \Cref{theorem:topL_canonical_vector}}

\begin{proof}[\bf Proof of \Cref{theorem:topL_canonical_vector}]
We adapt the proof for Theorem F.1 in \cite{allen-zhu2017DoublyAcceleratedMethods} to our settings. Let $\widetilde\bphi_\ell\defn \bhSigx^{1/2}\wh\bu_{\ell}^\dist$, $\widetilde\bpsi_\ell\defn \bhSigy^{1/2}\wh\bv_{\ell}^\dist$, $\widetilde\bphi_\ell^\prime\defn \bhSigx^{1/2}\wh\bu_{\ell}^{\dist\prime}$, $\widetilde\bpsi_\ell^\prime\defn \bhSigy^{1/2}\wh\bv_{\ell}^{\dist\prime}$, $\widetilde\bmr_\ell \defn (\widetilde\bphi_\ell\trans, \widetilde\bpsi_\ell\trans)\trans/2^{1/2}$, $\widetilde \bmr_\ell^\prime \defn (\widetilde\bphi_\ell^{\prime\mbox{\tiny{T}}}, \widetilde\bpsi_\ell^{\prime\mbox{\tiny{T}}})\trans/2^{1/2}$ and 
\beqrs
  \widetilde \bR_\ell \defn  \left\{\widetilde \bR_{\ell - 1}, \begin{pmatrix}
  	\widetilde\bphi_\ell &  \widetilde\bphi_\ell\\
  	\widetilde\bpsi_\ell & - \widetilde\bpsi_\ell
  \end{pmatrix}\bigg / 2^{1/2} \right\}
\eeqrs
for $\ell = 1, \ldots, L$, where $\widetilde \bR_0 = ()$. It is easy to check that the projection in \Cref{alg:DistributedCCA_topL} takes the form of $\widetilde \bmr_\ell^{\prime\prime} = \widetilde\bmr_\ell^\prime -\widetilde \bR_{\ell -1}\widetilde \bR_{\ell - 1}\trans\widetilde\bmr_\ell^\prime$ and $\widetilde\bphi_\ell = \widetilde\bphi_\ell^{\prime\prime}/\norm{\widetilde\bphi_\ell^{\prime\prime}}$, $\widetilde\bpsi_\ell = \widetilde\bpsi_\ell^{\prime\prime}/\norm{\widetilde\bpsi_\ell^{\prime\prime}}$ where we write $\widetilde \bmr_\ell^{\prime\prime} = (\widetilde\bphi_\ell^{\prime\prime\mbox{\tiny{T}}}, \widetilde\bpsi_\ell^{\prime\prime\mbox{\tiny{T}}})\trans/2^{1/2}$ with $\widetilde\bphi_\ell^{\prime\prime}\in\mR^{d_\x}$ and $\widetilde\bpsi_\ell^{\prime\prime}\in\mR^{d_\y}$.

%

Let 
\beqrs
  \wh\bC^{(\ell -1)} = \begin{pmatrix}
  	\bm{0} & \wh\bT^{(\ell -1)}\\
  	(\wh\bT^{(\ell -1)})\trans & \bm{0}
  \end{pmatrix} = (\I - \widetilde \bR_{\ell - 1}\widetilde \bR_{\ell - 1}\trans)\begin{pmatrix}
  	\bm{0} & \wh\bT\\
  	\wh\bT\trans & \bm{0}
  \end{pmatrix}(\I - \widetilde \bR_{\ell - 1}\widetilde \bR_{\ell - 1}\trans).
\eeqrs
We claim that $\norm{\wh\bC^{(\ell -1)}}\geq \wh\rho_\ell$ for each $\ell = 1, \ldots, L$. This is indeed the direct result of Courant minimax principle. Specifically, $\wh\bC^{(\ell - 1)} = (\I - \widetilde \bR_{\ell - 1}\widetilde \bR_{\ell - 1}\trans)\wh\bC(\I - \widetilde \bR_{\ell - 1}\widetilde \bR_{\ell - 1}\trans)$  is a projection of $\wh\bC$ into a $d - \ell + 1$ dimensional space, and hence the largest eigenvalue $\norm{\wh\bC^{(\ell -1)}}$ would be at least as large as $\wh\rho_\ell$, which ensures our claim. 

Let  $\wh\rho = \norm{\wh\bC^{(L - 1)}}\geq \wh\rho_L$. Note that all column vectors in $\widetilde\bR_\ell$ are already  eigenvectors of $\wh\bC^{(\ell)}$ corresponding to eigenvalues zero. Let $\W_\ell$ be the column orthogonal matrix whose columns are eigenvectors in $\widetilde\bR_\ell^{\perp}$ of $\wh\bC^{(\ell)}$ with \textit{absolute} eigenvalues in the range $[0,(1 - \delta + \tau_\ell)\wh\rho]$, where $\tau_\ell \defn (\ell\delta)/(2L)$ for $0\leq \ell \leq L$. 

Denote by $\wh\bR_{\leq (1 - \delta)\wh\rho_L}$ the eigenvectors of $\wh\bC$ with \textit{absolute} eigenvalues less or equal to $(1 - \delta)\wh\rho_L$. We will show that for each $\ell = 0, \ldots, L$, there exists some matrix $\bQ_\ell$ such that
\beqr\label{theorem:topL_canonical_vector:eq2}
  \norm{\wh\bR_{\leq (1 - \delta)\wh\rho_L} - \W_\ell\bQ_\ell}\leq \omega_\ell\in[0,1)\ \text{and}\ \norm{\bQ_\ell}\leq 1
\eeqr
for some sequence $\{\omega_\ell\}_{\ell = 0}^L$ of small numbers. 
In fact, the first inequality above implies $\norm{\I - \wh\bR_{\leq (1 - \delta)\wh\rho_L}\trans\W_L\bQ_L}\leq \omega_L$. Therefore, the smallest singular value of $\wh\bR_{\leq (1 - \delta)\wh\rho_L}\trans\W_L\bQ_L$ is at least $ 1 - \omega_L>0$ by Weyl's inequality. This, together with the fact that $\norm{\bQ_L}\leq 1$, implies that $\sigma_{\min}(\wh\bR_{\leq (1 - \delta)\wh\rho_L}\trans\W_L)\geq 1 - \omega_L$, or equivalently,
\beqrs
  \I - \wh\bR_{\leq (1 - \delta)\wh\rho_L}\trans\W_L\W_L\trans\wh\bR_{\leq (1 - \delta)\wh\rho_L}\preceq 1 - (1 - \omega_L)^2\I,
\eeqrs
where $A\preceq B$ means $B - A$ is non-negative definite.  Because $\widetilde\bR_L$ is (column) orthogonal to $\W_L$, the proceeding display implies
\beqrs
  \wh\bR_{\leq (1 - \delta)\wh\rho_L}\trans\widetilde\bR_L\widetilde\bR_L\trans\wh\bR_{\leq (1 - \delta)\wh\rho_L} & \preceq  &  \wh\bR_{\leq (1 - \delta)\wh\rho_L}\trans(\I - \W_L\W_L\trans)\wh\bR_{\leq (1 - \delta)\wh\rho_L}\\
  & \preceq &  1 - (1 - \omega_L)^2\I\preceq 2\omega_L\I,
\eeqrs
which implies that $\norm{\wh\bR_{\leq (1 - \delta)\wh\rho_L}\trans\widetilde\bR_L}\leq (2\omega_L)^{1/2}$, and hence $\norm{\wh\U_{\leq (1 - \delta)\wh\rho_L}\trans\wh\bSig_\x\widetilde\U_L}^2\leq 2\omega_L$ and $\norm{\wh\V_{\leq (1 - \delta)\wh\rho_L}\trans\wh\bSig_\y\widetilde\V_L}^2\leq 2\omega_L$. 

It remains us to show the claim \eqref{theorem:topL_canonical_vector:eq2}. We will do this by induction. For $\ell = 0$, we simply take $\W_0 = \wh\bR_{\leq (1 - \delta)\wh\rho_L}$, $\omega_0 = 0$ and $\bQ_0 = \I$. Now suppose for each $\ell = 0, \ldots, L - 1$, there exists some orthogonal matrix $\bQ_\ell$ with $\norm{\bQ_\ell}\leq 1$ such that $\norm{\wh\bR_{\leq (1 - \delta)\wh\rho_L} - \W_\ell\bQ_\ell}\leq \omega_\ell\in[0,1)$. We construct $\bQ_{\ell + 1}$ as follows.


Since 
$\kappa = \max(\norm{\wh\bT_1 - \wh\bT})\geq \max( \norm{\wh\bT_1^{(\ell)} - \wh\bT^{(\ell)}})$, by the proof of  \Cref{theorem:top_canonical_vector}, i.e., equations \eqref{theorem:top_canonical_vector:medium_results} and \eqref{c_range}, with $\delta \leftarrow\delta /2$, we have (note that all column vectors in $\W_\ell$ and $\widetilde\bR_\ell$ corresponds to eigenvectors of $\wh\bC^{(\ell)}$ with \textit{absolute} eigenvalues less than or equal to $(1 - \delta + \tau_\ell)\norm{\wh\bC^{(L - 1)}}\leq (1-\delta/2)\norm{\wh\bC^{(\ell)}}$),
\beqrs
  \norm{\W_\ell\trans\widetilde\bmr_{\ell + 1}^{\prime}}^2\leq \omega_{T,T^\prime}^{(\ell)}\ \text{and} \ \norm{\widetilde\bR_\ell\trans\widetilde\bmr_{\ell + 1}^{\prime}}^2\leq \omega_{T,T^\prime}^{(\ell)},
\eeqrs
where $\omega_{T,T^\prime}^{(\ell)}\defn (\frac{128\omega^2}{\delta^2\wh\rho_\ell^2})^T +   \frac{512}{1 - 128\omega^2/(\delta\wh\rho_\ell)^2} (\frac{64\kappa^2}{\gamma^2\omega^2})^{T^\prime} \leq 1/2$. Now noting that $\widetilde\bmr_{\ell + 1}^{\prime\prime}$ is the projection of $\widetilde\bmr^\prime_{\ell + 1}$ into $\widetilde\bR_\ell^{\perp}$, we have
\beqr\label{W_l_rpp}
  \norm{\W_\ell\trans\widetilde\bmr_{\ell + 1}^{\prime\prime}}^2\leq \norm{\W_\ell\trans\widetilde\bmr_{\ell + 1}^\prime}^2\leq \omega_{T,T^\prime}^{(\ell)}.
\eeqr

Appealing to the similar techniques for obtaining \eqref{c_range}, we can get $\norm{\widetilde\bphi_\ell^{\prime\prime}}/2 = 1 - \norm{\widetilde\bpsi_\ell^{\prime\prime}}/2\in[0.01, 0.99]$. Note that the columns of $\W_\ell$ are all eigenvectors of $\wh\bC^{(\ell)}$ whose absolute eigenvalues are below some threshold, and hence must consist only symmetric vectors, that is, 
$$\W_\ell = \begin{pmatrix}
	\ba_1 & \ldots & \ba_t\\
	\pm\bb_1 & \ldots & \pm\bb_t
\end{pmatrix}.
$$
This, in together with \eqref{W_l_rpp}, implies that
\beqrs
  \norm{\W_\ell\trans\widetilde\bmr_{\ell + 1}^{\prime\prime}}^2 & = & \sum_{i = 1}^t (\widetilde\bphi_\ell^{\prime\prime\mbox{\tiny{T}}}\ba_i + \widetilde\bpsi_\ell^{\prime\prime\mbox{\tiny{T}}}\bb_i)^2/2 + (\widetilde\bphi_\ell^{\prime\prime\mbox{\tiny{T}}}\ba_i - \widetilde\bpsi_\ell^{\prime\prime\mbox{\tiny{T}}}\bb_i)^2/2\\
  & = &\norm{\widetilde\bphi_\ell^{\prime\prime\mbox{\tiny{T}}}(\ba_1,\ldots, \ba_t)}^2 + \norm{\widetilde\bpsi_\ell^{\prime\prime\mbox{\tiny{T}}}(\bb_1,\ldots, \bb_t)}^2\leq \omega_{T,T^\prime}^{(\ell)}.
 \eeqrs
 Now we have
 \beqrs
  \norm{\W_\ell\trans\widetilde\bmr_{\ell + 1}}^2 & = & \norm{\widetilde\bphi_\ell\trans(\ba_1,\ldots, \ba_t)}^2 + \norm{\widetilde\bpsi_\ell\trans(\bb_1,\ldots, \bb_t)}^2\\
  &\leq &\{\norm{\widetilde\bphi_\ell^{\prime\prime\mbox{\tiny{T}}}(\ba_1,\ldots, \ba_t)}^2 + \norm{\widetilde\bpsi_\ell^{\prime\prime\mbox{\tiny{T}}}(\bb_1,\ldots, \bb_t)}^2\}/0.02\leq \omega_{T,T^\prime}^{(\ell)}/0.02.
\eeqrs
Similarly, we can show that  
\beqrs
  \norm{\W_\ell\trans\begin{pmatrix}
  	  \widetilde\bphi_\ell\\
  	  - \widetilde\bpsi_\ell
  \end{pmatrix}\bigg / 2^{1/2}}^2 \leq \omega_{T,T^\prime}^{(\ell)}/0.02.
\eeqrs

Combining the above two facts, we have
\beqrs
  \norm{\W_\ell\trans\begin{pmatrix}
  	  \widetilde\bphi_\ell & \widetilde\bphi_\ell\\
  	  \widetilde\bphi_\ell & - \widetilde\bpsi_\ell
  \end{pmatrix}\bigg / 2^{1/2}}^2 &\leq& \norm{\W_\ell\trans\begin{pmatrix}
  	  \widetilde\bphi_\ell & \widetilde\bphi_\ell\\
  	  \widetilde\bphi_\ell & - \widetilde\bpsi_\ell
  \end{pmatrix}\bigg / 2^{1/2}}_F^2\\
  & = &  \norm{\W_\ell\trans\widetilde\bmr_{\ell + 1}}^2 + \norm{\W_\ell\trans\begin{pmatrix}
  	  \widetilde\bphi_\ell\\
  	  - \widetilde\bpsi_\ell
  \end{pmatrix}\bigg / 2^{1/2}}^2\\
  & = & 2\omega_{T,T^\prime}^{(\ell)}/0.02= 100\omega_{T,T^\prime}^{(\ell)}\defn \zeta_{T,T^\prime}^{(\ell)}.
\eeqrs

Now applying Lemma C.4 of \cite{allen-zhu2017DoublyAcceleratedMethods} with $M = \wh\bC^{(\ell)}$, $M^\prime =  \wh\bC^{(\ell + 1)}$, $U = \W_{\ell}$, $V = \W_{\ell + 1}$, $v = \begin{pmatrix}
  	  \widetilde\bphi_\ell & \widetilde\bphi_\ell\\
  	  \widetilde\bphi_\ell & - \widetilde\bpsi_\ell
  \end{pmatrix}\bigg / 2^{1/2}$, $\mu = (1 - \delta - \tau_\ell)\wh\rho$, $r = d-\ell$, and $\tau = (\tau_{\ell + 1} - \tau_{\ell})\wh\rho$, we have a matrix $\wt\bQ_{\ell}$ such that $\norm{\wt\bQ_{\ell}}\leq 1$ and 
\beqrs
  \norm{\W_\ell - \W_{\ell + 1}\wt\bQ_{\ell}}\leq \left\{\frac{507\wh\rho_1^2\zeta_{T,T^\prime}^{(\ell)}}{(\tau_{\ell + 1} - \tau_{\ell})^2\wh\rho^2} + 1.5\zeta_{T,T^\prime}^{(\ell)}\right\}^{1/2}\leq \frac{32\wh\rho_1 L (\zeta_{T,T^\prime}^{(\ell)})^{1/2}}{\wh\rho_L\delta}. 
\eeqrs
This, together with \eqref{theorem:topL_canonical_vector:eq2} implies
\beqrs
  \norm{\W_{\ell + 1}\wt\bQ_{\ell}\bQ_\ell - \wh\bR_{\leq (1 - \delta)\wh\rho_L}} & \leq &   \norm{\W_{\ell + 1}\wt\bQ_{\ell}\bQ_\ell - \W_\ell\bQ_\ell} +   \norm{\W_\ell\bQ_\ell - \wh\bR_{\leq (1 - \delta)\wh\rho_L}}\\
  & \leq & \frac{32\wh\rho_1 L (\zeta_{T,T^\prime}^{(\ell)})^{1/2}}{\wh\rho_L\delta} + \omega_\ell.
\eeqrs
Define $\bQ_{\ell + 1} \defn \wt\bQ_{\ell}\bQ_\ell$. We have $\norm{\bQ_{\ell + 1}}\leq 1$ and 
\beqrs
  \norm{\W_{\ell + 1}\bQ_{\ell + 1} - \wh\bR_{\leq (1 - \delta)\wh\rho_L}}\leq \omega_{\ell+ 1} & \defn  & \omega_\ell + 32\wh\rho_1 L (\zeta_{T,T^\prime}^{(\ell)})^{1/2}/(\wh\rho_L\delta) \\
  & = & \cdots\\
  & = & \sum_{j = 1}^{\ell}32\wh\rho_1 L (\zeta_{T,T^\prime}^{(j)})^{1/2}/(\wh\rho_L\delta)\\
  & \leq & (\ell + 1)32\wh\rho_1 L (\zeta_{T,T^\prime}^{(\ell)})^{1/2}/(\wh\rho_L\delta),
\eeqrs
for $\ell = 1, \ldots, (L - 1)$. In conclusion, we have
\beqrs
& & \norm{\W_{L}\bQ_{L} - \wh\bR_{\leq (1 - \delta)\wh\rho_L}}\\
& \leq & 32\wh\rho_1 L^2 (\zeta_{T,T^\prime}^{(L)})^{1/2}/(\wh\rho_L\delta)\\
& = & \frac{320\wh\rho_1 L^2}{\wh\rho_L\delta}\left\{\left(\frac{128\omega^2}{\delta^2\wh\rho_L^2}\right)^T +   \frac{512}{1 - 128\omega^2/(\delta\wh\rho_L)^2} \left(\frac{64\kappa^2}{\gamma^2\omega^2}\right)^{T^\prime}\right\}^{1/2}.
\eeqrs
We complete the proof of \Cref{theorem:topL_canonical_vector}.
\end{proof}

\section{Proof of \Cref{corollary:variability_topL}}
\begin{proof}
	We can establish this corollary by following the same proof for establishing Theorem  F.1 in \cite{allen-zhu2017DoublyAcceleratedMethods}.
\end{proof}

\section{Proof of \Cref{corollary:topL_canonical_vector:population}}
\begin{proof}
	Note that $\wh\U_{> (1 - \delta)\wh\rho_L}\wh\U_{> (1 - \delta)\wh\rho_L}\trans + \wh\U_{\leq (1 - \delta)\wh\rho_L}\wh\U_{\leq (1 - \delta)\wh\rho_L}\trans = \bhSig_\x^{-1}$. We have
\beqrs
& & \Norm{\widetilde\U_L\trans\wh\bSig_\x\U_{\leq (1 - 2\delta)\rho_L}} \\
& = &  \Norm{\widetilde\U_L\trans\wh\bSig_\x(\wh\U_{> (1 - \delta)\wh\rho_L}\wh\U_{> (1 - \delta)\wh\rho_L}\trans + \wh\U_{\leq (1 - \delta)\wh\rho_L}\wh\U_{\leq (1 - \delta)\wh\rho_L}\trans)\wh\bSig_\x\U_{\leq (1 - 2\delta)\rho_L}} \\
& \leq & \Norm{\widetilde\U_L\trans\wh\bSig_\x\wh\U_{> (1 - \delta)\wh\rho_L}\wh\U_{> (1 - \delta)\wh\rho_L}\trans\wh\bSig_\x\U_{\leq (1 - 2\delta)\rho_L}} \\
& & \hspace*{2cm} +\Norm{\widetilde\U_L\trans\wh\bSig_\x\wh\U_{\leq (1 - \delta)\wh\rho_L}\wh\U_{\leq (1 - \delta)\wh\rho_L}\trans\wh\bSig_\x\U_{\leq (1 - 2\delta)\rho_L}}\\
& \leq & \Norm{\wh\U_{> (1 - \delta)\wh\rho_L}\trans\wh\bSig_\x\U_{\leq (1 - 2\delta)\rho_L}} + \varepsilon\\
& \leq & \frac{\norm{\wh\bT - \bT}}{(1 - \delta) (\wh\rho_L - \rho_L) + \delta\rho_L}  + \varepsilon,
\eeqrs
where the last inequality follows from the gap-free Wedin Theorem \citep[Lemma C.3]{allen-zhu2017DoublyAcceleratedMethods}. By the similar arguments, we can show the same bound for $\Norm{\widetilde\V_L\trans\wh\bSig_\y\V_{\leq (1 - 2\delta)\rho_L}}$. This completes the proof of \Cref{corollary:topL_canonical_vector:population}. 
\end{proof}

\section{Proof of Auxiliary Lemmas}

\subsection{Proof of Lemma \ref{lemma:outer}}

\begin{proof}
It follows that $\norm{\bmr} = \norm{\bmtr} = 1$. Let $\beps = \bmtr - \M_{\overline{\rho}_1}\bmr$.  By \eqref{lemma:outer:residual}, we have $\norm{\beps} \leq \varepsilon$. Denote by 
	\beqrs
	\wh\mu_1 \defn (\overline{\rho}_1 - \wh\rho_1)^{-1}\geq \wh\mu_2 \defn (\overline{\rho}_1 - \wh\rho_2)^{-1}\geq \ldots\geq \wh\mu_{d - 2}\defn (\overline{\rho}_1 + \wh\rho_{2})^{-1} \geq \wh\mu_d\defn (\overline{\rho}_1 + \wh\rho_{1})^{-1},
	\eeqrs
	with corresponding eigenvectors $\wh\bmr_1, \ldots, \wh\bmr_d$.
Because $\M_{\overline{\rho}_1} = (\overline{\rho}_1\I - \bC)^{-1} = \sum_{\ell = 1}^d\wh\mu_\ell\wh\bmr_\ell\wh\bmr_\ell\trans$, we have, for each $1\leq \ell \leq d$,
\beqr\label{lemma:outer:eq1}
\inner{\wh\bmr_\ell}{\bmtr} = \wh\mu_\ell\inner{\wh\bmr_\ell}{\bmr} + \inner{\wh\bmr_\ell}{\beps}.
\eeqr
This ensures a lower bound on $\norm{\bmtr}$,
\beqr\label{lemma:outer:lower_bound}
\nonumber \norm{\bmtr}^2 & = & \sum_{\ell = 1}^d \abs{\inner{\wh\bmr_\ell}{\bmtr}}^2  \geq  \sum_{\ell\colon \wh\lambda_\ell>(1 - \delta)\wh\lambda_1} \abs{\inner{\wh\bmr_\ell}{\bmtr}}^2 \\
\nonumber & \geq & \sum_{\ell\colon \wh\lambda_\ell>(1 - \delta)\wh\lambda_1} \left\{(\overline{\rho}_1 - \wh\rho_\ell)^{-2}\abs{\inner{\wh\bmr_\ell}{\bmr}}^2/2 + 
\abs{\inner{\wh\bmr_\ell}{\beps}}^2\right\}\\
& \geq & (32\omega^2)^{-1} - \norm{\beps}^2 \geq (32\omega^2)^{-1} - \varepsilon^2 \geq (64\omega^2)^{-1},
\eeqr
where the first inequality follows the fact that $\abs{\inner{\wh\bmr_\ell}{\bmtr}}^2\geq 0$, and the second one is due to \eqref{lemma:outer:eq1} and $(a+b)^2\geq 0$ for any $a,b\in\mR$, and to prove the last one we used the upper bound of $\wb\rho_1 - \wh\rho_\ell\leq 2\omega$ for each $\ell$ such that  $\wh\lambda_\ell>(1 - \delta)\wh\lambda_1$ and the conditions \eqref{lemma:outer:eigen} and \eqref{lemma:outer:initial} in \Cref{lemma:outer}.

 On the other hand, for each $\ell$ such that $\wh\lambda_\ell\leq (1 - \delta)\wh\lambda_1$, we have $\overline{\rho}_1 - \wh\lambda_\ell\geq \wh\rho_1 - \wh\lambda_\ell\geq \delta\wh\rho_1$. This implies
\beqr\label{lemma:outer:eq2}
\abs{\inner{\wh\bmr_\ell}{\bmtr}}  \leq (\delta\wh\rho_1)^{-1}\abs{\inner{\wh\bmr_\ell}{\bmr}} + \abs{\inner{\wh\bmr_\ell}{\beps}} \leq (\delta\wh\rho_1)^{-1}\abs{\inner{\wh\bmr_\ell}{\bmr}} + \varepsilon.
\eeqr
This, together with \eqref{lemma:outer:eq1}, ensures that
\beqrs
\abs{\inner{\wh\bmr_\ell}{\bmtr}}/\norm{\bmtr} & \leq & \{8\omega/(\delta\wh\rho_1)\}\abs{\inner{\wh\bmr_\ell}{\bmr}}/\norm{\bmr} + 8\omega\varepsilon
\eeqrs
The second claim follows by combining \eqref{lemma:outer:lower_bound} with 
\beqrs
\sum_{\ell\colon \wh\lambda_\ell\leq (1 - \delta)\wh\lambda_1} \abs{\inner{\wh\bmr_\ell}{\bmtr}}^2 & \leq &  2(\delta\wh\rho_1)^{-2}\sum_{\ell\colon \wh\lambda_\ell\leq (1 - \delta)\wh\lambda_1} \abs{\inner{\wh\bmr_\ell}{\bmr}}^2 + 2\sum_{\ell\colon \wh\lambda_\ell\leq (1 - \delta)\wh\lambda_1}\abs{\inner{\wh\bmr_\ell}{\beps}}^2\\
& \leq & 2(\delta\wh\rho_1)^{-2}\sum_{\ell\colon \wh\lambda_\ell\leq (1 - \delta)\wh\lambda_1} \abs{\inner{\wh\bmr_\ell}{\bmr}}^2 + 2\varepsilon^2,
\eeqrs
where the first inequality is due to $(a + b)^2\leq 2a^2 + 2b^2$ and \eqref{lemma:outer:eq2}, and the second inequality is because $\norm{\beps}\leq \varepsilon$.
\end{proof}
\subsection{Proof of Lemma \ref{lemma:inner}}
\begin{proof}
It follows that
\beqrs
\bmtr^{(t +1)} = \wh\bmr^{(t +1)}_{j} - \M_{\rhob_1}(\M_{\rhob_1}^{-1}\wh\bmr^{(t + 1)}_j - \wh\bmr^{(t)}),\\
\wh\bmr^{(t +1)}_{j + 1} = \wh\bmr^{(t +1)}_{j} - \M_{\rhob_1, 1}(\M_{\rhob_1}^{-1}\wh\bmr^{(t + 1)}_j - \wh\bmr^{(t)}).
\eeqrs
Taking the difference of the equations above, we have
\beqrs
  \norm{\wh\bmr^{(t +1)}_{j + 1} - \bmtr^{(t +1)}} & = & \norm{(\M_{\rhob_1} - \M_{\rhob_1, 1})(\M_{\rhob_1}^{-1}\wh\bmr^{(t + 1)}_j - \wh\bmr^{(t)})}\\
  & = & \norm{(\I - \M_{\rhob_1, 1}\M_{\rhob_1}^{-1})(\wh\bmr^{(t + 1)}_j - \M_{\rhob_1}\wh\bmr^{(t)})}\\
  & \leq & \norm{\I - \M_{\rhob_1, 1}\M_{\rhob_1}^{-1}}\norm{\bmr^{(t + 1)}_j - \bmtr^{(t + 1)}}.
\eeqrs
The first term on the right-hand side of the above display is bounded as 
\beqrs
  & & \norm{\I - \M_{\rhob_1, 1}\M_{\rhob_1}^{-1}}\\
  & = & \Norm{\I - \begin{pmatrix}
  	\bhSigx^{1/2} & \\
  	& \bhSigy^{1/2}
  \end{pmatrix}
    \begin{pmatrix}
  	\rhob_1\bhSig_{\x,1} & -\bhSig_{\x, \y,1} \\
  	-\bhSig_{\x, \y,1}\trans & \rhob_1\bhSig_{\y,1}
  \end{pmatrix}^{-1}
    \begin{pmatrix}
  	\rhob_1\bhSig_{\x} & -\bhSig_{\x, \y} \\
  	-\bhSig_{\x, \y}\trans & \rhob_1\bhSig_{\y}
  \end{pmatrix}
      \begin{pmatrix}
  	\bhSigx^{1/2} & \\
  	& \bhSigy^{1/2}
  \end{pmatrix}^{-1}}\\
  & = & \Norm{\I - 
      \begin{pmatrix}
  	\rhob_1\bhSig_{\x,1} & -\bhSig_{\x, \y,1} \\
  	-\bhSig_{\x, \y,1}\trans & \rhob_1\bhSig_{\y,1}
  \end{pmatrix}^{-1}
    \begin{pmatrix}
  	\rhob_1\bhSig_{\x} & -\bhSig_{\x, \y} \\
  	-\bhSig_{\x, \y}\trans & \rhob_1\bhSig_{\y}
  \end{pmatrix}
  }\\
  & = & \norm{\I - 
      \bhH_1^{-1}\bhH}\leq  \norm{\bhH_1^{-1}}\norm{\bhH_1 - \bhH}\leq  \{2/(\gamma\omega)\}\norm{\bhH_1 - \bhH},
\eeqrs
with probability approaching one, where the last inequality is by the fact that 
\beqrs
  \norm{\bhH_1^{-1}}^{-1} & = & \lambda_{\min}(\bhH_1) = \lambda_{\min}\left\{\begin{pmatrix}
  	\rhob_1\bhSig_{\x,1} & -\bhSig_{\x, \y,1} \\
  	-\bhSig_{\x, \y,1}\trans & \rhob_1\bhSig_{\y,1}
  \end{pmatrix}\right\}\\
  & = & \lambda_{\min}\left\{
  \begin{pmatrix}
  	\bhSig_{\x, 1}^{1/2} & \\
  	& \bhSig_{\y,1}^{1/2}
  \end{pmatrix}\begin{pmatrix}
  	\rhob_1\I & -\wh\bT_1 \\
  	-\wh\bT_1\trans & \rhob_1\I
  \end{pmatrix}
  \begin{pmatrix}
  	\bhSig_{\x, 1}^{1/2} & \\
  	& \bhSig_{\y,1}^{1/2}
  \end{pmatrix}
  \right\}\\
  & \geq & \gamma/2
  \lambda_{\min}\left\{
  \begin{pmatrix}
  	\rhob_1\I & -\wh\bT_1 \\
  	-\wh\bT_1\trans & \rhob_1\I
  \end{pmatrix}
  \right\}\\
  & \geq & \gamma/2\{\lambda_{\min}(\rhob_1\I - \wh\bC) - \norm{\wh\bT_1 - \wh\bT}\}\\
  &\geq & \gamma/2(\omega - \norm{\wh\bT_1 - \wh\bT})\geq \gamma\omega/4,
\eeqrs
and 
\beqrs
& & \norm{\bhH_1 - \bhH} \defn \kappa \\
& = &  \Norm{\begin{pmatrix}
  	\rhob_1\bhSig_{\x,1} & -\bhSig_{\x, \y,1} \\
  	-\bhSig_{\x, \y,1}\trans & \rhob_1\bhSig_{\y,1}
  \end{pmatrix} - \begin{pmatrix}
  	\rhob_1\bhSig_{\x} & -\bhSig_{\x, \y} \\
  	-\bhSig_{\x, \y}\trans & \rhob_1\bhSig_{\y}
  \end{pmatrix}}\\
 & \leq & \Norm{\begin{pmatrix}
  	\rhob_1\bhSig_{\x,1} &  \\
  	 & \rhob_1\bhSig_{\y,1}
  \end{pmatrix} - \begin{pmatrix}
  	\rhob_1\bhSig_{\x} &  \\
  	 & \rhob_1\bhSig_{\y}
  \end{pmatrix}}\\
  & &  + 
  \Norm{\begin{pmatrix}
  	  & -\bhSig_{\x, \y,1} \\
  	-\bhSig_{\x, \y,1}\trans &  
  \end{pmatrix} - \begin{pmatrix}
  	  & -\bhSig_{\x, \y} \\
  	-\bhSig_{\x, \y}\trans &  
  \end{pmatrix}}\\
  & \leq & 2\max(\norm{\bhSigx - \bhSig_{\x,1}}, \norm{\bhSigy - \bhSig_{\y,1}}, \norm{\bhSigxy - \bhSig_{\x, \y,1}}).
\eeqrs

%
\end{proof}

\vskip 0.2in
\bibliography{distCCA_refs_zotero}

\end{document}

%% file: definitionchen.tex
\usepackage[final]{microtype}

\usepackage{amssymb,
amsthm,
amsmath,mathrsfs,bm,color}
\usepackage{graphicx}
\usepackage{subcaption}
\usepackage{multirow}
\usepackage{booktabs}
\usepackage{epsfig}
\usepackage{natbib}
\usepackage{setspace}
\usepackage{enumerate}
\usepackage{threeparttable}

\usepackage{url}
\usepackage{doi}
\usepackage{algorithm,algorithmic}
\numberwithin{equation}{section}

\newcommand{\defn}{\stackrel{\mbox{{\tiny def}}}{=}}
\newcommand{\trans}{^{\mbox{\tiny{T}}}}
\newtheorem{theo}{\bf Theorem}
\newtheorem{lemm}{\bf Lemma}

\def\I{{\bf I}}

\def\x{{\bf x}}
\def\y{{\bf y}}

\def\calD{\mathcal{D}}

\def\wh{\widehat}
\def\wb{\overline}
\def\wt{\widetilde}
\def\cov{\textrm{cov}}

\def\tr{\textrm{tr}}
\def\dist{\mathrm{dist}}

\def\beqr{\begin{eqnarray}}
\def\eeqr{\end{eqnarray}}
\def\beqrs{\begin{eqnarray*}}
\def\eeqrs{\end{eqnarray*}}
\def\bb{\mbox{\boldmath$\beta$}}
\def\ba{\mbox{\boldmath$\alpha$}}

\def\bSig{\mathbf{\Sigma}}
\def\beps{\mbox{\boldmath$\varepsilon$}}

\def\mR{\mathbb{R}}

\newcommand{\M}{\mathbf{M}}

\newcommand{\X}{\mathbf{X}}

\newcommand{\diag}{\mathrm{diag}}

\newcommand{\strans}{^{*\mbox{\tiny{T}}}}

\newcommand{\bPhi}{\mathbf{\Phi}}

\newcommand{\bT}{\mathbf{T}}

\newcommand{\Y}{\mathbf{Y}}

\newcommand{\B}{\mathbf{B}}
\newcommand{\bC}{\mathbf{C}}
\newcommand{\V}{\mathbf{V}}
\newcommand{\U}{\mathbf{U}}
\newcommand{\D}{\mathbf{D}}
\newcommand{\bv}{\mathbf{v}}
\newcommand{\bu}{\mathbf{u}}
\newcommand{\A}{\mathbf{A}}

\usepackage{url}
\DeclareFontFamily{U}{mathx}{\hyphenchar\font45}
\DeclareFontShape{U}{mathx}{m}{n}{
	<5> <6> <7> <8> <9> <10>
	<10.95> <12> <14.4> <17.28> <20.74> <24.88>
	mathx10
}{}
\DeclareSymbolFont{mathx}{U}{mathx}{m}{n}
\DeclareFontSubstitution{U}{mathx}{m}{n}
\DeclareMathAccent{\widecheck}{0}{mathx}{"71}
\DeclareMathAccent{\wideparen}{0}{mathx}{"75}

\newcommand{\bhSig}{\widehat{\bSig}}

\newcommand{\argmin}{\mathop{\rm arg\min}}
\newcommand{\argmax}{\mathop{\rm arg\max}}

\def\Z{{\bf Z}}

\newcommand{\pool}{\mathrm{pool}}
\newcommand{\norm}[1]{\Vert#1\Vert}
\newcommand{\Norm}[1]{\left\Vert#1\right\Vert}
\newcommand{\mnorm}[1]{{\vert\kern-0.25ex\vert\kern-0.25ex\vert #1 
    \vert\kern-0.25ex\vert\kern-0.25ex\vert}}
\newcommand{\Mnorm}[1]{{\left\vert\kern-0.25ex\left\vert\kern-0.25ex\left\vert #1 
    \right\vert\kern-0.25ex\right\vert\kern-0.25ex\right\vert}}
\newcommand{\inner}[2]{\langle #1, #2 \rangle}
\newcommand{\Inner}[2]{\left\langle #1, #2 \right\rangle}
\newcommand{\abs}[1]{\vert#1\vert}
\newcommand{\Abs}[1]{\left\vert#1\right\vert}

\newtheorem{coro}{\bf Corollary}

\usepackage{caption}

\usepackage{tikz}

\definecolor{myblue}{HTML}{0072BD}
\definecolor{mygreen}{HTML}{77AC30}
\definecolor{myred}{HTML}{D95319}
\definecolor{mypurple}{HTML}{7E2F8E}
\DeclareRobustCommand\full  {\tikz[baseline=-0.6ex]{\draw[thick,myred] (0.2,0)--(0.3,0);
\draw[thick,myred] (0.25,-0.15)--(0.25,0.15);
\draw[thick,myred] (0.15,-0.1)--(0.35,0.1);
\draw[thick,myred] (0.15,0.1)--(0.35,-0.1);
\draw[thick,myred] (0,0)--(0.5,0);}}
\DeclareRobustCommand\dottedcircle{\tikz[baseline=-0.6ex]{\draw[thick, densely  dotted,mypurple](0,0)--(0.5,0); \draw[mypurple, solid, thick] (0.25,0) circle (1.2mm);}}

\DeclareRobustCommand\dashed{\tikz[baseline=-0.6ex]{\draw[thick,solid,myblue](2.mm,-1.2mm) rectangle (4mm,1.2mm);\draw[thick,dashed,myblue] (0,0)--(0.54,0);}}
\DeclareRobustCommand\chain {\tikz[baseline=-0.6ex]{\draw[thick,solid,mygreen] (0.35em,0.2em) --(0.95em,0.2em) -- (0.65em,-0.4em) -- (0.35em,0.2em);\draw[thick,dash dot dot,mygreen] (0,0)--(0.54,0);}}

\newcommand{\bphi}{\bm{\phi}}
\newcommand{\bpsi}{\bm{\psi}}
\newcommand{\bPsi}{\bm{\Psi}}

\newcommand{\bH}{\mathbf{H}}
\newcommand{\bhH}{\wh\bH}
\newcommand{\g}{\mathbf{g}}
\newcommand{\W}{\mathbf{W}}

\newcommand{\bhv}{\wh\bv}
\newcommand{\bSigx}{\bSig_{\x}}
\newcommand{\bSigy}{\bSig_{\y}}
\newcommand{\bSigxy}{\bSig_{\x, \y}}
\newcommand{\bhSigx}{\widehat{\bSig}_{\x}}
\newcommand{\bhSigy}{\widehat{\bSig}_{\y}}
\newcommand{\bhSigxy}{\widehat{\bSig}_{\x, \y}}
\newcommand{\rhob}{\overline{\rho}}

\newcommand{\bmr}{\bm{r}}
\newcommand{\bmtr}{\wt\bmr}
\newcommand{\bQ}{\mathbf{Q}}
\newcommand{\bR}{\mathbf{R}}

\usepackage{cleveref}
\crefname{theo}{Theorem}{Theorems}
\crefname{lemm}{Lemma}{Lemmas}
\crefname{coro}{corollary}{Corollaries}
\usepackage{xargs}                      
\usepackage[colorinlistoftodos,prependcaption,textsize=tiny]{todonotes}
\newcommandx{\unsure}[2][1=]{\todo[linecolor=red,backgroundcolor=red!25,bordercolor=red,#1]{#2}}
\newcommandx{\change}[2][1=]{\todo[linecolor=blue,backgroundcolor=blue!25,bordercolor=blue,#1]{#2}}
\newcommandx{\info}[2][1=]{\todo[linecolor=OliveGreen,backgroundcolor=OliveGreen!25,bordercolor=OliveGreen,#1]{#2}}
\newcommandx{\improvement}[2][1=]{\todo[linecolor=Plum,backgroundcolor=Plum!25,bordercolor=Plum,#1]{#2}}
\newcommandx{\thiswillnotshow}[2][1=]{\todo[disable,#1]{#2}}

%% file: distCCA_JMLR.bbl
\begin{thebibliography}{38}
\providecommand{\natexlab}[1]{#1}
\providecommand{\url}[1]{\texttt{#1}}
\expandafter\ifx\csname urlstyle\endcsname\relax
  \providecommand{\doi}[1]{doi: #1}\else
  \providecommand{\doi}{doi: \begingroup \urlstyle{rm}\Url}\fi

\bibitem[Allen-Zhu and Li(2016)]{allen-zhu2016LazySVDEvenFaster}
Zeyuan Allen-Zhu and Yuanzhi Li.
\newblock Lazysvd: Even faster svd decomposition yet without agonizing pain.
\newblock In D.~Lee, M.~Sugiyama, U.~Luxburg, I.~Guyon, and R.~Garnett,
  editors, \emph{Advances in Neural Information Processing Systems}, volume~29.
  Curran Associates, Inc., 2016.
\newblock URL
  \url{https://proceedings.neurips.cc/paper/2016/file/c6e19e830859f2cb9f7c8f8cacb8d2a6-Paper.pdf}.

\bibitem[Allen-Zhu and Li(2017)]{allen-zhu2017DoublyAcceleratedMethods}
Zeyuan Allen-Zhu and Yuanzhi Li.
\newblock Doubly accelerated methods for faster {CCA} and generalized
  eigendecomposition.
\newblock In Doina Precup and Yee~Whye Teh, editors, \emph{Proceedings of the
  34th International Conference on Machine Learning}, volume~70 of
  \emph{Proceedings of Machine Learning Research}, pages 98--106. PMLR, 06--11
  Aug 2017.
\newblock URL \url{https://proceedings.mlr.press/v70/allen-zhu17b.html}.

\bibitem[Anderson(1999)]{anderson1999AsymptoticTheoryCanonical}
T.~W. Anderson.
\newblock Asymptotic {{Theory}} for {{Canonical Correlation Analysis}}.
\newblock \emph{Journal of Multivariate Analysis}, 70\penalty0 (1):\penalty0
  1--29, July 1999.

\bibitem[Anderson(2003)]{anderson2003IntroductionMultivariateStatistical}
T.~W. Anderson.
\newblock \emph{An Introduction to Multivariate Statistical Analysis}.
\newblock Wiley Series in Probability and Statistics. {Wiley-Interscience},
  {Hoboken, N.J}, 3rd ed edition, 2003.

\bibitem[Bach and Jordan(2005)]{bach2005ProbabilisticInterpretationCanonical}
Francis Bach and Michael Jordan.
\newblock A probabilistic interpretation of canonical correlation analysis.
\newblock Technical Report 688, {University of California}, {Dept. of
  Statistics}, May 2005.

\bibitem[Bao et~al.(2019)Bao, Hu, Pan, and
  Zhou]{bao2014CanonicalCorrelationCoefficients}
Zhigang Bao, Jiang Hu, Guangming Pan, and Wang Zhou.
\newblock Canonical correlation coefficients of high-dimensional {{Gaussian}}
  vectors: {{Finite}} rank case.
\newblock \emph{The Annals of Statistics}, 47\penalty0 (1):\penalty0 612--640,
  February 2019.

\bibitem[Cai and Zhang(2018)]{cai2018RateoptimalPerturbationBounds}
T.~Tony Cai and Anru Zhang.
\newblock Rate-optimal perturbation bounds for singular subspaces with
  applications to high-dimensional statistics.
\newblock \emph{The Annals of Statistics}, 46\penalty0 (1):\penalty0 60--89,
  February 2018.

\bibitem[Chen et~al.(2022)Chen, Lee, Li, and
  Yang]{chen2021DistributedEstimationPrincipal}
Xi~Chen, Jason~D. Lee, He~Li, and Yun Yang.
\newblock Distributed {{Estimation}} for {{Principal Component Analysis}}: {{An
  Enlarged Eigenspace Analysis}}.
\newblock \emph{Journal of the American Statistical Association}, 117\penalty0
  (540):\penalty0 1775--1786, October 2022.

\bibitem[Chen et~al.(2019)Chen, Yang, Li, and
  Zhao]{chen2019DroppingConvexityMore}
Zhehui Chen, Lin~F. Yang, Chris~J. Li, and Tuo Zhao.
\newblock Dropping convexity for more efficient and scalable online multiview
  learning.
\newblock \emph{arXiv:1702.08134}, September 2019.

\bibitem[Cheng et~al.(2021)Cheng, Wei, and
  Chen]{cheng2021TacklingSmallEigengaps}
Chen Cheng, Yuting Wei, and Yuxin Chen.
\newblock Tackling {{Small Eigen-Gaps}}: {{Fine-Grained Eigenvector
  Estimation}} and {{Inference Under Heteroscedastic Noise}}.
\newblock \emph{IEEE Transactions on Information Theory}, 67\penalty0
  (11):\penalty0 7380--7419, November 2021.

\bibitem[Dhillon et~al.(2011)Dhillon, Foster, and
  Ungar]{dhillon2011MultiviewLearningWord}
Paramveer Dhillon, Dean~P Foster, and Lyle Ungar.
\newblock Multi-view learning of word embeddings via {{CCA}}.
\newblock In J.~{Shawe-Taylor}, R.~Zemel, P.~Bartlett, F.~Pereira, and K.~Q.
  Weinberger, editors, \emph{Advances in Neural Information Processing
  Systems}, volume~24. {Curran Associates, Inc.}, 2011.

\bibitem[Dua and Graff(2017)]{Dua:2019}
Dheeru Dua and Casey Graff.
\newblock {{UCI}} machine learning repository, 2017.

\bibitem[Fan et~al.(2014)Fan, Han, and Liu]{fan2014ChallengesBigData}
Jianqing Fan, Fang Han, and Han Liu.
\newblock Challenges of big data analysis.
\newblock \emph{National Science Review}, 1\penalty0 (2):\penalty0 293--314,
  June 2014.

\bibitem[Fan et~al.(2019)Fan, Wang, Wang, and
  Zhu]{fan2019DistributedEstimationPrincipal}
Jianqing Fan, Dong Wang, Kaizheng Wang, and Ziwei Zhu.
\newblock Distributed estimation of principal eigenspaces.
\newblock \emph{The Annals of Statistics}, 47\penalty0 (6):\penalty0
  3009--3031, December 2019.

\bibitem[Fan et~al.(2021)Fan, Guo, and
  Wang]{fan2021CommunicationEfficientAccurateStatistical}
Jianqing Fan, Yongyi Guo, and Kaizheng Wang.
\newblock Communication-{{Efficient Accurate Statistical Estimation}}.
\newblock \emph{Journal of the American Statistical Association}, pages 1--11,
  September 2021.

\bibitem[Gao et~al.(2015)Gao, Ma, Ren, and
  Zhou]{gao2015MinimaxEstimationSparse}
Chao Gao, Zongming Ma, Zhao Ren, and Harrison~H. Zhou.
\newblock Minimax estimation in sparse canonical correlation analysis.
\newblock \emph{The Annals of Statistics}, 43\penalty0 (5):\penalty0
  2168--2197, 2015.

\bibitem[Gao et~al.(2017)Gao, Ma, and Zhou]{gao2017SparseCCAAdaptive}
Chao Gao, Zongming Ma, and Harrison~H. Zhou.
\newblock Sparse {{CCA}}: Adaptive estimation and computational barriers.
\newblock \emph{The Annals of Statistics}, 45\penalty0 (5):\penalty0
  2074--2101, October 2017.

\bibitem[Gao et~al.(2019)Gao, Garber, Srebro, Wang, and
  Wang]{gao2019StochasticCanonicalCorrelation}
Chao Gao, Dan Garber, Nathan Srebro, Jialei Wang, and Weiran Wang.
\newblock Stochastic {{Canonical Correlation Analysis}}.
\newblock \emph{Journal of Machine Learning Research}, 20\penalty0
  (167):\penalty0 1--46, 2019.

\bibitem[Garber and Hazan(2015)]{garber2015FastSimplePCA}
Dan Garber and Elad Hazan.
\newblock Fast and simple {{PCA}} via convex optimization.
\newblock \emph{arXiv:1509.05647}, November 2015.

\bibitem[Garber et~al.(2017)Garber, Shamir, and
  Srebro]{garber2017CommunicationefficientAlgorithmsDistributeda}
Dan Garber, Ohad Shamir, and Nathan Srebro.
\newblock Communication-efficient algorithms for distributed stochastic
  principal component analysis.
\newblock In \emph{International {{Conference}} on {{Machine Learning}}}, pages
  1203--1212. {PMLR}, July 2017.

\bibitem[Golub and Van~Loan(1983)]{golub1983MatrixComputations}
Gene~H. Golub and Charles~F. Van~Loan.
\newblock \emph{Matrix {{Computations}}}.
\newblock Number~3 in Johns {{Hopkins}} Series in the Mathematical Sciences.
  {Johns Hopkins University Press}, {Baltimore}, 3 edition, 1983.

\bibitem[Golub and Zha(1995)]{golub1995CanonicalCorrelationsMatrix}
Gene~H. Golub and Hongyuan Zha.
\newblock The canonical correlations of matrix pairs and their numerical
  computation.
\newblock In Avner Friedman, Willard Miller, Adam Bojanczyk, and George
  Cybenko, editors, \emph{Linear {{Algebra}} for {{Signal Processing}}},
  volume~69, pages 27--49. {Springer New York}, {New York, NY}, 1995.

\bibitem[Hardoon et~al.(2004)Hardoon, Szedmak, and
  {Shawe-Taylor}]{hardoon2004CanonicalCorrelationAnalysis}
David~R. Hardoon, Sandor Szedmak, and John {Shawe-Taylor}.
\newblock Canonical correlation analysis: An overview with application to
  learning methods.
\newblock \emph{Neural Computation}, 16\penalty0 (12):\penalty0 2639--2664,
  December 2004.

\bibitem[Hestenes and Stiefel(1952)]{hestenes1952MethodsConjugateGradients}
M.~Hestenes and E.~Stiefel.
\newblock Methods of conjugate gradients for solving linear systems.
\newblock \emph{Journal of Research of the National Bureau of Standards},
  49:\penalty0 409--436, 1952.

\bibitem[Hotelling(1936)]{hotelling1936RelationsTwoSets}
Harold Hotelling.
\newblock Relations between two sets of variates.
\newblock \emph{Biometrika}, 28\penalty0 (3/4):\penalty0 321--377, 1936.

\bibitem[Johnson and Zhang(2013)]{johnson2013AcceleratingStochasticGradient}
Rie Johnson and Tong Zhang.
\newblock Accelerating stochastic gradient descent using predictive variance
  reduction.
\newblock In C.J. Burges, L.~Bottou, M.~Welling, Z.~Ghahramani, and K.Q.
  Weinberger, editors, \emph{Advances in Neural Information Processing
  Systems}, volume~26. Curran Associates, Inc., 2013.
\newblock URL
  \url{https://proceedings.neurips.cc/paper/2013/file/ac1dd209cbcc5e5d1c6e28598e8cbbe8-Paper.pdf}.

\bibitem[Jordan et~al.(2019)Jordan, Lee, and
  Yang]{michaeli.jordan2019communication}
Michael~I. Jordan, Jason~D. Lee, and Yun Yang.
\newblock Communication-efficient distributed statistical inference.
\newblock \emph{Journal of the American Statistical Association}, 114\penalty0
  (526):\penalty0 668--681, 2019.

\bibitem[Lecun et~al.(1998)Lecun, Bottou, Bengio, and
  Haffner]{lecun1998GradientbasedLearningApplied}
Y.~Lecun, L.~Bottou, Y.~Bengio, and P.~Haffner.
\newblock Gradient-based learning applied to document recognition.
\newblock \emph{Proceedings of the IEEE}, 86\penalty0 (11):\penalty0
  2278--2324, 1998.

\bibitem[Lv et~al.(2020)Lv, He, Huang, Yang, and
  Chen]{lv2020OneshotDistributedAlgorithm}
Kexin Lv, Fan He, Xiaolin Huang, Jie Yang, and Liming Chen.
\newblock One-shot distributed algorithm for generalized eigenvalue problem.
\newblock \emph{arXiv:2010.11625}, October 2020.

\bibitem[Ma and Li(2020)]{ma2020SubspacePerspectiveCanonical}
Zhuang Ma and Xiaodong Li.
\newblock Subspace perspective on canonical correlation analysis: Dimension
  reduction and minimax rates.
\newblock \emph{Bernoulli}, 26\penalty0 (1):\penalty0 432--470, February 2020.

\bibitem[Ma et~al.(2015)Ma, Lu, and Foster]{ma2015FindingLinearStructure}
Zhuang Ma, Yichao Lu, and Dean Foster.
\newblock Finding linear structure in large datasets with scalable canonical
  correlation analysis.
\newblock In Francis Bach and David Blei, editors, \emph{Proceedings of the
  32nd International Conference on Machine Learning}, volume~37 of
  \emph{Proceedings of Machine Learning Research}, pages 169--178, Lille,
  France, 07--09 Jul 2015. PMLR.
\newblock URL \url{https://proceedings.mlr.press/v37/maa15.html}.

\bibitem[Shamir et~al.(2014)Shamir, Srebro, and Zhang]{shamir2014communication}
Ohad Shamir, Nati Srebro, and Tong Zhang.
\newblock Communication-efficient distributed optimization using an approximate
  newton-type method.
\newblock In Eric~P. Xing and Tony Jebara, editors, \emph{Proceedings of the
  31st International Conference on Machine Learning}, volume~32 of
  \emph{Proceedings of Machine Learning Research}, pages 1000--1008, Bejing,
  China, 22--24 Jun 2014. PMLR.
\newblock URL \url{https://proceedings.mlr.press/v32/shamir14.html}.

\bibitem[Snoek et~al.(2006)Snoek, Worring, {van Gemert}, Geusebroek, and
  Smeulders]{snoek2006ChallengeProblemAutomated}
Cees G.~M. Snoek, Marcel Worring, Jan~C. {van Gemert}, Jan-Mark Geusebroek, and
  Arnold W.~M. Smeulders.
\newblock The challenge problem for automated detection of 101 semantic
  concepts in multimedia.
\newblock In \emph{Proceedings of the 14th {{ACM}} International Conference on
  {{Multimedia}}}, {{MM}} '06, pages 421--430, {New York, NY, USA}, October
  2006. {Association for Computing Machinery}.

\bibitem[Wainwright(2019)]{wainwright2019HighDimensionalStatisticsNonAsymptotic}
Martin~J. Wainwright.
\newblock \emph{High-{{Dimensional Statistics}}: {{A Non-Asymptotic
  Viewpoint}}}.
\newblock {Cambridge University Press}, first edition, February 2019.

\bibitem[Wang et~al.(2016)Wang, Wang, Garber, Garber, and
  Srebro]{wang2016EfficientGloballyConvergent}
Weiran Wang, Jialei Wang, Dan Garber, Dan Garber, and Nati Srebro.
\newblock Efficient globally convergent stochastic optimization for canonical
  correlation analysis.
\newblock In D.~Lee, M.~Sugiyama, U.~Luxburg, I.~Guyon, and R.~Garnett,
  editors, \emph{Advances in Neural Information Processing Systems}, volume~29.
  Curran Associates, Inc., 2016.
\newblock URL
  \url{https://proceedings.neurips.cc/paper/2016/file/42998cf32d552343bc8e460416382dca-Paper.pdf}.

\bibitem[Yu et~al.(2015)Yu, Wang, and Samworth]{yu2015UsefulVariantDavis}
Y.~Yu, T.~Wang, and R.~J. Samworth.
\newblock A useful variant of the {{Davis}}\textendash{{Kahan}} theorem for
  statisticians.
\newblock \emph{Biometrika}, 102\penalty0 (2):\penalty0 315--323, June 2015.

\bibitem[Zaharia et~al.(2016)Zaharia, Xin, Wendell, Das, Armbrust, Dave, Meng,
  Rosen, Venkataraman, Franklin, Ghodsi, Gonzalez, Shenker, and
  Stoica]{zaharia2016ApacheSparkUnified}
Matei Zaharia, Reynold~S. Xin, Patrick Wendell, Tathagata Das, Michael
  Armbrust, Ankur Dave, Xiangrui Meng, Josh Rosen, Shivaram Venkataraman,
  Michael~J. Franklin, Ali Ghodsi, Joseph Gonzalez, Scott Shenker, and Ion
  Stoica.
\newblock Apache spark: A unified engine for big data processing.
\newblock \emph{Communications of the ACM}, 59\penalty0 (11):\penalty0 56--65,
  October 2016.

\bibitem[Zhao et~al.(2015)Zhao, Kolar, and Liu]{zhao2015GeneralFrameworkRobust}
Tianqi Zhao, Mladen Kolar, and Han Liu.
\newblock A general framework for robust testing and confidence regions in
  high-dimensional quantile regression.
\newblock \emph{arXiv:1412.8724}, March 2015.

\end{thebibliography}
